\documentclass[acmtocl]{acmtrans2m}

\usepackage{xspace}
\usepackage{graphicx}
\usepackage[all]{xy}

\usepackage{amsmath}
\usepackage{amssymb}
\usepackage{enumerate}
\usepackage{url}

\newcommand{\st}{\ensuremath{.}}
\newcommand{\mytimes}{{\times}}

\newcommand{\allsymbol}{\ensuremath{\mathrm{A}}\xspace}
\newcommand{\anysymbol}{\ensuremath{\mathrm{E}}\xspace}

\newcommand{\xmodels}{\ensuremath{\models^\mathrm{X}}\xspace}
\newcommand{\allmodels}{\ensuremath{\models^\allsymbol}\xspace}
\newcommand{\anymodels}{\ensuremath{\models^\anysymbol}\xspace}

\newcommand{\xstable}{\ensuremath{\text{stable}^\mathrm{X}}\xspace}
\newcommand{\anystable}{\ensuremath{\text{stable}^\anysymbol}\xspace}
\newcommand{\allstable}{\ensuremath{\text{stable}^\allsymbol}\xspace}

\newcommand{\expequiv}{\ensuremath{\equiv}\xspace}
\newcommand{\xexpequiv}{\ensuremath{\equiv^\mathrm{X}}\xspace}
\newcommand{\anyexpequiv}{\ensuremath{\equiv^\anysymbol}}

\newcommand{\nmodels}{\ensuremath{\not\models}\xspace}

\newcommand{\nop}[1]{}

\newcommand{\mop}{\ensuremath{\mathsf{L}}\xspace}

\newcommand{\mopk}{\ensuremath{\mathsf{K}}\xspace}
\newcommand{\mopnot}{\ensuremath{\mathsf{not}}\xspace}

\newcommand{\limpl}{\ensuremath{\supset}}
\newcommand{\UNA}{\ensuremath{\mathrm{UNA}}\xspace}
\newcommand{\oga}{\ensuremath{\mathit{oga}}}
\newcommand{\ga}{\ensuremath{\mathit{ga}}}
\newcommand{\og}{\ensuremath{\mathit{og}}}
\newcommand{\PIA}{\ensuremath{\mathrm{PIA}}\xspace}

\newcommand{\lpor}{\mid}

\newcommand{\dnot}{\ensuremath{not\text{ }}}

\newcommand{\snot}{\ensuremath{\sim}\xspace}

\newcommand{\ainter}{\ensuremath{\langle\inter,T\rangle}\xspace}
\newcommand{\inter}{\ensuremath{w}\xspace}
\newcommand{\domain}{\ensuremath{U}\xspace}
\newcommand{\interfun}{\ensuremath{{\cdot^I}}\xspace}
\newcommand{\interfunj}{\ensuremath{{\cdot^{I_j}}}\xspace}
\newcommand{\funsym}{\ensuremath{I}}
\newcommand{\interdef}{\ensuremath{\inter = \langle \domain, \interfun
  \rangle}\xspace}

\newcommand{\lang}{\ensuremath{\mathcal{L}}\xspace}

\newcommand{\flang}{\ensuremath{\mathcal{L}}\xspace}

\newcommand{\signature}{\ensuremath{\Sigma}\xspace}

\newcommand{\fsymb}{\ensuremath{\mathcal{F}}\xspace}
\newcommand{\psymb}{\ensuremath{\mathcal{P}}\xspace}
\newcommand{\vsymb}{\ensuremath{\mathcal{V}}\xspace}

\newcommand{\gr}[1]{\ensuremath{gr(#1)}\xspace}

\newcommand{\lp}{\ensuremath{P}\xspace}
\newcommand{\fot}{\ensuremath{\Phi}\xspace}

\newcommand{\fmodl}{\ensuremath{\flang_{\mop}}\xspace}

\newcommand{\embop}{\ensuremath{\tau}\xspace}
\newcommand{\embvar}{\ensuremath{\chi}}
\newcommand{\emb}{\ensuremath{\embop_\embvar}\xspace}
\newcommand{\embhp}{\ensuremath{\embop_{\mathit{HP}}}\xspace}
\newcommand{\embeb}{\ensuremath{\embop_{\mathit{EB}}}\xspace}
\newcommand{\embeh}{\ensuremath{\embop_{\mathit{EH}}}\xspace}

\newcommand{\embdis}{\ensuremath{\embop_\embvar^\vee}\xspace}
\newcommand{\embgen}{\ensuremath{\embop_\embvar^{(\vee)}}\xspace}

\newcommand{\embhpdis}{\ensuremath{\embop_{HP}^\vee}\xspace}
\newcommand{\embebdis}{\ensuremath{\embop_{EB}^\vee}\xspace}
\newcommand{\embehdis}{\ensuremath{\embop_{EH}^\vee}\xspace}

\newcommand{\combop}{\ensuremath{\iota}\xspace}
\newcommand{\combvar}{\embvar\xspace}
\newcommand{\comb}{\ensuremath{\combop_\combvar}\xspace}
\newcommand{\combhp}{\ensuremath{\combop_{\mathit{HP}}}\xspace}
\newcommand{\combeb}{\ensuremath{\combop_{\mathit{EB}}}\xspace}
\newcommand{\combeh}{\ensuremath{\combop_{\mathit{EH}}}\xspace}

\newcommand{\combhpdis}{\ensuremath{\combop_{HP}^\vee}\xspace}
\newcommand{\combebdis}{\ensuremath{\combop_{EB}^\vee}\xspace}
\newcommand{\combehdis}{\ensuremath{\combop_{EH}^\vee}\xspace}
\newcommand{\combpar}{\ensuremath{(\fot,\lp)}\xspace}

\newcommand{\varsub}{\ensuremath{\beta}\xspace}

\newcommand{\varass}{\ensuremath{B}\xspace}

\newcommand{\names}{\ensuremath{\mathcal{N}}\xspace}

\newcommand{\shoin}{\ensuremath{\mathcal{SHOIN}}\xspace}
\newcommand{\shiq}{\ensuremath{\mathcal{SHIQ}}\xspace}
\newcommand{\sroiq}{\ensuremath{\mathcal{SROIQ}}\xspace}

\newcommand{\arbclassprg}{\ensuremath{\mathcal{LP}}\xspace}
\newcommand{\safeclass}{\ensuremath{\mathit{s}}\arbclassprg}
\newcommand{\groundclass}{\ensuremath{\mathit{g}}\arbclassprg}
\newcommand{\narbclassprg}{\ensuremath{\mathit{n}\mathcal{LP}}\xspace}
\newcommand{\nsafeclass}{\ensuremath{\mathit{sn}}\arbclassprg}
\newcommand{\ngroundclass}{\ensuremath{\mathit{gn}}\arbclassprg}

\newcommand{\arbclassth}{\ensuremath{\mathcal{F}\!\mathit{o}\hspace{-.2ex}\mathcal{L}}\xspace}
\newcommand{\uniclass}{\ensuremath{\mathcal{U}\hspace{-.2ex}\mathit{ni}}\xspace}
\newcommand{\ghornclass}{\ensuremath{\mathit{g}\mathcal{H}\hspace{-.2ex}\mathit{orn}}\xspace}
\newcommand{\hornclass}{\ensuremath{\mathcal{H}\hspace{-.2ex}\mathit{orn}}\xspace}
\newcommand{\propclass}{\ensuremath{\mathcal{P}\hspace{-.2ex}\mathit{rop}}\xspace}
\newcommand{\emptyclass}{\{\emptyset\}\xspace}

\newcommand{\dllog}{{\ensuremath{\mathcal{DL}{+}\mathit{log}}}\xspace}
\newcommand{\allog}{{\ensuremath{\mathcal{AL}\mbox{-}\mathit{log}}}\xspace}

\newcommand{\triple}[3]{\ensuremath{triple(\text{#1},} \ensuremath{\text{#2},} \ensuremath{\text{#3})}\xspace}

\newtheorem{theorem}{Theorem}[section]

\newtheorem{corollary}[theorem]{Corollary}
\newtheorem{proposition}[theorem]{Proposition}
\newtheorem{lemma}[theorem]{Lemma}
\newdef{definition}[theorem]{Definition}
\newdef{remark}[theorem]{Remark}
\newdef{example}[theorem]{Example}

\title{Embedding Non-Ground Logic Programs into Autoepistemic Logic
for Knowledge Base Combination%
}

\begin{abstract}
  In the context of the Semantic Web, several approaches to the
  combination of ontologies, given in terms of theories of classical
  first-order logic and rule bases, have been proposed. They either
  cast rules into classical logic or limit the interaction between
  rules and ontologies. Autoepistemic logic (AEL) is an attractive
  formalism which allows to overcome these limitations, by serving as
  a uniform host language to embed ontologies and nonmonotonic logic
  programs into it. For the latter, so far only the propositional
  setting has been considered.  In this paper, we present three
  embeddings of normal and three embeddings of disjunctive non-ground
  logic programs under the stable model semantics into first-order
  AEL.  While the embeddings all correspond with respect to objective
  ground atoms, differences arise when considering non-atomic formulas
  and combinations with first-order theories.  We compare the
  embeddings with respect to stable expansions and autoepistemic
  consequences, considering the embeddings by themselves, as well as
  combinations with classical theories.  Our results reveal
  differences and correspondences of the embeddings and provide useful
  guidance in the choice of a particular embedding for knowledge
  combination.
\end{abstract}





\author{JOS DE BRUIJN\\
  Free University of Bozen-Bolzano
  \and
  THOMAS EITER and HANS TOMPITS\\
  Technische Universit\"{a}t Wien \and
  AXEL POLLERES\\
  Digital Enterprise Research Institute, National University of Ireland,
  Galway }

\begin{document}

\begin{bottomstuff}
%An extended abstract
A preliminary version of this article appeared in: \emph{Proc.\
% Proceeding of the
20th International Joint Conference on Artificial Intelligence
{\rm (}IJCAI 2007{\rm )}}, AAAI Press, pp.\ 304--309.

Author's address: J.\ de Bruijn, Faculty of Computer Science, Free University
of Bozen-Bolzano, Piazza Domenicani 3, 39100 Bolzano, Italy; email: {\tt debruijn@inf.unibz.it}; T.\ Eiter and H.\ Tompits, Institut f\"ur
Informationssysteme 184/3, Technische Universit\"at Wien, Favoritenstra{\ss}e
9-11, A-1040 Vienna, Austria; email: {\tt
  \{eiter,tompits\}@kr.tuwien.ac.at}; A.\ Polleres, Digital Enterprise
Research Institute, National University of Ireland, Galway, IDA Business Park,
Lower Dangan Galway, Ireland; email: {\tt axel.polleres@deri.org}.

This work was partially funded by the European Commission projects Knowledge
Web (IST 507482), DIP (IST 507483), REWERSE (IST 506779), and ONTORULE (IST
231875), by the Austrian Science Fund~(FWF) projects P17212-N04 and P21840, by
the CICyT of Spain project TIC-2003-9001, and by the Science Foundation
Ireland under Grant No.~SFI/02/CE1/I131.
\end{bottomstuff}
\maketitle

\section{Introduction}

In the last years, significant effort has been devoted to bring
the vision of a Semantic Web closer to reality. Adopting a
layered architecture, a number of building blocks have been
proposed that serve different purposes, from low-level data
encoding to high-level semantic representation. In this
architecture, the building blocks for ontologies, rules, and query
languages play a prominent role. Furthermore, to ensure
interoperability and wide applicability, standard
representation formalisms are propagated by the World Wide Web
Consortium (W3C), including the Resource Description Framework
(RDF) \cite{Klyne+Carroll-ResoDescFram04,Hayes-Sema04}, the Web
Ontology Language (OWL)
\cite{Patel-Schneider+HayesETAL-OntoLangSemaAbst03,Motik+Patel-SchneiderETAL-OntoLangSpecFunc:09},
and the recent Rule Interchange Format Basic Logic Dialect (RIF
BLD)~\cite{Boley+Kifer-BasiLogiDial07}. In addition, the RIF
logical framework \cite{Kifer-rif08} lays the foundation for
Web rule languages extending RIF BLD with nonmonotonic
negation.

Each of these formalisms has a formal semantics, which is
either expressible in terms of classical logic or logic
programming
\cite{Bruijn-Heymans-LogiFoun-07,Horrocks+Patel-Schneiter-Redu03,Kifer-rif08}.
%However, an integration of these formalisms that puts the respective building blocks
%satisfactorily together has not yet been achieved, though
%several (partial) proposals do exist.
There is a need for combining these formalisms,
which is illustrated by the following simple example.

\begin{example}
  \label{ex:sw-intro} RDF is the basic data description language of
the Semantic Web. Atomic statements (\emph{triples}) in RDF are
of the form ({subject}, {predicate}, {object}).  Subjects and
objects, which may be resources on the Web (denoted by URIs),
may be  shared between triples, yielding a graph-based data
model. As  demonstrated by \citeN{Bruijn-Heymans-LogiFoun-07},
the semantics of  RDF \cite{Hayes-Sema04} can be captured in
first-order logic, in  terms of a formula where sets of such
triples are expressed as  conjunctions of facts. For purposes
of illustration and compatibility with description logics (DL),
we use here a simplified  notation for RDF triples where class
membership (rdf:type)  statements  are represented using unary
predicates, and all other statements  using  binary predicates.
For a full encoding of RDF in  first-order logic,
\citeN{Bruijn-Heymans-LogiFoun-07} use a ternary predicate
$triple$ to represent all RDF statements (see also
Section~\ref{sec:rdf}).

Consider a fictitious Web site \texttt{gangsterepics.com} that
publishes information about gangster movies:

{\small
\begin{tabbing}
\qquad $\exists x$.\=$title(\textrm{TheGodFather},\textrm{``The Godfather''}) \wedge$ \\
\> $title(\textrm{TheGodFather2},\textrm{``The Godfather: Part II''}) \wedge$\\
\> $title(\textrm{PulpFiction},\textrm{``Pulp Fiction''}) \wedge$ \=
 $director(\textrm{PulpFiction},\textrm{quentinTarantino}) \wedge$\\
\> $director(\textrm{TheGodFather},\textrm{$x$}) \wedge$  \>
 $director(\textrm{TheGodFather2},\textrm{$x$}) \wedge$\\
\> $mentionedAt(\textrm{PulpFiction},\textrm{{gangsterepics.com}})\wedge$\\
\> $mentionedAt(\textrm{TheGodFather},\textrm{{gangsterepics.com}})\wedge$\\
\> $mentionedAt(\textrm{TheGodFather2},\textrm{{gangsterepics.com}})$
\end{tabbing}
}
The creator of the page did not know the director of ``The
Godfather'' movies, but wanted to express the fact that both
parts had the same director.  To this end, he used an
existentially quantified variable, called ``blank node'' in
RDF. RDF Schema~\citeyear{rdfs} and OWL have the possibility to
express structural relations between predicates. For instance,
an OWL (DL)
% \comment{JB: I'm somewhat uneasy about using OWL
% Full here, but it is not possible to use OWL DL in combination
% with the triple predicate.AP:Do the unary/binary preds now
% remedy this?changed to DL}
ontology stating that all the movies
mentioned at \texttt{gangsterepics.com} are either epics or
gangster movies, and that $\mathit{director}$ is the inverse of
$\mathit{directorOf}$
% \comment{AP:if you agree with my change to
% unary/binary, then we should re-introduce the DL axioms,
% otherwise the mention of OWL (DL) doesn't make sense... JB:
% looks good}
can be expressed in terms of DL axioms
\[
\begin{array}{rl}
 \exists \mathit{mentionedAt}.\{\mbox{gangsterepics.com}\} & \sqsubseteq Epic \sqcup \mathit{GangsterMovie},\\
 \mathit{director}^- &\equiv \mathit{directorOf},
\end{array}
\]
which may be viewed as a set of first-order logic
sentences
% \comment{AP:uncommented that, is it now ok with the
% unary/binary encoding? Do we need a reference? Borghida (don't
% remember at the moment what was the right reference?)? JB: I
% think the encoding should be fine. Not sure we need a
% reference; this encoding of DLs in FOL is pretty well-known.}
\[
\begin{array}{r@{\,}l}
 \forall x.   \mathit{mentionedAt}(x,\textrm{{gangsterepics.com}}) \limpl &(\mathit{Epic}({x}) \vee \mathit{GangsterMovie}({x})),  \\
 \forall x,y.  \mathit{director}({x},{y}) \equiv & \mathit{directorOf}({y},{x}).
\end{array}
\]
Apart from classical first-order statements, it may be useful
to express nonmonotonic information, e.g., that any gangster
movie not mentioned on \url{gangsterepics.com} is an
independent movie.
% \comment{JB: actually, this is not a
% reasonable statement, since IMDB also includes information
% about independent films.\par HT: replaced IMDB by
% gangsterepics.com. Perhaps this makes more sense.I think that
% works. JB: looks good}
For such nonmonotonic statements, logic
programming based rules languages that include negation are
better suited. That is, we may use the following rule:
\[
\begin{array}{ll}
\mathit{IndieMovie}({x}) \leftarrow \mathit{GangsterMovie}({x}), \dnot{{\mathit{mentionedAt}}({x},\textrm{{gangsterepics.com}})}.
\end{array}
\]
%Likewise, query languages for the Semantic Web, such as SPARQL
%have a logic programming based semantics,\comment{JB: not true.
%There may exist an encoding of the language in LP, but it is
%certainly not based on LP.} based on Datalog
%%~\cite{ullm-89,eite-etal-97f},
%as shown in \cite{angl-guti-2008,poll-2007}. For instance, the SPARQL query \texttt{SELECT ?X WHERE {?Y $directorOf$ ?X}},
%which asks for all directors, may be viewed as a Datalog rule:
Consider now the following query, which asks for all directors, using the
$\mathit{directorOf}$ predicate, written using Datalog notation:
$$
\mathit{answer}(x)\leftarrow \mathit{directorOf}(y,x).
$$
Answering such a query essentially depends on how to interpret
and formally combine data (RDF), ontologies (OWL), and rule
bases. Given that each of these parts is expressible as either
a classical first-order theory or a logic program, the question
is how to combine logic programs and classical first-order
theories in a unifying logical framework.
\end{example}

A combination of this kind is not obvious, due to the very different setting
of classical logic and logic programming
\cite{bruijn06-repres-issues-about-combin-of,MotikHRS06}, and many proposals
for combination have been made (we review several of them in
Sections~\ref{sec:SW-application} and \ref{sec:related-work}). Like in the
previous example, an ontology in the form of a classical theory\footnote{While
  ontologies are not always first-order definable, for the purpose of this
  paper we confine ourselves to such ontologies.  We note that the Semantic
  Web ontology languages are first-order definable
  \cite{Sattler+CalvaneseETAL-RelawithOtheForm:03}.} and a logic program
should be viewed as complementary descriptions of the same domain. Therefore,
a separation between predicates defined in these two components should not be
enforced. Furthermore, it is desirable to neither restrict the interaction
between the classical and the rules components nor to impose syntactic or
semantic restrictions on the individual components. That is, the classical
component may be an arbitrary theory $\Phi$ of some first-order language with
equality, and the rules component may be an arbitrary non-ground normal or
disjunctive logic program $P$, interpreted using, e.g., the common stable
model
semantics~\cite{Gelfond+Lifschitz-StabModeSemaLogi:88,Gelfond+Lifschitz-ClasNegaLogiProg:91}.%
\footnote{For computational reasons, such restrictions (e.g., DL-safety) may
  be imposed, cf.\ Sections~\ref{sec:SW-application} and
  \ref{sec:related-work}.} The goal is a combined theory, $\combop(\Phi,P)$,
in a uniform logical formalism. Naturally, this theory should amount to $\Phi$
if $P$ is empty, and to $P$ if $\Phi$ is empty. Therefore, such a combination
must provide {\em faithful embeddings} $\sigma(\Phi)$ and $\tau(P)$ of $\Phi$
and $P$, respectively, such that $\sigma(\Phi) = \combop(\Phi,\emptyset)$ and
$\tau(P) = \combop(\emptyset,P)$. In turn, knowledge combination may be
carried out on top of $\sigma(\cdot)$ and $\tau(\cdot)$, where in the simplest
case one may choose $\combop(\Phi,P) = \sigma(\Phi)\cup \tau(P)$.

This raises the following questions: (a)~which uniform
formalism is suitable and (b)~which embeddings are suitable
and, furthermore, how do potentially suitable embeddings relate
to each other and behave under knowledge combination?

Concerning the first question,
\citeN{Motik+Rosati-FaitInteDescLogi:07} use a variant of
Lifschitz's bimodal nonmonotonic \emph{logic of minimal
knowledge and negation-as-failure}\/ (MKNF)
\cite{Lifschitz-NonmDataEpisQuer:91}. While the proposed
embeddings of the first-order (FO) theory and the logic program
are both faithful in the sense described above, the particular
combination proposed by Motik and Rosati is only one among many
possible methods and MKNF is only one possible underlying
formalism for such combinations (we discuss these issues in
more detail in Section~\ref{sec:discussion}). Indeed,
\citeN{Bruijn-PearceETAL-QuanEquiLogiHybr-07} use quantified
equilibrium logic (QEL) \cite{PV05}  as a host formalism.
Unlike Motik and Rosati, de Bruijn et al.~do not propose a new
semantics for combinations, but rather show that QEL can
capture the semantics of combinations by \citeN{rosati06} and
can be used, for example, to define notions of equivalence of
combinations.

\emph{Autoepistemic logic} (AEL) \cite{moore85:_seman}, which
extends classical logic with a single nonmonotonic modal belief
operator, being essentially the nonmonotonic variant of the
modal logic $ \mathbf{kd45} $
\cite{shvarts90:_autoep_modal_logic,Marek+Truszczynski-NonmLogi:93},
is an attractive candidate for serving as a uniform host
formalism for combinations. Compared to other well-known
nonmonotonic formalisms, like Reiter's default logic
\cite{reit-80}, FO-AEL offers a uniform language in which
(nonmonotonic) rules themselves can be expressed at the object
level. This conforms with the idea of treating an ontology and
a logic program together as a unified theory.  Furthermore, in
FO-AEL we can decide, depending on the context, whether (the
negation of) a particular atomic formula should be interpreted
nonmonotonically simply by including a modal operator.  This
enables us to use the same predicate in both a monotonic and a
nonmonotonic context.  This is in contrast to circumscription
\cite{McCarthy-ApplCircFormComm:86}, in which one has to
decide, for the entire theory, which predicates are to be
minimized.

Embedding a classical theory in AEL is trivial, and several
embeddings of logic programs in AEL have been described
\cite{Gelfond+Lifschitz-StabModeSemaLogi:88,marek93:_reflex_autoep_logic_logic_progr,lifschitz93:_exten,chen93:_minim,przymusinski91:_stabl_seman_disjun_progr}.
However, they all have been developed for the propositional
case only, whereas we need to deal with non-ground theories and
programs. This requires us to consider \emph{first-order
autoepistemic logic} (FO-AEL)
\cite{konolige91:_quant,kaminski02:_revis_quant_autoep_logic,levesque00:_logic_of_knowl_bases},
and non-ground versions of these embeddings. We consider the
semantics for FO-AEL as defined by \citeN{konolige91:_quant},
because it faithfully extends first-order logic with equality
(other variants are discussed in Section~\ref{sec:discussion}).

\medskip
\noindent Motivated by these issues, our contribution in this
paper is twofold:
\begin{longenum}
\item
We define several embeddings of non-ground logic programs into FO-AEL,
taking into account subtle issues of quantification in FO-AEL. In more
detail, we present three embeddings, $\embhp$, $\embeb$, and $\embeh$,
for normal logic programs which extend respective embeddings for the
propositional case
\cite{gelfond87:_strat_autoep_theor,Gelfond+Lifschitz-StabModeSemaLogi:88,marek93:_reflex_autoep_logic_logic_progr,chen93:_minim,Lifschitz-MiniBeliNegaFail:94},
and three embeddings, $\embhpdis$, $\embebdis$, and $\embehdis$, for
disjunctive logic programs, where $\embhpdis$ and $\embehdis$ extend
embeddings considered in the ground case
\cite{przymusinski91:_stabl_seman_disjun_progr,marek93:_reflex_autoep_logic_logic_progr}.
We show that all these embeddings are faithful in the sense that the
stable models of the logic program $\lp$ and the sets of objective
ground atoms in the stable expansions of the embeddings $\emb(\lp)$
($\embvar\in\{\mathit{HP},\mathit{EB},\mathit{EH}\}$) are in a one-to-one correspondence
(Theorem~\ref{cor:corr-ground-atomic}).  However, the embeddings
behave differently on formulas beyond ground atoms, in some cases
already for simple ground formulas. This, in turn, may impact the
behavior of the embeddings when used in combinations of logic
programs and classical theories.
%, which may then differ already on ground atoms.
This raises the question under which conditions the embeddings differ
and under which conditions they correspond.
Of particular interest for knowledge combination is
how these embeddings behave relatively to each other
in combinations with classical theories.

\item To answer these questions, we conduct two comparative studies
of the behavior of the various embeddings.  We consider three classes of
programs: ground, safe, and arbitrary
logic programs under the stable model semantics.

\begin{longenum}
\item We first determine correspondences between the stable expansions
  of different embeddings $\emb$ beyond ground atomic formulas
  (Propositions~\ref{prop:corr-ground-formulas}-\ref{prop:corr-ground-programs}),
  and present inclusion relations between the sets of consequences of
  the embeddings (Theorems~\ref{th:cons-og} and \ref{th:cons-oga}). These results already allow
  to draw a few conclusions on the behavior of embeddings in
  combinations.

\item We then determine correspondences between stable expansions for
  combinations of logic programs with classical theories. Here,
  we take the shape of
    the logic program, the shape of the classical theory, and the type
    of formulas of interest for the correspondence into
    account. To this end, we consider
  different fragments of classical logic that are important for
  knowledge representation, including Horn, universal, and generalized
  Horn theories.  The latter are of particular interest for
  ontologies, since they essentially include RDF Schema
  \cite{Bruijn-Heymans-LogiFoun-07}, Horn-\shiq
  \cite{Hustadt05datacomplexity}, and the OWL~2
  profiles QL, RL, and EL \cite{Motik+GrauETAL-OntoLang:08};
  furthermore, they essentially include also Tuple Generating
  Dependencies \cite{abiteboul}, which are a popular class of constraints in
  databases. Our main result for embeddings in combinations
  (Theorem~\ref{th:fo-ael-exp-correspondence}) gives a complete
  picture of the correspondences, which reveals that they behave
  differently in general, and shows the restrictions on the program or
  theory that give rise to correspondence.
\end{longenum}
\end{longenum}

The results of these studies not only deepen the understanding
of the individual embeddings, but also have practical
implications with respect to their use. They tell us in which
situations one embedding may be used instead of another.

Noticeably, the embeddings of logic programs we study
%in this paper
can be seen as building blocks for actual combinations of a classical
theory $\fot$ and a logic program $\lp$. The most straightforward
combination
%of a classical theory \fot and a logic program \lp
is $ \combop(\fot, \lp)=\sigma(\fot)\cup \emb(\lp)$, where $
\sigma $ is the identity mapping and \emb is one of the
embeddings
%studied in this paper.
we consider. One could also imagine adding axioms to, or
changing axioms in \fot; similarly, rules could be changed in,
or added to \lp before translating them (e.g., grounding rules
as customary in logic programming). If $ \fot' $ and $ \lp' $
are the thus obtained classical theory and logic program, our
results are still applicable to the combination $
\combop'(\fot, \lp)=\fot'\cup \emb(\lp')$. In fact, whenever
the combination is of the form $ \fot' \cup \emb(\lp')$,
regardless of \fot and \lp, the correspondences and differences
between the embeddings we establish hold. Furthermore, the
effect of different program rewritings $P'$ in combinations may
be assessed.

To illustrate the use of our results, we show applications to
the Semantic Web. More specifically, we show that the semantics
of existing combinations of ontologies and rules in this
context can be captured, and that via our correspondence result
properties of the semantics can be derived, as well as their
behavior in other (modified) combinations. Finally we show how
the embeddings we consider can be used to extend combinations
to richer languages, particularly extensions of rule languages
with nonmonotonic negation. However, while we focus here on the
Semantic Web, applications in other contexts (e.g., data
modeling languages like UML plus OCL) might be explored.

%\comment{TE: Say that why motivated by SW, there are other
%areas with these and similar combinations are desired
%(e.g. OCL?). Could be also go to the part where applications are mentioned.}

\medskip

\noindent The remainder of the paper is structured as follows.
We review the definitions of first-order logic and logic
programs in Section \ref{sec:preliminaries}.  We proceed to
describe first-order autoepistemic logic (FO-AEL) and present a
novel characterization of stable expansions for certain kinds
of theories in Section \ref{sec:first-order-auto}.  The
embeddings of normal and disjunctive logic programs and our
results about faithfulness of the embeddings are described in
Section \ref{sec:embedd-logic-progr}.  We investigate the
relationships between the embeddings themselves, and under
combination with first-order theories, in Sections
\ref{sec:prop-embedd} and \ref{sec:prop-comb}. We discuss applications to the
Semantic Web in Section~\ref{sec:SW-application} and further implications in
Section~\ref{sec:discussion}. We discuss related work in
Section~\ref{sec:related-work}, and conclude and outline future
work in Section \ref{sec:conclusions}.
%
%\noindent
Proofs of the results in Sections \ref{sec:prop-embedd}
and \ref{sec:prop-comb} can be found in the appendix.

\section{Preliminaries}
\label{sec:preliminaries}

Let us briefly recapitulate some basic elements of first-order
logic and logic programs as well as some relevant notation.

\subsection{First-Order Logic}
\label{sec:preliminaries-fol}

We consider first-order logic with equality.  A language $\flang$ is defined over a
signature $\signature=(\fsymb,\psymb)$, where \fsymb and \psymb are
countable sets of \emph{function} and \emph{predicate symbols},
respectively. Function symbols with arity $0$ are also called
\emph{constants}. Furthermore, \vsymb is a countably infinite set of
\emph{variables}. Terms and atomic formulas (atoms) are
constructed as usual. Ground
terms are also called \emph{names}; $\names_\signature$ denotes the set
of names of a given signature $\signature$. Complex formulas are
constructed as usual using the primitive symbols $\lnot$, $\land$,
 $\exists$,  `$($', and `$)$'.
As usual, $ \phi \lor \psi $ is short for $ \lnot(\lnot \phi \land \lnot \psi)$, $ \phi
\limpl \psi $ is short for $ \lnot \phi \lor \psi$, and $ \forall x \st \phi(x)$ is short
for $ \lnot \exists x \st \lnot \phi(x) $.
We sometimes write
$ t_1\neq t_2 $, where
$ t_1$ and $ t_2 $ are terms, as an abbreviation for $ \lnot (t_1 =t_2) $. The universal closure of a formula
$\phi$ is denoted by $(\forall) \;\phi$. $\flang_g$ is the restriction of
\flang to ground formulas and $\flang_{ga}$ is the restriction of
$\flang_{g}$ to atomic formulas. An \emph{FO theory} $\fot
\subseteq \flang$ is a set of
closed formulas,
i.e., every variable is bound by a quantifier.

An \emph{interpretation} of a language \flang~is a tuple \interdef,
where $U$ is a nonempty set, called the \emph{domain}, and \interfun
is a mapping which assigns to
every $n$-ary function symbol $f \in \fsymb$  a function $f^I:U^n \rightarrow U$ and to every $n$-ary predicate symbol $p \in \psymb$ a relation $p^I
\subseteq U^n$.  A
\emph{variable assignment} $B$ for \inter is a mapping that assigns to every
variable $x \in \vsymb$
an element $x^B \in U$.  A
  variable assignment $B'$ is an \emph{$x$-variant} of $B$ if $y^B =
  y^{B'}$ for every variable $y \in \vsymb$ such that $y \neq x$.
 The interpretation of a term $t$, denoted $t^{w,B}$, is defined as
usual; if $t$ is ground, we sometimes write $t^w$.

We call an individual $k$ \emph{named} if there is some name $t
\in \names$ such that $t^w = k$, and \emph{unnamed} otherwise.
Interpretations are \emph{named} if all individuals are named.
We say that the \emph{unique names assumption} applies to an
interpretation if all names are interpreted distinctly, and we
say that the \emph{standard names assumption} applies if, in
addition,
the interpretation is named.%
\footnote{We note here that the term ``standard names assumption'' is used with
various slightly different meanings in the literature; see
Section~\ref{sec:discussion} for further discussion.}
%  \citeN{levesque90:_all_i}
%and \citeN{Motik+Rosati-FaitInteDescLogi:07}
%require a countably infinite set of constant symbols, and do not consider
%function symbols;  \citeN{levesque00:_logic_of_knowl_bases}
% consider standard names not to be part of the signature; they form the domain of
% interpretation and terms are mapped to them.

A \emph{name substitution} $\varsub$ is a partial function that
assigns variables in $\vsymb$ names from $\names$; we also
write $x/\varsub(x)$ for $(x,\varsub(x))$. As usual, $\varsub$
is \emph{total} if its domain is $\vsymb$.  Given a variable
assignment $\varass$ for an interpretation \inter, we define
the set of named variables in \varass as
$V^{\inter,\varass}_\names=\{x\mid x^B\text{ is named}\}$.  A
substitution \varsub is \emph{associated with} \varass if its
domain is $V^{\inter,B}_\names$ and $x^\varass =
\varsub(x)^{\inter}$, for each $x\in V^{\inter,B}_\names$.
%
%A \emph{name substitution} $\varsub$ is a set $\{x_1/t_1, \dots,
%x_k/t_k\}$, where $x_1,\dots,x_k$ are distinct variables and
%$t_1,\dots,t_k$ are names;  \varsub is \emph{total} if it contains
%some $x_i/t_i$ for every variable $x_i \in \vsymb$. We define the set of named variables
%in a variable assignment $B$ as $V^{\names,B}=\{x\mid
% x^B\text{ is named}\}$. Given a variable
%assignment \varass for an interpretation \inter, a substitution \varsub
%is \emph{associated with} \varass
%%if $\varsub = \{x/t\mid
%%x\in \vsymb\text{, $t^w = x^B$, for some name }t\}$, then \varsub is
%if \varsub is of the form $\{x_1/t_1, \dots,
%x_k/t_k\}$, where $t_i\in\names$ and $t_i^\inter = x_i^B$, $1\leq i\leq k$, and $\{x_1,\dots,x_k\}=V^{\names,B}$.
%%  $ which holds $\varsub = \{x/t\mid
%%x\in \vsymb\text{, $t^w = x^B$, for some name }t\}$ is
The \emph{application} of a
name substitution $\varsub$ to some term, formula, or theory $ \chi $,
denoted by $
\chi\varsub$, is defined as syntactical
replacement, as usual. Clearly, if the unique names assumption
applies, each variable assignment has a unique associated
substitution; if the standard names assumption applies, each
associated substitution is total.

\begin{example}
  \label{ex:fo-equality}
  Consider a language \flang~with constants $\fsymb = \{a,$ $b,c\}$, and
  an interpretation $\interdef$ with $U = \{k,l,m\}$
  such that $a^w = k$, $b^w = l$, and $c^w = l$, and the variable
  assignment \varass: $x^B = k$, $y^B = l$, and $z^B =
  m$.  \varass has two associated name substitutions, $\varsub_{1}
  = \{x/a, y/b\}$ and $\varsub_{2} = \{x/a,y/c\}$, which are both not total.
\end{example}

%%%%%%%%%%%%%%%%%%%%%%%%%%%%%%%%%%%%%%%%%%%%%%%%%%%%%%%%%%%%%%%%%%%%%
\subsection{Logic Programs}
\label{sec:preliminaries-programs}
%%%%%%%%%%%%%%%%%%%%%%%%%%%%%%%%%%%%%%%%%%%%%%%%%%%%%%%%%%%%%%%%%%%%%

A \emph{disjunctive logic program} \lp~consists of rules of the form
\begin{equation}
\label{eq:rule} h_1 \lpor \dots \lpor h_l\ \leftarrow\ b_1,\ \dots,\ b_m,\
\dnot{c_{1}},\ \dots,\ \dnot{c_n}
\end{equation}
 where
$h_1,\dots,h_l,b_1,\dots,b_m,c_1,\dots,c_n$ are equality-free
atoms, with $ m, n \geq 0 $ and $ l \geq 1$. $H(r) = \{h_1,\dots,h_l\}$
is the set of \emph{head atoms}
of $r$, $B^+(r) = \{b_1,\dots,b_m\}$ is the set of \emph{positive
body atoms} of $r$, and $B^-(r)=\{c_{1},\dots,c_n\}$ is the set of
\emph{negated body atoms} of $r$. If $l = 1$, then $r$ is
\emph{normal}. If $B^-(r) = \emptyset$, then $r$ is \emph{positive}.
If every variable in $r$ occurs in $B^+(r)$, then $r$  is
\emph{safe}. If every rule $r \in \lp$ is normal (resp., positive, safe), then
\lp is normal (resp., positive, safe).

Each program \lp has a signature $\Sigma_\lp$, which contains
the function and predicate symbols that occur in \lp.
We assume that $\signature_\lp$ contains some
$0$-ary function symbol if it has predicate symbols of arity greater than $0$.
With $\flang_\lp$ we denote the first-order language over
$\signature_\lp$.  As
usual, \emph{Herbrand
interpretations} $M$ of \lp are subsets of the set of ground
atoms of $\flang_\lp$.

The \emph{grounding} of a logic program \lp, denoted $\gr{\lp}$, is
the union of all possible ground instantiations of \lp, obtained by
replacing each variable in a rule $r$ with a name in $\names_{\signature_\lp}$,
for each rule $r
\in \lp$.

Let \lp~be a positive program. A Herbrand interpretation $M$ of
\lp~is a \emph{Herbrand model} of \lp if, for every rule $r \in \gr{\lp}$,
$B^+(r) \subseteq M$ implies $H(r) \cap M \not= \emptyset$, and,
for every $ t \in \names_{\signature_\lp} $, $ t=t \in M$.
%\comment{jb: do we need to include a remark as to why the equality
%atoms must be included in the models, or is it clear? ap: no, it is
%by no means clear at this point.}
A Herbrand model $M$ is \emph{minimal} iff for every model $M'$ such
that $M' \subseteq M$, $M' = M$.

Following \citeN{Gelfond+Lifschitz-ClasNegaLogiProg:91}, the
\emph{reduct} of a logic program \lp~with respect to an
interpretation $M$, denoted $\lp^M$, is obtained from $\gr{\lp}$~by
deleting (i) each rule $r$ with $B^-(r) \cap M \neq \emptyset$ and
(ii) $\dnot c$ from the body of every remaining rule $r$ with $c \in
B^-(r)$. If $M$ is a minimal Herbrand model of $\lp^M$, then $M$ is
a \emph{stable model} of $\lp$.

\begin{example}
  Consider the program
$$\lp=\{p(a); \ p(b); \ q(x) \mid r(x) \leftarrow p(x), \dnot s(x)\}$$
and the interpretation $M_1=\{p(a),p(b),q(a),r(b)\}$.\footnote{For brevity,
we leave out equality atoms in the example.} The reduct
$$\lp^{M_1} = \{p(a); p(b); q(a) \mid r(a) \leftarrow p(a);
q(b) \mid r(b) \leftarrow p(b)\}$$
has $M_1$ as a minimal model, thus $M_1$ is a stable model of \lp.  The other stable models
of \lp are $M_2=\{p(a),$ $p(b),$ $q(a),$ $q(b)\}$, $M_3=\{p(a),$
$p(b),$ $q(b),$ $r(a)\}$, and $M_4=\{p(a),$ $p(b),$ $r(a),$ $r(b)\}$.
\end{example}

%A disjunctive logic program $\lp$ is \emph{consistent} if it has a stable
%model. Furthermore, $\lp$~\emph{cautiously entails} a
%ground atom $\alpha$ if $\alpha \in M$ for every stable
%model $M$ of \lp and \lp~\emph{bravely
%  entails} $\alpha$ if $\alpha \in M$ for some
%stable model $M$ of \lp.

\section{First-order Autoepistemic Logic}
\label{sec:first-order-auto}

We adopt first-order autoepistemic logic (FO-AEL)
under the any- and all-name semantics of
\citeN{konolige91:_quant}.
These
semantics
allow quantification over arbitrary domains
and generalize classical first-order logic with equality, thereby
allowing a trivial embedding of first-order theories (with equality). Other
approaches like those by
\citeN{kaminski02:_revis_quant_autoep_logic} or \citeN{levesque00:_logic_of_knowl_bases}
require interpretations to follow the unique or standard names assumptions
and therefore do not allow such direct embeddings.

An {\em FO-AEL language \fmodl} is defined relative to a first-order
language \flang by allowing the unary modal operator \mop
in the construction of formulas---$\mop\phi$ is usually read as  ``$\phi$ is known'' or ``$\phi$ is believed''.
%\begin{enumerate}[i]
%\item any atomic formula in
%\flang~is a formula in \fmodl;
%\item if $\phi$ is a
%{formula} in \fmodl,
%  then $\mop \phi$, called a \emph{modal atom},\footnote{$\mop \phi$ is
%    usually read as  ``$\phi$ is
%    known'' or ``$\phi$ is believed.''} is a formula in
%  \fmodl; and
%\item {complex formulas}
%  are constructed as in standard first-order logic.
%\end{enumerate}
As usual, closed formulas, i.e., formulas without free variable occurrences, are
called \emph{sentences}; formulas of the form $\mop\phi$, where $\phi$ is a formula, are {\em modal atoms}; and $\mop$-free
formulas are \emph{objective}. \emph{Standard autoepistemic logic} is variable-free
FO-AEL.

To distinguish between semantic notions defined for the any-  resp.\ all-name semantics, we use the symbols \anysymbol (``\emph{E}xistence of name'') and \allsymbol (``for \emph{A}ll names'').

An \emph{autoepistemic interpretation} is a pair $\langle
\inter,\Gamma \rangle$, where \interdef is a first-order
interpretation and $\Gamma \subseteq \fmodl$ is a set of sentences,
called a \emph{belief set}. Satisfaction of a formula $\mop \phi$ in an interpretation $\langle
\inter,\Gamma\rangle$ with respect to a variable assignment \varass under
the \emph{any-name semantics}, denoted $(\inter,B) \anymodels_{\Gamma} \mop \phi$, is
defined as
\begin{itemize}
 \item[]
$(\inter,B) \anymodels_{\Gamma} \mop \phi$ iff,
for \emph{some}  name
      substitution $\varsub$ associated with $B$, $\phi\varsub$ is closed
and  $\phi\varsub \in \Gamma$.
\end{itemize}
Satisfaction of arbitrary formulas is then as follows, where $\phi ,\psi \in\fmodl$:
\begin{itemize}
\item $(\inter,B) \anymodels_{\Gamma} p(t_1,\dots,t_n)$ iff
$(t_1^{\inter,B},\dots,t_n^{\inter,B}) \in p^I$;

\item $(\inter,B) \anymodels_{\Gamma} t_1=t_2$ iff $t_1^{\inter,B} = t_2^{\inter,B}$;

\item $(\inter,B) \anymodels_{\Gamma} \lnot \phi$ iff $(\inter,B)
      \not\anymodels_{\Gamma} \phi$;

\item $(\inter,B) \anymodels_{\Gamma} \phi \land \psi$ iff $(\inter,B) \anymodels_{\Gamma} \phi $ and $(\inter,B) \anymodels_{\Gamma} \psi$;

\item $(\inter,B) \anymodels_{\Gamma} \exists x\st \phi$ iff for some
      $x$-variant $B'$ of $B$, $(\inter,B') \anymodels_{\Gamma} \phi$.
\end{itemize}
%% JB: as these connectives are just shortcuts, we don't need the additional definitions
%From these, corresponding conditions for the defined connectives $\lor$ and $\limpl$, as well as for $\forall$, follow in the usual manner.
An interpretation $\langle \inter,\Gamma \rangle$ is a \emph{model}
                of $\phi$, denoted $\inter \anymodels_{\Gamma} \phi$,
if $(\inter,B) \anymodels_{\Gamma} \phi$
for every variable assignment $B$ for $\inter$. This extends to
sets of formulas in the usual way. A set of formulas $\fot \subseteq \fmodl$
\emph{entails} a formula $\phi\in \fmodl$  with respect to a belief set $\Gamma$,
denoted $\fot \anymodels_\Gamma \phi$, if for every interpretation $\inter$
such that $\inter \anymodels_\Gamma \fot$, $\inter \anymodels_\Gamma
\phi$.

The notions of satisfaction and entailment under the {\em all-name
semantics}, for which we use the symbol $\allmodels_{\Gamma}$, are analogous, with
the only difference that satisfaction of modal atoms is defined as
\begin{itemize}
\item[] $(\inter,B) \allmodels_{\Gamma} \mop \phi$ iff, for \emph{all} name
 substitutions  $\varsub$ associated with $B$, $\phi\varsub$ is closed
and    $\phi\varsub \in \Gamma$.
\end{itemize}
Note that the any- and all-name semantics always coincide for objective
formulas and, if the unique (or standard) names assumption applies,
also for arbitrary formulas in $\fmodl$; this was also observed by
\citeN{kaminski02:_revis_quant_autoep_logic}.
In such situations, i.e., where both semantics coincide, we sometimes
use $\models_\Gamma$ rather than $\anymodels_\Gamma$ or $\allmodels_\Gamma$.
Furthermore, when talking about entailment $\fot \models_\Gamma \phi$ under the standard names assumption, we mean entailment considering only interpretations for which the
standard names assumption holds. That is,  $\fot \models_\Gamma \phi$ under the
standard names assumption if for every
interpretation \inter such that the standard names assumption applies in \inter and $
\inter \models_\Gamma \fot $, $ \inter \models_\Gamma \phi $.

\begin{example}
  \label{ex:any-name}
  Consider the formula $\phi = \forall x(p(x) \limpl \mop
  p(x))$ and some interpretation $\langle \inter,\Gamma\rangle$.
Then, $\inter
  \anymodels_\Gamma \phi$ iff, for every variable assignment \varass, $(\inter,B)
  \anymodels_{\Gamma} p(x) \limpl \mop p(x)$, which in turn holds iff
   $(\inter,B)\not\anymodels_{\Gamma}
  p(x)$ or $(\inter,B)\anymodels_{\Gamma} \mop p(x)$.
  Now,
  $(\inter,B)\anymodels_{\Gamma} \mop p(x)$, with $x^B = k$, iff, for some $t \in
  \names_\signature$, $t^\inter = k$, and $p(t) \in \Gamma$. Thus, $\phi$ is
false (unsatisfied) in
  any interpretation where $p^I$ contains unnamed individuals. Analogous
  for the all-name semantics.
\end{example}
The following example illustrates the difference between the any- and all-name semantics.

\begin{example}
\label{ex:any-all} Consider a language with constant symbols $a,b$
and unary predicate symbol $p$, and an interpretation $\langle
\inter,\Gamma \rangle$ with $\inter = \langle
\{k\},\interfun\rangle$ and $\Gamma = \{p(a)\}$. Then,
$\inter \anymodels_\Gamma \exists x. \mop p(x)$, while
$\inter \not\allmodels_\Gamma \exists x. \mop p(x)$,
since $b^\inter =a^\inter = k$ but $p(b) \notin \Gamma$.
\end{example}

A stable expansion is a set of beliefs of an ideally
introspective agent (i.e., an agent with perfect reasoning
capabilities and with knowledge about its own beliefs), given
some theory $ \fot \subseteq \fmodl $. Formally, a belief set
$T \subseteq \fmodl$ is a \emph{\anystable expansion} of a
theory $\fot \subseteq \fmodl$ iff $T = \{\phi \mid \fot
\anymodels_T \phi\}\text{.}$ Similarly, $T$ is a
\emph{\allstable expansion} of $\fot$ iff $T = \{\phi \mid \fot
\allmodels_T \phi\}\text{.}$

Recall that $\flang_g $ and $ \flang_{ga}$
denote the restrictions of $\flang $ to ground and ground  atomic
formulas, respectively. Given a set of sentences $ \Gamma \subseteq \fmodl $, $\Gamma_o$, $\Gamma_{og}$, and $\Gamma_{\oga}$ denote the
restrictions of $\Gamma $ to objective, objective ground, and objective
ground atomic formulas, respectively, i.e.,~$\Gamma_o=\Gamma\cap \flang$,
$\Gamma_{og}=\Gamma\cap\flang_g$, and $\Gamma_{\oga}=\Gamma\cap\flang_{ga}$.

Every stable expansion $T$ of \fot is a \emph{stable set} \cite{stalnaker:93}, which means that it
satisfies the following conditions: (a) $T$ is closed under
first-order entailment, (b) if $\phi \in T$ then $\mop \phi \in T$,
and (c) if $\phi \notin T$ then $\lnot\mop \phi \in T$. Furthermore,
if $T$ is consistent, the converse statements of (b) and (c) hold.

\citeN{konolige91:_quant} shows that a stable expansion $T$
of a theory $ \fot \subseteq \fmodl $ is determined by its objective subset $T_o$,
also called the \emph{kernel} of $T$. He further obtained the following
result:

\begin{proposition}[\cite{konolige91:_quant}]
\label{prop:no-nesting-expansions}
Let $\fot \subseteq \fmodl$ be a theory without nested
modal operators, $\Gamma\subseteq\flang$ a set of objective formulas,
and $\mathrm{X}\in \{\mathrm{E},\mathrm{A}\}$.  Then, $\Gamma = \{
\phi \in \flang\mid \fot
\xmodels_{\Gamma} \phi\}$ iff $\Gamma=T_o$, for some
\xstable expansion $T$ of \fot.
\end{proposition}
We slightly adapt this result as follows:

 \begin{proposition}
 \label{prop:atoms-expansions}
Let $\fot \subseteq \fmodl$ be a theory with only objective
 atomic formulas in the scope of occurrences of $\mop$,
$\Gamma\subseteq\flang$ a set of objective formulas,
and $\mathrm{X}\in \{\mathrm{E},\mathrm{A}\}$.
 Then, $\Gamma = \{\phi \in \flang \mid \fot \xmodels_{\Gamma_{\oga}}\phi\}$
 iff $\Gamma=T_o$, for some \xstable expansion $T$ of \fot.
 \end{proposition}
  \begin{proof}
Since modal atoms in \fot contain only objective atomic formulas,
we obtain $\fot \xmodels_{\Gamma_o} \phi$ iff $\fot
\xmodels_{\Gamma_{\oga}} \phi$, because, by the definition of
satisfaction of modal formulas, non-ground and non-atomic formulas
in $\Gamma_o$ do not affect satisfaction of formulas in $\fot$. Thus,$\{
\phi \in \flang\mid \fot \xmodels_{\Gamma_{\oga}} \phi\}=\{ \phi \in
\flang\mid \fot \xmodels_{\Gamma_o} \phi\}$ follows.

Since there is no nesting of modal operators in $\fot $,
we
combine this result with Proposition \ref{prop:no-nesting-expansions} to
obtain $\Gamma_o = \{ \phi \in \flang\mid \fot \xmodels_{\Gamma_o}
\phi\}=\{ \phi \in \flang\mid \fot \xmodels_{\Gamma_{\oga}} \phi\}$ iff
$\Gamma_o=T\cap\flang$ is the kernel of a \xstable expansion $T$ of
$\fot$.%
%
%The proof for the case of the all-name semantics is analogous.
  \end{proof}

We note here that, unlike in standard autoepistemic logic, in FO-AEL two different stable expansions may have the same objective subsets, both under
the any- and all-name semantics.  Consider, for example, the theories $
\fot=\{\forall x \st p(x) \} $ and $ \fot'=\{\forall x \st \mop p(x)\}$ and
their respective stable expansions $ T $ and $ T' $.  We have that $ T_o=T'_o $
is the closure under first-order entailment of $ \{\forall x \st p(x)\} $,
but we also have that $\forall x \st \mop p(x) \in T' $ but $ \forall x \st \mop p(x) \notin T $,
because $\forall x \st \mop p(x) $ is not satisfied in any interpretation
that has unnamed individuals.

\section{Embedding Non-Ground Logic Programs}
\label{sec:embedd-logic-progr}

We define an embedding $ \tau $ as a function that takes a logic program \lp
as its argument and returns a set of sentences in the FO-AEL
language obtained from $\signature_\lp$.

\citeN{janhunen99:_inter_non_logic} studied translations between nonmonotonic formalisms and formulated a number of desiderata for such translation functions, namely faithfulness, polynomiality, and modularity (FPM). We adapt these notions to our case of embedding logic programs into FO-AEL.

An embedding $ \tau $ is \emph{faithful} if, for any logic program $\lp$, there is a one-to-one
correspondence between the stable models of $\lp$ and the consistent stable expansions of $\tau(\lp)$
with respect to %equality-free
ground atomic formulas.

An embedding $\tau$ is \emph{polynomial} if, for
  any  logic program  $\lp$,
$\tau(\lp)$ can be computed in time polynomial in the size of $ \lp $.

An embedding $\tau$ is \emph{modular} if, for
  any two logic programs  $\lp_1$ and $\lp_2$,
$\tau(\lp_1 \cup
  \lp_2) = \tau(\lp_1) \cup
  \tau(\lp_2)$.
Furthermore, we call $\tau$ \emph{signature-modular} if, for any two logic programs
$\lp_1$ and $\lp_2$ with the same signature
\signature, $\tau(\lp_1 \cup \lp_2) = \tau(\lp_1)
\cup \tau(\lp_2)$.

Since the unique names assumption does not hold in FO-AEL in
general, it is necessary to axiomatize default uniqueness of names
(as introduced by \citeN{konolige91:_quant}) to assure
faithfulness of several of the embeddings. Given a signature \signature,
by $\UNA_\signature$ we
denote the set of axioms
\begin{description}
\item[\UNA] $\quad\lnot \mop(t_1=t_2) \limpl t_1\neq t_2\text{,}
 \quad\text{for all
  distinct
 }
t_1,t_2 \in \names_{\signature}\text{.}$
\end{description}
Default uniqueness, in contrast to
rigid uniqueness (i.e., $\UNA$ axioms of the form $t_1\neq t_2$),
allows first-order theories that are later combined with the
embedding to ``override'' such inequalities, rather than introducing
inconsistency. For example, the theory $\fot= \{\lnot \mop(a=b) \limpl a\neq b \} $
 has a single expansion that includes $a \neq b $; the single expansion of
  $ \fot \cup \{a=b \} $ is consistent and
 includes $ a=b $%; the inequality is ``overridden''
 .

Observe that the $\UNA$ axioms depend on the signature. In addition, the union of the $\UNA$
axioms of two signatures is not necessarily the same as the set of $\UNA$ axioms of the
union of these two signatures: given two signatures $ \signature_1 $ and $ \signature_2
$ such that $ \fsymb_1\neq \fsymb_2 $, $ \UNA_{\signature_1} \cup \UNA_{\signature_2}
\neq \UNA_{\signature_1\cup \signature_2} $, i.e., the $\UNA$ axioms corresponding to
different signatures cannot be combined in a modular fashion.  This means that
embeddings that include such $\UNA$ signatures are not modular, but may be
signature-modular.

We first present the embeddings of normal programs and then proceed with the embeddings of disjunctive programs.

\subsection{Embedding Normal Logic Programs}

We consider three embeddings of non-ground logic programs into FO-AEL, denoted $\embhp$,
$\embeb$, and $\embeh$. ``$\mathit{HP}$'' stands for
``$\mathit{H}$orn for $\mathit{P}$ositive rules'' (positive rules
are translated to objective Horn clauses); ``$\mathit{EB}$'' stands
for ``$\mathit{E}$pistemic rule $\mathit{B}$odies'' (the body of a
rule can only become true if it is \emph{known} to be true); and
``$\mathit{EH}$'' stands for ``$\mathit{E}$pistemic rule
$\mathit{H}$eads'' (if the body of a rule is true, the head is
\emph{known} to be true).

The $\mathit{HP}$ embedding is an extension of the one which
originally led Gelfond and Lifschitz to the definition of the stable
model semantics
\cite{gelfond87:_strat_autoep_theor,Gelfond+Lifschitz-StabModeSemaLogi:88}.
The
$\mathit{EB}$ and $\mathit{EH}$ embeddings are extensions of embeddings by
\citeN{marek93:_reflex_autoep_logic_logic_progr}. The $\mathit{EH}$ embedding was
independently described by \citeN{lifschitz93:_exten}
and
by \citeN{chen93:_minim}. The original motivation for the $\mathit{EB}$ and $\mathit{EH}$ embeddings was the possibility to directly embed programs with
strong negation and disjunction.
Furthermore, Marek and
Truszczynski arrived at their embeddings through embeddings of
logic programs in  \emph{reflexive autoepistemic
logic}~\cite{Schwartz-ReflAutoLogi:92},
which is equivalent to McDermott's nonmonotonic modal
$\mathbf{sw5}$~\cite{mcdermott82:_nonmon_ii},
and the subsequent embedding of reflexive autoepistemic logic into standard
AEL.
Lifschitz and Schwarz arrived at the $\mathit{EH}$ embedding through an embedding of
logic programs in Lifschitz's
nonmonotonic logic of \emph{minimal
belief and
negation-as-failure}\/ (MBNF)~\cite{Lifschitz-MiniBeliNegaFail:94}
and the subsequent embedding
of MBNF into standard AEL.  Finally, Chen also
arrived at the $\mathit{EH}$ embedding
via MBNF, but
he subsequently embedded MBNF in Levesque's \emph{logic of only
knowing}~\cite{levesque90:_all_i},
a subset of which corresponds with standard AEL.

\begin{definition}
  \label{def:normal-embedding}
  Let $r$ be a normal rule of the
form \eqref{eq:rule}. Then,
\[
\begin{array}{r@{~}c@{~}l}
     \embhp(r) & = & (\forall)\; \textstyle{\bigwedge_i} b_i  \land
\textstyle{\bigwedge}_j
     \lnot \mop c_j \limpl h_1,\\[.5ex]
     \embeb(r) & = &
(\forall)\; \textstyle{\bigwedge_i} (b_i \land \mop b_i) \land
\textstyle{\bigwedge_j} \lnot \mop
c_j \limpl h_1, \\[.5ex]
     \embeh(r) & = &
(\forall)\;\textstyle{\bigwedge_i} (b_i \land \mop b_i) \land
\textstyle{\bigwedge_j} \lnot \mop
c_j \limpl h_1 \land \mop h_1.%\\%[-3.5ex]
  \end{array}
\]
Furthermore, given a normal logic program $P$, we define:
\[
\emb(\lp) = \{\emb(r)\mid r\in \lp\} \cup \UNA_{\signature_\lp}, \ \ \chi \in
    \{\mathit{HP},\mathit{EB},\mathit{EH}\}.
\]
\end{definition}

For all three embeddings, we assume
$\signature_{\emb(\lp)} = \signature_\lp$ (here and henceforth
``$\chi$'' ranges over $\mathit{HP}$,
$\mathit{EB}$, and $\mathit{EH}$). Furthermore, by $\emb^-$ we
denote the embedding $\emb$ \emph{without} the $\UNA$ axioms: given
a normal logic program \lp, $\emb^-(\lp)=\emb(\lp) -
\UNA_{\signature_\lp}$. The embeddings $\emb^-$ are modular and polynomial.
The embeddings  $\emb$ are signature-modular
and polynomial, provided $\names_{\signature_\lp}$ is polynomial in the size of $ \lp $
(e.g., if there are no function symbols with arity greater than $ 0 $).
In the examples of embeddings in the remainder of the paper we do not
write the $\UNA$ axioms explicitly.

A notable difference between the embedding \embhp, on the one hand,
and the embeddings $\embeb $ and $\embeh$, on the other,
is that, given a logic program \lp, the stable expansions
of $\embhp(\lp)$ include the ``contrapositives'' of the rules in \lp
(viewed classically and where $\neg\mop a$ is $\dnot{}a$), which is not
true for $\embeb(\lp)$ and $\embeh(\lp)$ in general.

\begin{example}
\label{ex:contrapositive}
  Consider
$\lp = \{p \leftarrow q, \dnot r\}$. The
  stable expansion of $\embhp(\lp)= \{q\land \lnot \mop r
  \limpl p\}$ includes $\neg p \limpl \neg q\lor \mop r$;
  the expansion of $\embeb(\lp)= \{q \land \mop q \land \lnot \mop r \limpl
p\}$
  includes $\neg p \limpl \neg \mop q
  \lor \neg q\lor \mop r$, but not $\neg p \limpl \neg q\lor \mop r$.
\end{example}

For standard AEL and ground logic programs, the following
faithfulness result straightforwardly extends results by
\citeN{Gelfond+Lifschitz-StabModeSemaLogi:88} and
\citeN{marek93:_reflex_autoep_logic_logic_progr}.
%\cite{Gelfond+Lifschitz-StabModeSemaLogi:88,marek93:_reflex_autoep_logic_logic_progr}.
\begin{proposition}
  \label{prop:prop-emb}
A Herbrand interpretation $M$ of a ground normal logic
program \lp is a stable model of  \lp iff there exists a
consistent stable expansion $T$ of $\emb^-(\lp)$  in standard
AEL such that $M=T\cap\flang_{ga}$.
\end{proposition}

Observe from the proposition that we do not require the UNA
axioms in the embeddings of ground programs.  These axioms are
required in the general case when embedding non-ground
programs, as illustrated by Example \ref{ex:emb-any-all} below.
The following example illustrates the embeddings for the case
of non-ground programs.

\begin{example}
\label{ex:general}
  Consider
$\lp = \{q(a);$ $p(x);$ $r(x) \leftarrow
  \dnot s(x),$ $ p(x)\}$, which has the single stable model $M =
  \{q(a),p(a),r(a)\}$. Likewise, each of the embeddings
$\emb(\lp)$ has a single consistent stable expansion $T^\embvar$:
\[
  \begin{array}{@{}r@{~}c@{~}l@{}}
    T^{HP} & = & \{q(a), p(a), \mop p(a), \dots,\forall x (p(x)),\lnot \mop \forall x (\mop p(x)),
      \forall x (\lnot\mop s(x) \limpl r(x)),
      \dots \},\\[1ex]
T^{EB} & = & \{q(a), p(a), \mop p(a), \dots,\forall x (p(x)),\lnot \mop \forall x
(\mop p(x)), \lnot\mop(\forall x (\lnot\mop s(x) \limpl r(x))), \dots \}, \\[1ex]
T^{EH} & = &
  \{q(a), p(a), \mop p(a), \dots, \forall x (p(x)),\forall x(\mop p(x)),
      \forall x (\lnot\mop s(x) \limpl r(x)),
  \dots \}.
\end{array}
\]
\end{example}

The stable expansions in Example \ref{ex:general} agree on
objective ground atoms, but not on arbitrary formulas. We now
extend Proposition \ref{prop:prop-emb} to the non-ground case.
To this end, we use the following two lemmas.

\begin{lemma}
  \label{lem:named-model}
Let \lp be a normal logic program, let $\mathrm{X}\in\{\mathrm{E},\mathrm{A}\}$,
let $ T $ be a \xstable expansion of
$\emb(\lp)$, and let $ \alpha $ be an objective ground atom. Then, $\emb(\lp)
  \xmodels_{T_{\oga}} \alpha$ iff
$\emb(\lp) \models_{T_{\oga}} \alpha$ under the standard names
assumption. Moreover, $\embhp^-(\lp)
\allmodels_{T_{\oga}} \alpha$ iff $\embhp^-(\lp) \models_{T_{\oga}} \alpha$
under the standard names assumption.
\end{lemma}

  \begin{proof}
   We start with the first statement.

\noindent ($ \Rightarrow $)
% By definition, $\emb(\lp)
%  \anymodels_{T_{\oga}} \alpha$ iff for every
%interpretation \inter it holds that whenever $ \inter
%\anymodels_{T_\oga}\emb(\lp)$, $ \inter \anymodels_{T_\oga}\alpha $. If for every
%interpretation \inter, whenever $ \inter \anymodels_{T_\oga}\emb(\lp)$, $ \inter
%\models_{T_\oga}\alpha $, then it must hold that for every interpretation $
%\inter' $ for which the standard names assumption applies, whenever $ \inter'
%\models_{T_\oga}\emb(\lp)$, $ \inter' \models_{T_\oga}\alpha $, and thus (again,
%by definition) $\emb(\lp)
%  \models_{T_{\oga}} \alpha$ under the standard names assumption. Therefore, if
%$\emb(\lp)
%  \anymodels_{T_{\oga}} \alpha$, then
%$\emb(\lp)
%  \models_{T_{\oga}} \alpha$ under the standard names assumption.
%  Analogous for the all-name semantics.
This is obvious, as interpretations under the standard names
assumption are just special interpretations.

\medskip

 \noindent   ($ \Leftarrow $)  We start with the case of the any-names assumption.
 Assume, on the contrary, that $\emb(\lp)
  \models_{T_{\oga}} \alpha$ under the standard names assumption, but
$\emb(\lp)
  \not \anymodels_{T_{\oga}} \alpha$.  This means
that there is some interpretation $\interdef$ such that $
\inter \anymodels_{T_{\oga}}\emb(\lp) $, but $ \inter \not \anymodels_{T_\oga} \alpha $.

By the fact that the only occurrences of the equality symbol in $\emb(\lp)$
are in the $\UNA$ axioms, the only atoms in $
T_{\oga}$ involving equality are of the form $ t=t $, for $ t \in \names_{\signature_\lp} $. Consider two distinct names $ t_1, t_2\in \names_{\signature_\lp} $ and
the $\UNA$ axiom $ \lnot \mop t_1 =t_2 \limpl t_1\neq t_2\in \UNA_{\signature_\lp} $.
Since $ \langle \inter, T_{\oga} \rangle $ is a model of the axiom and $ t_1 =t_2\notin
T_{\oga} $, $ \inter \anymodels_{T_{\oga}}t_1\neq t_2$.  Consequently, it must be the
case that $\interfun$ maps every name
to a distinct individual in \domain.

We assume that the mapping $\cdot^{I}$ extends to ground terms in the natural way, i.e.,
$f(t_1,\dots, t_m)^I= f^I(t_1^I, \dots, t_m^I)$.
We construct the interpretation $\inter' =
\langle U',\cdot^{I'} \rangle$ as
    follows: $U' = \names$, $t^{I'}= t$, for $t \in \names$, and
    $\langle t_1,\dots,t_n \rangle \in p^{I'}$ if $\langle t_1^I,\dots,t_n^I
    \rangle \in p^{I}$ for $n$-ary predicate symbol $p$ and every $\langle t_1,\dots, t_n \rangle \in \names^n$. Clearly, the standard names
    assumption holds for $\inter'$, and $\inter$ and $\inter'$ agree on
objective
   ground atoms, i.e., $\inter\models \alpha$ iff $\inter'\models\alpha$ for any
$\alpha\in\flang_{ga}$. We now
show that $ \inter' \models_{T_\oga}\emb(\lp) $.

Clearly, $\langle \inter',T_{\oga}\rangle$
     satisfies the $\UNA$ axioms
      since the standard names
assumption holds for $ \inter' $ and
since $ T_{\oga} $ contains only the trivial equalities.
We first consider the embedding \embeh and some
\begin{center}
    $(\forall)\;\text{ } \bigwedge_{1\leq i \leq m}(b_i \land \mop b_i )\land \bigwedge_{1\leq j \leq n} (\lnot \mop c_j) \limpl h_1 \land
\mop h_1 \in \embeh(\lp)$.
\end{center}
  Since $\inter
    \models_{T_{\oga}} \embeh(\lp)$,
\begin{center}
    $(\inter,B) \models_{T_{\oga}} \bigwedge_{1\leq i \leq m}(b_i \land \mop b_i) \land \bigwedge_{1\leq j \leq n}( \lnot \mop c_j)\limpl h_1 \land
\mop h_1$
\end{center}
for every variable assignment $B$ of $\inter$.

Now, consider a variable
assignment $B'$ of $\inter'$ and the corresponding variable assignment
$B$ of $\inter$, which we define as follows:
$x^B = k$ iff there is a $ t \in \names_{\signature_\lp}$ such
that $x^{B'} = t$ and $t^I = k$. Observe that $ B $ assigns every variable
to a named individual. Consider a name substitution
\varsub which is associated with $ B $; since all names are interpreted as
distinct individuals (by the $\UNA$ axioms),
(\dag) \varsub is
unique. Moreover, by
construction of $B$,
\varsub is also the only substitution associated with $B'$ .

By construction of $\inter'$, and
since \varsub is the unique substitution associated with
$ B $ (and $B'$), we have, for every
objective atom $ \alpha $ such that $ B $ is defined for all variables
in $ \alpha $, that
$(\inter,B) \anymodels_{T_{\oga}}
\alpha $ iff $(\inter',B') \anymodels_{T_{\oga}}
    \alpha $ and  $(\inter,B) \anymodels_{T_{\oga}}
    \mop\alpha $ iff $(\inter',B') \anymodels_{T_{\oga}}
    \mop\alpha $.  Consequently,
 if $(\inter,B) \anymodels_{T_{\oga}}
    h_1\land \mop h_1$, then
    $(\inter',B') \models_{T_{\oga}} h_1\land \mop h_1$, and if $(\inter,B)
\not \anymodels_{T_{\oga}} \bigwedge b_i \land \mop b_i$,
then
    $(\inter',B') \nmodels_{T_{\oga}} \bigwedge b_i \land \mop b_i$.
Furthermore,
 $(\inter,B)
    \not \anymodels_{T_{\oga}}\bigwedge \lnot \mop c_j$ implies $c_i\varsub \in T_{\oga}$ for some $i\in\{1,\ldots,n\}$.  Hence, $(\inter',B') \nmodels_{T_{\oga}}\lnot \mop
c_1\land \dots \land \lnot \mop c_n$. So,
\begin{center}
$(\inter',B') \models_{T_{\oga}} b_1 \land \mop b_1 \land \dots \land b_m
\land\mop b_m
    \land \lnot \mop c_1\land \dots \land \lnot \mop c_n \limpl h_1\land\mop h_1$.
\end{center}
Thus, we obtain $\inter' \models_{T_{\oga}} \embeh(\lp)$. Since $\inter$ and
$\inter'$ agree on objective
    ground atoms, $\inter'\nmodels_{T_{\oga}} \alpha$, and thus
  $\embeh(\lp) \nmodels_{T_{\oga}} \alpha$ under the standard names
  assumption.  This contradicts the initial assumption. Therefore, $\embeh(\lp) \anymodels_{T_{\oga}} \alpha$.

The argument for the  embeddings
$\embeb$ and $\embhp$ is analogous: simply leave out the positive
occurrences of modal atoms in the consequents, respectively consequents and
antecedents, in the argument above.

Likewise, the argument for the case of the all-name semantics is analogous.
Observe that in the argument about variable
assignments (\dag),  \varsub is the only name substitution associated with $ B $;
hence, the any- and all-name semantics coincide, and the subsequent
arguments immediately apply also for the all-name semantics.

    \medskip

    \noindent
For the second statement, consider the above argument without
the part about the $\UNA$ axioms and the following simple
adaptation: if $(\inter,B) \not \allmodels_{T_{\oga}}
\bigwedge\lnot \mop c_j$, then for all associated name
substitutions \varsub, there is some $c_i\varsub \in T_{\oga}$,
$1\leq i\leq n$.  One of these name substitutions is the one
associated with $B'$; the remainder of the argument remains the
same. It follows that $\embhp^-(\lp)\allmodels_{T_{\oga}}
\alpha$ iff $\embhp^-(\lp) \models_{T_{\oga}} \alpha$.
  \end{proof}

The latter fails for the embeddings $ \embeb^-$ and $
\embeh^-$, as there may be several name substitutions
associated with the assignment $ B $ in the ``$ \Leftarrow $''
direction above, while there is a single substitution
associated with $ B' $ (see also Example \ref{ex:emb-any-all}).

\begin{lemma}
  \label{lem:reduction}
  Let \lp be a normal logic program and $\mathrm{X}\in\{\mathrm{E},\mathrm{A}\}$.
  There exists a
\xstable expansion $T$ of $\emb(\lp)$  iff
  there exists
  a \xstable expansion $T'$ of $\emb(\gr{\lp})$ such that $T'_{\oga}=T_{\oga}$.
The same result holds for $\embhp^-$ and stable$^\mathrm{A}$ expansions.
\end{lemma}

  \begin{proof}
    We prove the first statement, first for the special case that the standard names
    assumption applies, and then use Lemma \ref{lem:named-model} to extend it
    to cases where the standard names assumption does not apply.

    Consider a belief set $\Gamma \subseteq \fmodl$ and an interpretation
$\inter$
    for which the standard names assumption holds. We claim that
($*$) $\inter \xmodels_\Gamma \emb(\gr{\lp})$
    iff $\inter \xmodels_\Gamma \emb(\lp)$.  By the standard names
    assumption, we have that
    $\inter\xmodels_\Gamma \emb(\lp)$ iff for every $\phi \in \emb(\lp)$,
$\inter
    \xmodels_\Gamma \phi$.  In turn, this holds iff for every variable assignment $B$, $(\inter,B)
    \xmodels_{\Gamma} \phi$, which in turn holds iff for the name substitution
    \varsub associated with $ B $ (which is unique and total, by the
    standard names assumption), $\inter \xmodels_\Gamma
    \phi\varsub$.
    By  definition, $\emb(\gr{\lp})$ contains all (and
only) the formulas
    of the form $\phi\varsub$ where $\phi \in \emb(\lp)$ and \varsub
    is a name substitution associated with some variable
    assignment $B$ for $\inter$; the claim ($*$) follows immediately from
    this.

\medskip

\noindent($\Rightarrow$) Let $ T $ be a stable expansion of $ \emb(\lp)$. By
Lemma~\ref{lem:named-model} and the
above we have:
 $$\{\phi\in\flang_{ga}\mid \emb(\lp)\xmodels_{T_{\oga}}\phi\} =
 \{\phi\in\flang_{ga}\mid
\emb(\gr{\lp})\xmodels_{T_{\oga}}\phi\}.$$
% Therefore
% (by Proposition \ref{prop:atoms-expansions}),
Hence by Proposition \ref{prop:atoms-expansions},
$$T'_o = \{\phi\in\flang\mid \emb(\gr{\lp}) \xmodels_{T_{\oga}}\phi\}$$
is the kernel of a stable expansion $T'$ of $\emb(\gr{\lp})$ and $ T'
\cap\lang_{ga}=T_{ga}$.

\medskip

\noindent The converse is
analogous. For  the second statement of the lemma, the same proof
using Lemma~\ref{lem:named-model} works.
\end{proof}

\begin{theorem}
  \label{the:embedd-stable}
  Let \lp be a normal logic program and $\mathrm{X}\in\{\mathrm{E},\mathrm{A}\}$.
  A Herbrand interpretation $M$
  is a stable model of \lp iff there exists a consistent \xstable
  expansion $T$ of $\emb(\lp)$  such that $M=T_{\oga}$.
The same result holds for $\embhp^-$ and stable$^\mathrm{A} $  expansions.
\end{theorem}

\begin{proof}
  By Lemma \ref{lem:reduction}, we can reduce embeddability of
  non-ground logic programs to embeddability of ground logic programs.

  Consider an embedding $\emb(\gr{\lp})$ and a stable expansion $ T $. Clearly,
  there is no interaction between the $\UNA$ axioms
  and the axioms resulting from rules in \lp.  Therefore, $\emb^-(\gr{\lp})$
  has a stable expansion $ T' $ such that $ T'_{\oga}=T_{\oga}$, and vice versa.
  The theorem then follows immediately from Proposition \ref{prop:prop-emb}.
\end{proof}
 Note that this result does not extend to the embeddings $\emb^-$ under the
any-name semantics, nor does it extend to the embeddings $\embeb^-$ and $\embeh^-$ under the all-name
semantics, as illustrated by the following example.
\begin{example}
\label{ex:emb-any-all}
Consider $\lp = \{p(n_1); r(n_2); q \leftarrow
\dnot p(x)\}$ such that $\signature_\lp$ has only two names, $n_1$
and $n_2$. \lp has one stable model, $M = \{p(n_1), r(n_2),q\}$.
$\embhp^-(\lp) = \{p(n_1);r(n_2);\forall x(\lnot \mop p(x) \limpl
q)\}$ has one \anystable expansion, $T = \{p(n_1),$ $ r(n_2),$ $ \mop
p(n_1),$ $ \mop r(n_2),$ $ \lnot \mop p(n_2),\dots\}$. $T$ does
\emph{not} include $q$. To see why this is the case, consider an
interpretation
$\inter$ with only one individual $k$. $\mop p(x)$ is trivially true
under the any-name semantics, because there is some name for $k$
such that $p(t) \in T$ (viz.\ $t = n_1$). In the all-name semantics,
this situation does not occur, because for $\mop p(x)$ to be true,
$p(t)$ must be included in $T$ for \emph{every} name  ($t=n_1$ and
$t=n_2$) for $k$. One can similarly verify that the result does not
apply to the embeddings $\embeb^-$ and $\embeh^-$ under the all-name
semantics, by the positive modal atoms in the antecedents.
\end{example}

\subsection{Embedding Disjunctive Logic Programs}

The embeddings $\embhp$ and $\embeb$ cannot be straightforwardly
extended to the case of disjunctive logic programs, even for the
propositional case. Consider the program $\lp = \{a \lpor b
\leftarrow\}$, which has two stable models: $M_1 = \{a\}$ and $M_2 =
\{b\}$. However, a naive extension of $\embhp$,
$\embhp(\lp) = \{a \lor b\}$, has one stable expansion $T = \{a
\lor b, \mop(a \lor b), \lnot \mop a, \lnot \mop b,\dots\}$. In
contrast, $\embeh$ can be straightforwardly extended because of the
modal atoms in the consequent of the implication: $\embehdis(\lp) =
\{(a \land \mop a) \lor (b \land \mop b)\}$ has two stable
expansions $T_1 = \{a \lor b, a, \mop a, \neg \mop b,\dots\}$ and
$T_2 = \{a \lor b, b, \mop b, \neg \mop a,\dots\}$.

The so-called \emph{positive introspection axioms} ($\PIA$s)
\cite{przymusinski91:_stabl_seman_disjun_progr} remedy this
situation for defining extensions \embhpdis and \embebdis of \embhp and \embeb, respectively.
Let $\PIA_\signature$ be the set of axioms
\begin{description}
\item[$\PIA$] $\quad\alpha \limpl \mop \alpha\text{,} \hspace{0.5cm}  $for every objective
ground atom $\alpha$ of
$\lang_\signature$.
\end{description}
Each $\PIA$ ensures that a consistent stable expansion contains either
$\alpha$ or $\lnot \alpha$.

It would have been possible to define the $\PIA$s in a different way:
$(\forall)\; \phi \limpl\mop\phi$ for any objective atomic formula
$\phi$. This would, however, effectively close the domain of the
predicates in $\signature_\lp$ (see Example \ref{ex:any-name}). We deem
this aspect undesirable in combinations with FO theories.

\begin{definition}%[Embeddings of disjunctive programs]
  \label{def:disjunctive-embedding}
Let $r$ be a rule of form \eqref{eq:rule}. Then:
% \vspace{-0.18cm}
\[
\begin{array}{r@{~}c@{~}l}
     \embhpdis(r) & = & (\forall)\; \textstyle{\bigwedge_i} b_i  \land
\textstyle{\bigwedge_j}
     \lnot \mop c_j \limpl \textstyle{\bigvee_k} h_k,\\[.5ex]
     \embebdis(r) & = &
(\forall) \; \textstyle{\bigwedge_i} (b_i \land \mop b_i) \land
\textstyle{\bigwedge_j} \lnot \mop
c_j \limpl \textstyle{\bigvee_k} h_k, \\[.5ex]
     \embehdis(r) & = &
(\forall)\;\textstyle{\bigwedge_i} (b_i \land \mop b_i) \land
\textstyle{\bigwedge_j} \lnot \mop c_j \limpl \bigvee_k (h_k \land
\mop h_k).
\end{array}
\]
Given a disjunctive logic program $P$, we define:
\[
\begin{array}{r@{~}c@{~}l}
     \embhpdis(\lp) & = & \{\embhpdis(r)\mid r\in
    \lp\} \cup \PIA_{\Sigma_\lp}\cup \UNA_{\Sigma_\lp}, \\[.5ex]
     \embebdis(\lp) & = & \{\embebdis(r)\mid r\in
    \lp\} \cup \PIA_{\Sigma_\lp}\cup \UNA_{\Sigma_\lp}, \\[.5ex]
     \embehdis(\lp) & = & \{\embehdis(r)\mid r\in \lp\}\cup \UNA_{\Sigma_\lp}.
  \end{array}
\]
\end{definition}
\noindent As before, by $\emb^{\lor-}$ we denote the embedding
$\embdis$ \emph{without} the $\UNA$ axioms. Note that the observations about
modularity of the embeddings \emb extend to the disjunctive
embeddings \embdis; the $\PIA$s do not compromise modularity. However, polynomiality
of embeddings with $\PIA$s
is lost if the size of $\lang_{ga}$ is not polynomial in the size of \lp.
We do not write the $\UNA$
and $\PIA$ axioms explicitly in the examples below.

For standard AEL and ground disjunctive logic programs, the
correspondence between the stable models of $P$ and the stable
expansions $\embhpdis(P)$ and $\embehdis(P)$, respectively, is
due to \citeN{przymusinski91:_stabl_seman_disjun_progr} and
\citeN{marek93:_reflex_autoep_logic_logic_progr}.

\begin{proposition}
\label{prop:faithful-dis}
  A Herbrand interpretation $M$ of a ground disjunctive logic
  program \lp is a stable model of \lp iff there is a
  consistent stable expansion $T$ of $\embhp^{\lor-}(\lp)$
  {\rm (}resp., $\embeh^{\lor-}(\lp)$\/{\rm )} in standard
  AEL such that $M=T_{\oga}$.
\end{proposition}

We generalize this result to the case of FO-AEL and non-ground
programs, and additionally for $\embebdis$, similar to the case of normal
programs.

\begin{lemma}
  \label{lem:fo-ael-dis-named-model}
Let \lp be a logic program, let $\mathrm{X} \in \{\mathrm{E}, \mathrm{A} \}$,
let $ T $  be a \xstable expansion of
  $\embdis(\lp)$, and let $ \alpha $ be an objective ground atom. Then, $\embdis(\lp)
  \xmodels_{T_{\oga}} \alpha$ iff
  $\embdis(\lp) \models_{T_{\oga}} \alpha$ under the standard names
  assumption. Moreover, $\embeh^{\lor-}(\lp)
  \allmodels_{T_{\oga}} \alpha$ iff $\embeh^{\lor-}(\lp)
  \models_{T_{\oga}} \alpha$.
\end{lemma}

  \begin{proof}
  ($ \Rightarrow $)  Trivial (cf.~the ``$ \Rightarrow $''
direction in
%the proof of
Lemma \ref{lem:named-model}).

 \noindent   ($ \Leftarrow $) The argument is a straightforward adaptation of
the argument in the ``$ \Leftarrow $'' direction in the proof of Lemma
\ref{lem:named-model}: simply replace the consequent $ h_1 \land \mop h_1 $ with
the disjunction $(h_1\land \mop h_1)\lor \dots \lor (h_l \land \mop h_l) $.
Furthermore, it is also easy to see that, as $ \inter $ and $ \inter'
$ agree on ground atomic formulas, if the $\PIA$ axioms are satisfied in $
\langle \inter, T_{\oga}\rangle $, then they are satisfied in
 $ \langle \inter', T_{\oga}\rangle $.
  \end{proof}

\begin{lemma}
  \label{lem:fo-ael-dis-reduction}
  Let \lp be a logic program and let
  $\mathrm{X} \in \{\mathrm{E}, \mathrm{A} \}$.  There exists a
\xstable expansion $T$ of $\embdis(\lp)$  iff
  there exists
  a \xstable expansion $T'$ of $\embdis(\gr{\lp})$
  with $T'_{\oga}=T_{\oga}$.
The same result holds for $\embeh^{\lor-}$  and stable$^\mathrm{A} $
expansions.
\end{lemma}

  \begin{proof}
    The proof is obtained from the proof of Lemma
\ref{lem:reduction} by replacing occurrences of \emb with \embdis and
using Lemma~\ref{lem:fo-ael-dis-named-model} in place of Lemma~\ref{lem:named-model}.
\end{proof}

\begin{theorem}
  \label{the:embedd-disj}
  Let \lp be a logic program and let
  $\mathrm{X} \in \{\mathrm{E}, \mathrm{A} \}$.
  A Herbrand interpretation $M$
  is a stable model of \lp iff there exists a consistent \xstable
  expansion $T$ of $\embdis(\lp)$  such that $M=T_{\oga}$.
The same result holds for $\embeh^{\lor-}$ and stable$^\mathrm{A} $ expansions.
\end{theorem}

  \begin{proof}
    The reduction of embeddability of non-ground programs in
    FO-AEL to ground logic programs in standard AEL follows from Lemma
\ref{lem:fo-ael-dis-reduction}.

Embeddability of $\gr{\lp}$
    using $\embhpdis$ and $\embehdis$ follows from
    Proposition \ref{prop:faithful-dis}.
Embeddability of $\gr{\lp}$ using $\embebdis$ then follows
    from the embeddability of $\embhpdis$, combined with the $\PIA$ axioms
    ($ \alpha \limpl \mop \alpha $): consider a formula
    $ \bigwedge(b_i \land \mop b_i)\land \bigwedge(\lnot \mop c_j)\limpl \bigvee h_k $
    in $\embebdis(\gr{\lp})$ and some $b_i$.  If some model
    $ \langle \inter, \Gamma \rangle $ of $\embebdis(\gr{\lp})$ satisfies $b_i$,
    then $ \mop b_i $ must also be satisfied in
    $ \langle \inter, \Gamma \rangle $ (by the $\PIA$ axioms). Therefore, the
    stable expansions of $\embhpdis(\lp)$ and $\embebdis(\lp)$ must be the same.
  \end{proof}
A notable difference between the embeddings \embhpdis and \embebdis,
on the one hand, and $\embehdis$, on the other, is the presence, respectively
absence, of the $\PIA$ axioms, as illustrated in the following example.

\begin{example}
\label{ex:pia}
Consider $\lp = \{p \!\lpor\! q \leftarrow\; \}$.
Then,
  $\embhpdis(\lp) = \{p \lor q\} \cup \PIA_{\Sigma_\lp}$ has the stable expansions $T_1^{HP} = \{p, \neg q, \mop
p,\lnot
  \mop q, \dots\}$ and $T_2^{HP} = \{q, \neg p, \mop q,\lnot \mop p,
\dots\}$, while
 $\embehdis(\lp) = \{\, (p \land \mop p) \lor (q \land \mop q) \,\}$ has the stable expansions $T_1^{EH} = \{p,\mop p, \lnot \mop q,$ $
\dots\}$ and $T_2^{EH} = \{q, \mop p,\lnot \mop p,$ $ \dots\}$; the latter
include neither $\neg q$ nor $\neg p$.
\end{example}
Note that the embedding $\embhp$ cannot be naively
extended to logic programs with strong (``classical'') negation
\snot \cite{Gelfond+Lifschitz-ClasNegaLogiProg:91}, even for the
propositional case.  Take, for example, the logic program $\lp = \{p
\leftarrow \snot p\}$; it has
one stable model, namely $M = \emptyset$. The naive extension of
\embhp treats strong negation as negation in classical logic and
the embedding of \lp yields $\{\neg p \limpl p\}$, which has one
stable expansion, which includes $p$. It was shown by
\citeN{marek93:_reflex_autoep_logic_logic_progr} that, for the
propositional case, the embeddings $\embeb$ and $\embeh$ \emph{can} be
naively extended to the case of logic programs with strong negation:
consider a rule of the form (\ref{eq:rule}) such that
$h_i,b_j,c_k$ are either atoms or
strongly negated atoms, and an extension of the embeddings
$\embeb,\embeh$ such that $\snot$ is translated to
classical negation $\neg$; then, Proposition \ref{prop:prop-emb}
straightforwardly extends to these extended versions of
$\embeb$ and $\embeh$~\cite{marek93:_reflex_autoep_logic_logic_progr}. These
results can be straightforwardly extended to the non-ground case.
Embedding of logic programs with strong negation using $\embhp$
can be done by rewriting \lp to a logic program $\lp'$ without strong
negation and subsequently embedding $ \lp' $; see \cite{Gelfond+Lifschitz-ClasNegaLogiProg:91} for such
a rewriting.

\section{Relationships between the Embeddings}
\label{sec:prop-embedd}

In this section, we explore correspondences between the embeddings
presented in the previous section.
We compare
the stable expansions of the
individual embeddings and, at the level of inference, we compare the
sets of autoepistemic consequences.
To this end we introduce the following
notation:

\begin{definition}
Let $\fot_1,\fot_2\subseteq \fmodl$ be FO-AEL theories and $\mathrm{X}\in \{\mathrm{E},\mathrm{A}\}$.  We write $\fot_1\xexpequiv
\fot_2$
if $\fot_1$ and $\fot_2$ have the same \xstable expansions.
For $\gamma \in \{o, \og, \oga \}$ we write $\fot_1\xexpequiv_{\gamma} \fot_2$
if, for each \xstable expansion $T$ of $\fot_1$, there exists some
\xstable expansion $T'$ of $\fot_2$ such that $T_\gamma =
T'_\gamma$, and vice versa.
\end{definition}
Note the implication chain $\fot_1\xexpequiv \fot_2$ $\Rightarrow$
$\fot_1\xexpequiv_{o} \fot_2$
%, $\fot_1\xexpequiv_{o} \fot_2$ implies
$\Rightarrow$ $\fot_1\xexpequiv_{\og} \fot_2$
%, and $\fot_1\xexpequiv_{\og} \fot_2$ implies
$\Rightarrow$ $\fot_1\xexpequiv_{\oga} \fot_2$.
\begin{definition}
A formula $\phi$ is an \emph{autoepistemic$^\mathrm{X}$ consequence} of a theory
 $\fot\subseteq \fmodl$,
$\mathrm{X}\in \{\mathrm{E},\mathrm{A}\}$, if
$\phi$ belongs to every \xstable expansion of $\fot$. $Cn^{\mathrm{X}}(\fot)$
denotes the set of all autoepistemic$^\mathrm{X}$ consequences of $\fot$.
\end{definition}
The properties stated in this section holds regardless of
whether $ \mathrm{X}=\mathrm{E} $ or $ \mathrm{X}=\mathrm{A} $
is considered.
Therefore, we omit the superscript $\mathrm{X}$ from $ \xmodels$, $\xexpequiv_\gamma$, $ Cn^{\mathrm{X}}$, \xstable, etc.
Furthermore, we write $Cn_\gamma(\fot)$ for
$Cn(\fot)_\gamma$ (= $Cn(\fot)\cap {\fmodl}_\gamma$).

In our analysis, we consider different classes of logic programs.  With the
symbols
$\arbclassprg$,  $\safeclass$, and $\groundclass$ we denote the classes of
arbitrary, safe, and ground disjunctive logic programs, respectively.   Observe the
following inclusions between the classes:
$$
\groundclass\subseteq  \safeclass\subseteq \arbclassprg.
$$
We use
the letter $ n $ to denote the restriction of the respective classes to the case of
normal programs: $\narbclassprg$,  $\nsafeclass$, and $\ngroundclass$.

We start in Section~\ref{sec:relat-betw-stable} with an investigation of the correspondences between stable expansions and subsequently consider in Section~\ref{sec:relat-betw-models} correspondences
between sets of consequences.
Note
that while $\fot_1 \expequiv_\gamma \fot_2$ implies
$Cn_\gamma(\fot_1)=
Cn_\gamma(\fot_2)$,
the converse is not true in general.  Thus, for
applications based on consequence rather than stable expansions,
more flexibility between the choice of equivalent embeddings
can be expected as one-to-one correspondence between
stable expansions is not required.
In order not to interrupt the flow of reading, the proofs of
most of the results in this section can be found in the appendix.

\subsection{Relationships between Stable Expansions of Embeddings}
\label{sec:relat-betw-stable}

%In this section, we present
%%% our
%results on the relations between
%stable expansions of embeddings $\emb(\lp)$, where \lp is a logic
%program.
%
From Theorems \ref{the:embedd-stable} and \ref{the:embedd-disj}
we immediately obtain the following result concerning
correspondence of stable expansions, which is our main result
in this regard.

\begin{theorem}
\label{cor:corr-ground-atomic}
For every $ \lp\in\arbclassprg$,  $\embop(\lp)\expequiv_{\oga}\embop'(\lp)$ for all
$\embop, \embop'\in\{\embhpdis, \embebdis,$ $\embehdis\}$, and
if $ \lp\in\narbclassprg$, then
$\embop(\lp)\expequiv_{\oga}\embop'(\lp)$ for all
$\embop,\embop'\in\{\embhp, \embeb, \embeh, \embhpdis,$ $\embebdis, \embehdis\}$.
\end{theorem}
Thus, all embeddings may be used interchangeably when concerned with
ground
atoms.  This does not hold for the case of arbitrary objective ground
formulas.
\begin{example}
\label{ex:HP-EB}
 Consider the logic program $\lp=\{a \leftarrow b\}$. Then
$\embhp(\lp) = \{b \limpl a\}$ has a single stable expansion, which
contains $b \limpl a$; also $\embeb(\lp) = \{b \land \mop b \limpl
a\}$ has a single stable expansion, but it does not contain $b \limpl
a$.  Note that while the latter contains $\mop b \limpl b$, it does
not contain $ b \limpl \mop b$ (which would enable obtaining $b
\limpl a$).
\end{example}
The situation changes for the embeddings \embhpdis and
\embebdis{}
%, because of
due to the $\PIA$ axioms.

\begin{proposition}
  \label{prop:corr-ground-formulas}
For every $ \lp \in\narbclassprg$,
$\embeb(\lp)\expequiv_{og}\embeh(\lp)$,
and for every $ \lp \in\arbclassprg$,
$\embhpdis(\lp)\expequiv_{og}\embebdis(\lp)$.
\end{proposition}
For non-ground formulas we obtain the following result.

\begin{proposition}
  \label{prop:corr-safe-formulas}
For every $ \lp\in \nsafeclass$, $\embeb(\lp)\expequiv\embeh(\lp)$.
\end{proposition}
%To see that the embeddings $\embeb$ and $\embeh$ differ
%for arbitrary normal programs, consider
For arbitrary normal programs,
the embeddings $\embeb$ and $\embeh$ differ.

\begin{example}
\label{ex:EB-EH-nlp} Consider $\lp = \{p(a);p(x); q(x)
\leftarrow p(x)\}$. Then, the embedding $\embeh(\lp) = \{p(a)\land
\mop p(a), $ $\forall x\st p(x)\land \mop p(x),$ $\forall x\st
p(x) \land \mop p(x)\limpl q(x) \land \mop q(x)\}$ has one
%consistent
stable expansion, which contains $\forall x\st q(x)$, while
$\embeb(\lp) = \{p(a), $ $\forall x\st p(x), \forall x\st p(x) \land \mop
p(x)\limpl q(x)\}$ has one
%consistent
stable expansion, which does not contain $\forall x\st q(x)$,
%\comment{TE: I hided ``consistent'' 2 times. Wouldn't, strictly
%speaking, $\embeh(\{p(x)~\}) = \{\forall x(p(x)\land \mop
%p(x))\}$ without constants in the signature have only an inconsistent
%expansion? Jb: I updated the example to include a constant, because
%we always require a constant in the signature of a program.}
because $\forall x\st \mop p(x)$ is not
necessarily true when $\forall x\st p(x)$ is true; in other words, the converse
Barcan formula ($\mop\forall x\st \phi(x)\limpl \forall x\st \mop \phi(x)$) is not
universally valid, which is a property
of FO-AEL under both the any- and all-name semantics~\cite{konolige91:_quant}.
\end{example}
Note that the result also does not extend to the  embeddings \embhp
and \embhpdis.

\begin{example}
\label{ex:HP-HPv}
Consider
%the logic program
$\lp = \{q(x) \leftarrow p(x)\}$. Then, $\embhp(\lp)=\{\forall
x\st p(x) \limpl q(x)\}$ has one stable expansion, which
contains $\forall x\st p(x)\limpl q(x)$, while
$\embeb(\lp)=\{\forall x\st p(x) \land \mop p(x) \limpl q(x)\}$
has one stable expansion which does not contain $\forall x\st
p(x) \limpl q(x)$. This difference is caused by the fact that
$\mop p(x)$ will be false in case an unnamed individual is
assigned to $x$. Similar observations hold for $\embhpdis$; the $\PIA$ axioms do
not help, since they are only concerned with ground atoms and
thus do not apply to unnamed individuals.
\end{example}
\begin{proposition}
  \label{prop:corr-ground-programs}
If $ \lp\in \groundclass $, then
  $\embhpdis(\lp)\expequiv\embebdis(\lp)$.
\end{proposition}
Note that this result does not extend to the embedding \embehdis;
it does
not include the $\PIA$ axioms, and thus the argument used in the proof of
Proposition~\ref{prop:corr-ground-programs} does not apply.

\subsection{Relationships between Consequences of Embeddings}
  \label{sec:relat-betw-models}

In order to investigate the relationships between the embeddings
with respect to autoepistemic consequences, we first compare the
embeddings with respect to their autoepistemic models. Recall that
an autoepistemic interpretation $\langle \inter,T\rangle$ consists of a
first-order interpretation $\inter$ and a belief set $T\subseteq
\fmodl$.

\begin{proposition}
\label{prop:models-norm} For every $\lp\in\narbclassprg$ and every interpretation
\ainter,
$\inter \models_T \embeh(\lp)$ implies $\inter \models_T
\embeb(\lp)$ and $\inter \models_T \embhp(\lp)$ implies $\inter
\models_T \embeb(\lp)$.
 \end{proposition}

\begin{proposition}
\label{prop:models-disj} For every $\lp\in\arbclassprg$
and every interpretation \ainter, $\inter \models_T \embhpdis(\lp)$ implies
$\inter\models_T \embebdis(\lp)$.
Furthermore, if $\lp$ is safe, then $\inter \models_T \embebdis(\lp)$ implies $\inter \models_T
\embehdis(\lp)$.
  \end{proposition}

%\begin{proposition}
%\label{prop:models-safe-disj} Given a {safe} disjunctive logic
%program \lp, a first-order interpretation $\inter$, and a belief set
%$T$, if $\inter \models_T \embebdis(\lp)$ then $\inter \models_T
%\embehdis(\lp)$.
%\end{proposition}

%\begin{proof}
%  As in the proof of Proposition \ref{prop:corr-safe-formulas}, we can
%  restrict our attention to the case of named individuals. Now,
%  because of the PIA axioms, we have that, for any interpretation
%  $\langle \inter,T \rangle$ such that $\inter \models_T \embebdis(\lp)$,
%variable assignment $B$,  associated
%  name substitution \varsub, and objective atomic formula $\alpha$,
%whenever $(\inter,B) \models_{T} \alpha$,
%  $\inter \models_T \mop \alpha\varsub$. Therefore, $\inter \models_T
%\embehdis(\lp)$.
%\end{proof}

We now consider the relative behavior of the embeddings with respect
to autoepistemic consequences. In order to present our results in a compact and
accessible way, we show a small (yet sufficient) number of relationships between the sets
of consequences in a graph (Figure~\ref{fig:rel-cons}).  Every
particular relationship between embeddings can be easily derived from
paths in this graph.

Specifically, in Figure~\ref{fig:rel-cons}(a), $C_\chi^{(\vee)}$
is short for $Cn_o(\emb^{(\vee)}(\lp))$, the straight arrow
$\longrightarrow$ represents set inclusion ($\subseteq$), and
the dotted arrow $\dashrightarrow$ represents set inclusion in
case \lp is safe. Since $ \longrightarrow $ implies $
\dashrightarrow$, dotted arrows are only shown if straight
arrows are absent. Similarly, in Figure~\ref{fig:rel-cons}(b),
$C_\chi^{(\vee)}$ is short for $Cn_{\og}(\emb^{(\vee)}(\lp))$,
and $\longrightarrow$ represents set inclusion.

The main results visible from Figure~\ref{fig:rel-cons} are
that with respect to all objective consequences, $\embeb$ is
the weakest embedding (yielding a smallest set of conclusions)
while $\embhpdis$ and $\embehdis$ are strongest; if the
embedded program is safe, then $\embehdis$ is the strongest
embedding and $\embeb$ the weakest, collapsing with $\embeh$
and $\embehdis$. With respect to ground objective consequences,
$\embhpdis$ collapses with $\embebdis$ and is the strongest
embedding, while $\embeb$ is the weakest and again collapses
with $\embeh$ and $\embehdis$; safety of the program does not
change the picture. Note that with respect to objective ground
atomic consequences, all embeddings collapse (cf.\
Theorem~\ref{cor:corr-ground-atomic}).

The following lemma states the correctness of
Figure~\ref{fig:rel-cons}(a).

\begin{lemma}
\label{lem:rel-cons} If $C_\chi^{(\vee)} \,{\longrightarrow}\,
C_\gamma^{(\vee)}$ {\rm (}resp., $C_\chi^{(\vee)}
\,{\dashrightarrow}\, C_\gamma^{(\vee)}${\rm )} in the graph of
Figure~\ref{fig:rel-cons}{\rm (}a{\rm )}, then $Cn_o(\emb^{(\vee)}(\lp))
\,{\subseteq}\, Cn_o(\tau_\gamma^{(\vee)}(\lp))$ for every $\lp
\,{\in}\, \narbclassprg$
{\rm (}resp., for every $\lp \,{\in}\, \nsafeclass${\rm )}.
Furthermore, if $C_\chi^{\vee}
\,{\longrightarrow}\,
C_\gamma^{\vee}$ {\rm (}resp., $C_\chi^{\vee} \,{\dashrightarrow}\,
C_\gamma^{\vee}${\rm )}, then $Cn_o(\emb^{\vee}(\lp)) \,{\subseteq}\,
Cn_o(\tau_\gamma^{\vee}(\lp))$ for every $\lp \,{\in}\, \arbclassprg$
{\rm (}resp., for  every $\lp \,{\in}\, \safeclass${\rm )}.
\end{lemma}
Note that by transitivity of $\subseteq$, paths in the graph yield
further relations; e.g., $Cn_o(\embeb(\lp)) \subseteq
Cn_o(\embhpdis(\lp))$ since $C_{EB}$ reaches  $C_{HP}^\vee$
via a path with straight edges.

We now show that the graph exactly
characterizes the containment relationships via paths. To this end, we
first note some negative relationships between embeddings.

\begin{lemma}
\label{lem:not-rel-cons}
The following inclusion relations do not hold: $Cn_o(\tau_{EH}^{\vee}(\lp)) \subseteq
Cn_o(\tau_{HP}^{\vee}(\lp))$, for every $\lp\in\arbclassprg$;
$Cn_o(\tau_{EB}^{\vee}(\lp)) \subseteq
Cn_o(\tau_{HP}(\lp))$, for every $\lp\in\nsafeclass$; and
$Cn_o(\tau_{HP}(\lp)) \subseteq
Cn_o(\tau_{EB}^{\lor}(\lp))$, for every $\lp\in\nsafeclass$.
\end{lemma}
From these negative relationships, combined with the positive
ones above, we can infer further negative relationships. For
example, from $C^\lor_{EH} \not\subseteq C^\lor_{EB}$ and
$C^\lor_{EH} \subseteq C_{EH}$, we infer $C_{EH} \not\subseteq
C^\lor_{EB}$. Exploiting this, we show the following result.

\nop{******** old figure, submit version *******
\begin{narrowfig}{0.52\linewidth}
      $$\xymatrix{
        & \ar[r] C_{EB}^\vee & C_{HP}^\vee   \\
        C_{EH}^\vee \ar@/^/[r] \ar@{.>}[ur] & \ar@/^/[l] C_{EH}
        \ar@/^/@{.>}[r]  & \ar@/^/[l] \ar[ul] C_{EB}
        \ar[r] & C_{HP} \ar[ul]
      }$$
\caption{Containment relationships between sets of objective consequences.}
\label{fig:rel-cons}
\end{narrowfig}
******}
\begin{figure}

\centerline{$
\begin{array}{@{}c@{\qquad}c@{}}
\xymatrix{
         \ar[r] C_{EB}^\vee & C_{HP}^\vee   \\
        C_{EH}^\vee = C_{EH}\ar@{.>}[u] \ar@/^/@{.>}[r]   & \ar@/^/[l] \ar[lu] C_{EB}
        \ar[r] & C_{HP} \ar[lu]
      }
&
\xymatrix{
         C_{EB}^\vee = C_{HP}^\vee &  \\
        C_{EH}^\vee = C_{EH} = C_{EB} \ar[u] \ar[r] & C_{HP} \ar[lu]
      }
\\
~\\
(\mbox{a}) &  (\mbox{b})
\end{array}
$}

\caption{Inclusion between sets of (a) all objective
  and (b) all objective ground consequences.}
\label{fig:rel-cons}
\end{figure}
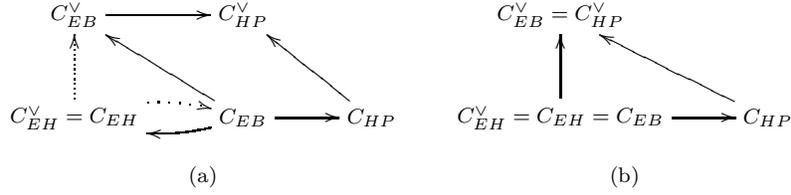

%\begin{narrowfig}{0.52\linewidth}
%      $$\xymatrix{
%        & %\ar@/^/@{.>}|-{\crosssymc}[dl]
%        \ar@/^/@{.>}|@{/}[dl]
%        C_{EB}^\vee\ar@/^/@{.>}|@{/}[d]\ar@{.>}|@{/}[dr] & \ar@{.>}|@{/}[l] C_{HP}^\vee \ar@{.>}|@{/}[d] \ar@{.>}|@{/}[dr]  \\
%        C_{EH}^\vee  \ar@/^/@{->}|@{/}[ur] & C_{EH}
%        \ar@/^/@{->}|@{/}[u] \ar@{->}|@{/}[r]  &   C_{EB}
%         & \ar@{.>}|@{/}[l] C_{HP}
%      }$$
%  \caption{Negative relationships between sets of consequences.}
%  \label{fig:not-rel-cons}
%\end{narrowfig}

\begin{theorem}
\label{th:cons-og} For $\lp\,{\in}\,\narbclassprg$ {\rm (}resp.,
$\lp\,{\in}\,\nsafeclass${\rm )}, $Cn_o(\emb^{(\vee)}(\lp))
{\subseteq} Cn_o(\tau_\gamma^{(\vee)}(\lp))$ iff
$C^{(\lor)}_\gamma$ is reachable from $C^{(\lor)}_\chi$ in the
graph in Figure~\ref{fig:rel-cons}{\rm (}a{\rm )} on a path with
$\longrightarrow$ arcs {\rm (}resp., with arbitrary arcs{\rm )}.
% imply $\dashrightarrow$ arcs).\\

Likewise, for  $\lp\,{\in}\,\arbclassprg$ {\rm (}resp., $\lp\,{\in}\,\safeclass${\rm )}, we have that $Cn_o(\emb^{\vee}(\lp)) \,{\subseteq}\,
Cn_o(\tau_\gamma^{\vee}(\lp))$ iff $C^{\lor}_\gamma$ is reachable from $C^{\lor}_\chi$
 on a path with $\longrightarrow$ {\rm (}resp.,  with arbitrary arcs{\rm )}.
\end{theorem}
\begin{proof}%[of Theorem \ref{th:cons-og}]
By Lemmas~\ref{lem:rel-cons} and \ref{lem:not-rel-cons}, the
respective containment relationships are correct. Clearly,
by reflexivity and transitivity of set
inclusion, paths in the graph of
Figure~\ref{fig:rel-cons}(a) are sound with respect to positive
containments. Their completeness, for both arbitrary $\lp$ and
safe $\lp$, is established using the following basic properties of
non-inclusion: (i) $A \not\subseteq B$ and $C\subseteq
B$ implies $A\not\subseteq C$ and (ii) $A\not\subseteq B$ and
$A\subseteq C$ implies $C\not\subseteq B$.\footnote{Note that
non-inclusion for normal programs implies non-inclusion for
disjunctive programs, since every normal program is \emph{a fortiori} a
disjunctive program.} Exhaustive application to the
%positive and negative
(non-)containments in
Lemmas~\ref{lem:rel-cons} and \ref{lem:not-rel-cons}  (e.g., using a simple logic program)
yields one of $C^{(\lor)}_\chi \subseteq
C^{(\lor)}_\gamma$ and $C^{(\lor)}_\chi \not\subseteq C^{(\lor)}_\gamma$ for each pair
$C^{(\lor)}_\chi$, $C^{(\lor)}_\gamma$.
\end{proof}
Accordingly, $C_{EB} \subseteq C^\lor_{EH}$ and
$C_{EB}\subseteq C^\lor_{HP}$ are the only nontrivial
inclusions for arbitrary programs besides those in
Figure~\ref{fig:rel-cons}(a);
%\comment{TE: I added ``objective'' in the caption of
%Figure~\ref{fig:rel-cons}.a to make clear what consequences.}
for safe programs, there are more.
We note that the figure is minimal, in the sense that if any of the arcs
is removed (or turned from solid into dashed), the theorem no longer
holds.

The containment relationships in Figure~\ref{fig:rel-cons}(b)
are easily obtained from the already established results.
\begin{theorem}
\label{th:cons-oga} For every $\lp\,{\in}\,\narbclassprg$,
% it holds that
$Cn_\og(\emb^{(\vee)}(\lp)) \subseteq
Cn_\og(\tau_\gamma^{(\vee)}(\lp))$ holds iff $C^{(\lor)}_\gamma$ is
reachable from $C^{(\lor)}_\chi$ in the graph in
Figure~\ref{fig:rel-cons}{\rm (}b{\rm )}.
Furthermore, for every $\lp\,{\in}\,\arbclassprg$,
$Cn_\og(\emb^{\vee}(\lp)) \subseteq
Cn_\og(\tau_\gamma^{\vee}(\lp))$ holds iff $C^{\lor}_\gamma$ is
reachable from $C^{\lor}_\chi$.
\end{theorem}
\begin{proof}%[of Theorem \ref{th:cons-oga}]
By Theorem~\ref{th:cons-og} and
Proposition~\ref{prop:corr-ground-formulas}, $C^\lor_{EH}$,
$C_{EH}$, and  $C_{EB}$ collapse, and by
Proposition~\ref{prop:corr-ground-formulas}, $C^\lor_{EB}$ and
$C^\lor_{HP}$ collapse.
That the other relationships in
Figure~\ref{fig:rel-cons}(a) remain unchanged follows from
Example~\ref{ex:HP-EB}, the proof of
Lemma~\ref{lem:not-rel-cons}, and reasoning about subsets as in
the proof of Theorem~\ref{th:cons-og}.
\end{proof}
\nop{*** With no effort, we can add a second figure that shows
the inclusions of objective ground (og) consequences, so that
we have a complete picture for (1) oga (all collapse, by
Theorem ~\ref{cor:corr-ground-atomic}), (2) o
(Figure~\ref{fig:rel-cons}) and og:  By Theorem~\ref{th:cons-og}
and Proposition~\ref{prop:corr-ground-formulas}, $C^\lor_{EH}$,
$C_{EH}$, and  $C_{EB}$ collapse, and by
Proposition~\ref{prop:corr-ground-formulas} $C^\lor_{EB}$ and
$C^\lor_{HP}$ collapse; that the other relationships in
Figure~\ref{fig:rel-cons} remain unchanged follows from
Example~\ref{ex:HP-EB}, the proof of
Lemma~\ref{lem:not-rel-cons}, and reasoning about subsets as in
the proof of Theorem~\ref{th:cons-og}. ***}

\section{Combinations with First-Order Theories}
\label{sec:prop-comb}

In this section we explore correspondences between the logic program
embeddings from Section~\ref{sec:embedd-logic-progr} in combinations
with FO theories. To this end, we consider a basic combination of
logic programs \lp and FO theories \fot defined as
\[
\combop_\chi^{(\lor)}(\fot,\lp)=\fot\ \cup\
\embgen(\lp) \subseteq \fmodl
%\footnote{One could imagine other, more involved,
%embeddings of the classical theory.  Such embeddings
%are a topic for future investigations.}
\]
where $\signature_{\fmodl}$
is the union of the signatures $\signature_\fot$ and
$\signature_\lp$. More involved combinations
(e.g., which augment $\lp$ and \fot with further rules and axioms,
respectively) might be recast to such basic combinations.

In the preceding sections we have considered both the any- and all-name
semantics, both in the definition of the embeddings and in our analysis of the
differences between the embeddings of logic programs.  It turned out that the
embeddings are faithful for both semantics (cf.~Theorems
\ref{the:embedd-stable} and \ref{the:embedd-disj}), implying correspondence
with respect to objective ground atoms between the two semantics for all
embeddings \embgen, and the relationships between the embeddings
stated in the
previous section hold for both semantics.  However,
in combinations with FO theories, the two semantics diverge since
names from the first-order part may not be provably identical to or
different from other names.
The following  example illustrates  differences between the semantics in the
face of positive and negative occurrences of the modal operator.

\begin{example}
\label{ex:any-all-combination-2} Consider the logic program \lp:
\begin{align*}
q(a),\\
r & \leftarrow p(x), \dnot q(x),\\
s(x) & \leftarrow p(x),
\end{align*}
and the FO theory $ \fot=\{p(b)\} $.  We note here that the signature of \lp
contains only one function symbol, the constant constant $ a $. Consequently, $
\UNA_{\signature_\lp}=\emptyset $.

$ \combeb(\fot, \lp)$ has one \anystable expansion $
T^\anysymbol $ and one \allstable expansion $ T^\allsymbol $. $
T^\anysymbol $ contains $ q(a)$, but not $ q(b)$; both contain
$ p(b) $, but not $ p(a)$.  Consider an interpretation
\interdef such that $ a^\inter=b^\inter=k $, $ k \in
p^\funsym,k \in q^\funsym, k\in s^\funsym $, and $ \inter \not
\models r $, and a variable assignment $ B $ such that $ x^B=k
$.  Then, $ \varsub=\{x/a \} $ is an associated name
substitution and $ q(x)\varsub \in T^\anysymbol $, and so $
(\inter, B) \anymodels_{T^\anysymbol} \mop q(x)$. Another
associated name substitution is $ \varsub=\{x/ b \} $, and so $
(\inter, B) \anymodels_{T^\anysymbol} \mop p(x)$.  So, $\inter
\anymodels_{T^\anysymbol} \combeb\combpar$ and thus $ r \notin
T^\anysymbol $. One can straightforwardly argue that
$\combeb\combpar\anymodels_T s(b)$, and therefore $s(b)\in
T^\anysymbol $.

Consider an interpretation $\inter'$ that is like \inter,
except that $ k \not \in s^\funsym $ and $ \inter' \models r $.
Since $ T^\allsymbol $ does not contain $ p(a) $, $ (\inter',
B) \not\allmodels_{T^\allsymbol} \mop p(x)$. $ T^\allsymbol $
does not contain $ q(b)$, but $(\inter, B)
\allmodels_{T^\allsymbol} \mop q(x)$  holds only if $
q(x)\varsub \in T^\allsymbol$ for \emph{every} \varsub
associated with $ B $, including $ \varsub=\{x / b \} $, and so
$(\inter', B) \allmodels_{T^\allsymbol}\lnot\mop q(x)$ and
$\inter' \allmodels_{T^\allsymbol} \combeb\combpar$.
Consequently, $s(b)\not\in T^\allsymbol $. It is
straightforward to verify that $ r \in T^\allsymbol $.
\end{example}
In order to avoid a proliferation of results, following
\citeN{konolige91:_quant}, we concentrate in this section on the
any-name semantics. In the following section we discuss our results in the light
of the standard names assumption, for which the any- and all-name semantics coincide.

In our analysis we consider the same syntactic
classes of programs as in the previous section and we consider the following
classes of objective theories:
\begin{quote}
% \item
% arbitrary ($\arbclassprg$), safe ($\safeclass$), and
% ground logic programs ($\groundclass$);
%
arbitrary ($\arbclassth=2^\lang $), universal ($\uniclass$),
Horn ($\hornclass$), generalized Horn
($\ghornclass)$,
propositional ($\propclass$), and empty ($\emptyclass$, i.e.,
in semantic terms, tautological) FO theories.
\end{quote}
Generalized Horn formulas have the form $(\forall)
b_1\land\cdots\land b_n\limpl \exists \vec{y}.h$ where all
$b_i$ and $h$ are atomic and variables $\vec{y}$ occur only in
$h$, which may be absent. Note that similar formulas $(\forall)
b_1\land\cdots\land b_n\limpl \exists \vec{y}.(h_1 \land \dots
\land h_m)$ where all $b_i$ and $h_j$ are atomic, can be easily
cast to \ghornclass formulas by replacing the conjunction $h_1
\land \dots \land h_m$ with $p(\vec{x})$, where $\vec{x}$ are
its free variables and $p$ is a new predicate symbol, and
including formulas $\forall \vec{x}.p(\vec{x}) \limpl h_j$ for
each $h_j$; the resulting theory is equivalent with respect to
the original signature. The class \ghornclass captures RDF
Schema \cite{Bruijn-Heymans-LogiFoun-07}, DLs such as
Horn-$\mathcal{SHIQ}$ \cite{Hustadt05datacomplexity} and the
OWL~2 profiles EL, QL, and RL, as well as Tuple Generating
Dependencies (TGDs), which play an important role in relational
databases; \citeN{Cali+GottlobETAL-:09} survey decidable
fragments of TGDs that amount to decidable fragments of
\ghornclass (see also Sections~\ref{sec:SW-application} and
\ref{sec:lp-dl-extensions}).

The order diagram is as follows (arrows stand for set inclusion):
$$
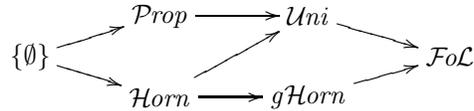
\xymatrix @R=0mm @H=3mm{
% & & &
 & \propclass\ar[r] & \uniclass\ar[dr]\\
% \groundclass\ar[r] & \safeclass\ar[r] & \arbclassprg ; &
\emptyclass \ar[ur] \ar[dr]
&  & & \arbclassth\\
% & & &
 & \hornclass\ar[r]\ar[uur] & \ghornclass\ar[ur]\\
}
$$
For all pairs of classes of logic programs and FO theories, we
determine the relationships between stable expansions of different
combinations $\combop_\chi^{(\lor)}(\fot,\lp)$ and
$\combop_\gamma^{(\lor)}(\fot,\lp)$ at different levels of
granularity.
As in Section~\ref{sec:relat-betw-stable}, we
concentrate here on correspondences of stable expansions $\anyexpequiv_\mathrm{x} $%
; they imply that relative to the class $\mathrm{x}$ of formulas, the embeddings $\embop_\chi^{(\lor)}(\lp)$ and
$\embop_\gamma^{(\lor)}(\lp)$ are
interchangeable in combinations.

In Section \ref{sec:relationships-combinations}, we state our main result on the
relationships between stable expansions of combinations and
make several observations. In Section \ref{sec:derivation-combinations}, we establish the
partial results necessary for deriving our main result.
The proofs of the partial results can be found in the appendix.

\subsection{Relationships between Stable Expansions of Combinations}
\label{sec:relationships-combinations}

Our results are summarized in Table~\ref{tab:summary-exp},
which gives a complete picture of the correspondences, where
each entry represents a {\em most general correspondence},
i.e., neither the correspondence $\anyexpequiv_\mathrm{x} $ nor
the logic program or FO theory class may be relaxed. This is
formally stated in the main theorem of this section (Theorem~\ref{th:fo-ael-exp-correspondence}).
In brief, our central results are that several of the embeddings become
interchangeable when considering positive normal programs
combined with \ghornclass or \hornclass theories (cf.\
Proposition~\ref{prop:positive-horn}) as well as the correspondences
for combinations with ground logic programs, even allowed to
contain negation (cf.\ the rightmost column of Table~\ref{tab:summary-exp}).

We call $\fot_1 \anyexpequiv_\mathrm{x}
\fot_2$ a \emph{trivial inference} from a set $Q$ of
equivalences if it is derivable from $Q$ by the fact that
$\fot_1\anyexpequiv \fot_2$ implies $\fot_1\anyexpequiv_{og} \fot_2$ and $\fot_1\anyexpequiv_{og} \fot_2$
implies $\fot_1\anyexpequiv_{\oga} \fot_2$, as well as by reflexivity, transitivity, and
symmetry of $\anyexpequiv_\mathrm{y}$, $\mathrm{y} \,{\in}\,\{\epsilon, \og, \oga \}$.

\begin{acmtable}{260pt}
\centering
\renewcommand{\arraystretch}{1.1}
  \begin{tabular}{|@{\ \mbox{}}c@{\ \mbox{}}|@{\ \mbox{}}c@{\ \mbox{}}|@{\
\mbox{}}c@{\ \mbox{}}|@{\ \mbox{}}c@{\ \mbox{}}|c@{\ \mbox{}}|}
%    \hline
     $\Phi$ $\;\backslash\;$ $\lp$
 & \arbclassprg & \safeclass &  \groundclass\\
    \hline
\hline\arbclassth &
$\combop_{EH}\anyexpequiv\combop_{EH}^\lor$ &   &
$\combop_{EB}\anyexpequiv\combop_{EH}$\\ &&&
$\combop_{HP}^\lor\anyexpequiv\combop_{EB}^\lor$\\
\hline\uniclass &
& & \\
\hline\ghornclass & & &
$\combop_{HP}\anyexpequiv_{\oga}\combop_{EH}$ \\
                  & & & $\combop_{HP}^\lor\anyexpequiv_{\oga}\combop_{EH}^\lor$\\
\hline \hornclass &
& & \\
\hline\propclass & $\combop_{HP}^\lor\anyexpequiv_{og}\combop_{EB}^\lor$
& $\combop_{EB}\anyexpequiv\combop_{EH}$ &\\
& $\combop_{EB}\anyexpequiv_{og}\combop_{EH}$
&  &\\
\hline $\emptyclass$ &
 \begin{tabular}{c}
    $\combop_{HP}\anyexpequiv_{\oga}\combop_{EB}$\\
    $\combop_{HP}\anyexpequiv_{\oga}\combop_{EH}$\\
    $\combop_{HP}^\lor\anyexpequiv_{\oga}\combop_{EH}^\lor$
\end{tabular}
     & &
  \end{tabular}
\caption{Correspondences between stable expansions of combinations
(on programs in the joint definition range);
  $\combop_x^{(\lor)}$ is short for
  $\combop_x^{(\lor)}(\fot,\lp)$.
% All correspondences hold for
%  normal programs and correspondences between disjunctive
%  combinations $\combop_x^\lor,\combop_y^\lor$ hold for
%  disjunctive programs.
  }
\label{tab:summary-exp}
\end{acmtable}

\begin{theorem}
\label{th:fo-ael-exp-correspondence}
Let $\mathcal{X}$ be a class of FO theories,
let $\mathcal{Y}$ be a class of programs, and let $\mathrm{x}\,{\in}\,
            \{\epsilon,$ $\og,$ $\oga \} $.  Then
$\combop_\chi(\fot, \lp) \anyexpequiv_\mathrm{x}
\combop_\gamma(\fot, \lp)$ holds for all $\fot\,{\in}\,\mathcal{X}$ and all $\lp\,{\in}\,
\mathcal{Y}$ iff
$\combop_\chi(\fot, \lp) \anyexpequiv_\mathrm{x}
\combop_\gamma(\fot, \lp)$  follows for cell $(\mathcal{X},
\mathcal{Y})$ in Table~\ref{tab:summary-exp} by trivial inferences,
where  $ \chi, \gamma \in \{\mbox{}_{HP}^\lor ,$ $\mbox{}_{EB}^\lor, \mbox{}_{EH}^\lor \}$ if $\lp\in\arbclassprg$ and $ \chi, \gamma \in \{\mbox{}_{HP}, \mbox{}_{EB},
\mbox{}_{EH},\mbox{}_{HP}^\lor , \mbox{}_{EB}^\lor, \mbox{}_{EH}^\lor \}$ if $\lp \,{\in}\, \narbclassprg$.
\end{theorem}
We will establish the results of Table~\ref{tab:summary-exp}
and provide some intuitive explanations about partial results
in the next subsection.

Note that removing any statement from Table
\ref{tab:summary-exp} or modifying any correspondence type
invalidates the theorem. We do not explicitly consider
correspondence of stable expansions with respect to objective
formulas, i.e., $ \anyexpequiv_o $.  Clearly, $
\fot_1\anyexpequiv \fot_2 $ implies $ \fot_1\anyexpequiv_o
\fot_2 $; in addition, all the counterexamples to $
\fot_1\anyexpequiv \fot_2 $ presented in  the following
subsection also apply to $ \fot_1\anyexpequiv_o \fot_2 $.
Hence, $\anyexpequiv_o$ coincides with $\anyexpequiv$.

%Most of the counterexamples in the proof of Proposition
%\ref{prop:counterexamples} use only positive programs. However,
%for the case (\ref{c5}) (where $ \fot \in \hornclass $), the
%counterexamples use negation.
The use of negation is essential for establishing
non-correspondence in some cases, as we have the following
result for positive programs and Horn theories.

\begin{proposition}
\label{prop:positive-horn} For every $(\fot, \lp) \in
\ghornclass\mytimes\nsafeclass $ such that \lp is positive,
$\combeb\combpar\anyexpequiv_{oga}\combeh\combpar$, and for
every $ (\fot, \lp) \in \hornclass\mytimes\narbclassprg $ such
that \lp is positive,
%% $\combhp\combpar\anyexpequiv_{oga}\combeb\combpar$ and
%% $\combhp\combpar\anyexpequiv_{oga}\combeh\combpar$.
$\combhp\combpar\anyexpequiv_{oga}\combeb\combpar
\anyexpequiv_{oga}\combeh\combpar$.
\end{proposition}
This result does not extend to the disjunctive embeddings,
because there are no \PIA axioms for the atoms involving names
not in \lp, and it does not extend to more general formulas
(cf., e.g., Example \ref{ex:contrapositive}).
%  on page \pageref{ex:contrapositive}).

\nop{**** hide (TE) ***** Remarkable are the correspondences
for the case of ground programs, as they show the effects of
grounding the program in the combination. In particular, the
embeddings \embeb and \embeh are interchangeable in
combinations of ground programs with arbitrary first-order
theories, and \embhp and \embeh are interchangeable in
combinations with \ghornclass theories. Also remarkable is the
fact that \embhp and \embeh are interchangeable in combinations
with \hornclass theories, when considering non-ground positive
programs. *****} We make the following further observations.

\begin{itemize}
\item The various combinations behave differently in the
    general case. Only two of them, $\mbox{}_{EH}$ and
    $\mbox{}^\lor_{EH}$, are always equivalent (they
    coincide on normal programs).
\item For combinations with arbitrary FO theories, further
    correspondences are only present for ground logic
    programs in some cases.  Narrowing to any of the
    classes that allow predicates of arity $> 0 $
    ($\hornclass$, $\ghornclass$, $\uniclass$) does not
    change the picture.

\item For arbitrary logic programs, only in case of
    propositional theories do some combinations behave
    equivalently. Requiring safety leads only for
    propositional theories and in one
  case ($\mbox{}_{EB}$ and $\mbox{}_{EH}$) to a stronger correspondence.
%; for
%  such an effect, in other cases a stronger condition like safety must
%  be applied

 \item $\uniclass$ and $\hornclass$ show no most general
     correspondences, which means that with respect to more
     general or more restrictive classes, their change in
     syntax does not affect equivalence.
 \item  In contrast, the important class $\ghornclass$ has
     maximal correspondences for ground programs. Thus, for
     combinations with FO theories from the classes
     $\hornclass$, $\ghornclass$, and $\uniclass$, we have
     equivalent behavior only for ground programs in some
     cases (apart from $\mbox{}_{EH}$ and
     $\mbox{}^\lor_{EH}$).

   \item The effect of program grounding (which is a
       customary technique applied by non-monotonic rule
       engines) in combinations is thus beneficial for the
       embeddings \embeb and \embeh, making them
       interchangeable; this is similar for \embhpdis and
       \embebdis and for \embhp,
 \embeb, \embeh, resp.\ \embhpdis, \embebdis, \embehdis, in
       combinations with a \ghornclass first-order part
       with respect to ground atomic formulas. In fact, the
       interchangeability of \embeb and \embeh lifts by
       Proposition~\ref{prop:positive-horn} to the
       non-ground level for \ghornclass FO parts and safe
       positive programs (which are at the core of rule
       bases in practice), with respect to ground atomic
       formulas. Importantly, the proof of
       Proposition~\ref{prop:positive-horn} shows that the
       combination is invariant under program grounding, so
       this technique can be readily applied. We get a
       similar picture for \embhp, \embeb, \embeh in
       combinations of \hornclass FO parts with positive
       programs. (For applications, see
       Section~\ref{sec:SW-application}.)

 \end{itemize}
We illustrate the use of the result in Theorem~\ref{th:fo-ael-exp-correspondence}
with an example. Note that if the stable expansions of two embeddings or combinations
correspond with respect to a certain class formulas, then
the embeddings, resp.\ combinations,
also agree on autoepistemic consequences for these classes.

\begin{example}
Consider
$\lp = \{q(a);$ $p(x);$ $r(x) \leftarrow
  \dnot s(x),$ $ p(x)\}$ from Example~\ref{ex:general}, which is
neither safe nor ground. Hence, to determine correspondence between
embeddings, we use the first column of Table~\ref{tab:summary-exp}. As \lp is normal, all equations in the column
are applicable. We have that, e.g.,
$\embeb(\lp)\anyexpequiv_{og}\embeh(\lp)$,
$\embhpdis(\lp)\anyexpequiv_{og}\embebdis(\lp)$, and $\embhp(\lp)\anyexpequiv_\oga\embeb(\lp)$.  Let \fot be a propositional
theory; then we also have
$\combeb\combpar\anyexpequiv_{\og}\combeh\combpar$ and
$\combhpdis\combpar\anyexpequiv_\og\combebdis\combpar$, but
not $\combhp(\lp)\anyexpequiv_\oga\combeb(\lp)$.
Furthermore, we can conclude that $\combeb(\lp)$ and $\combeh(\lp)$, and also
 $\combhpdis(\lp)$ and $\combebdis(\lp)$, agree on objective ground autoepistemic
consequences.
\end{example}

\subsection{Derivation of the Results}
\label{sec:derivation-combinations}

We start with the positive results. Trivially, $\combeh$ and
$\combehdis$ coincide for arbitrary FO theories, and  the
the equivalence results for empty \fot in Table~\ref{tab:summary-exp}
carry over from the respective results on embeddings in Section~\ref{sec:prop-embedd}.

We show that for ground programs, the $\embeb$ and
$\embeh$ embeddings are interchangeable in any combination with an FO
theory.

\begin{proposition}
\label{prop:eb-eh-safe}
For every $(\fot, \lp) \in \arbclassth\mytimes\ngroundclass $,
$\combeb\combpar\anyexpequiv\combeh\combpar$.
\end{proposition}
Intuitively, this holds because only named
individuals matter in rules, and hence the modal atoms $\mop h$ in embedded
rule heads do not matter. However, this does not generalize from
ground to safe programs, as the evaluation of literals $\lnot \mop
p(x)$ in the rule bodies does not amount to grounding (see Proposition~\ref{prop:counterexamples}(\ref{c5})).

Also the $\embhpdis$ and $\embehdis$ embeddings
are interchangeable in combinations with arbitrary FO theories if
the logic program is ground.

\begin{proposition}
\label{prop:hpdis-ebdis-ground}
For every  $(\fot, \lp) \in \arbclassth\mytimes\groundclass $,
$\combhpdis\combpar\anyexpequiv\combebdis\combpar$.
\end{proposition}
The reason is that we can eliminate all modal atoms $\mop b$ from
rule bodies with the $\PIA$ axioms in $\combebdis\combpar$ and obtain
$\combhpdis\combpar$.  Such elimination is not possible in the non-ground
case, since the $\PIA$s only apply to atoms from $ \signature_\lp $.

Moving now to fragments of $\arbclassth$, i.e., down the rows in
Table~\ref{tab:summary-exp}, we first have:

\begin{proposition}
\label{prop:hp-eh-ground}
For $(\fot, \lp) \in \ghornclass\mytimes\groundclass $,
$\combhpdis\combpar\anyexpequiv_{\oga}\combehdis\combpar$, and if
$\lp \in \ngroundclass$, then $\combhp
\combpar\anyexpequiv_{\oga}\combeh\combpar$.
\end{proposition}
Intuitively, in the first case, we can add
modal atoms $\mop b$ in the bodies and $\mop h$ in the
heads of $\embhp(\lp)$, by the
$\PIA$ axioms, thereby obtaining $\combehdis\combpar$. To go from \embeh
to \embhp is possible if $\fot$ is not disjunctive with respect to
atoms $h$. This is the case for a Horn $\fot$, and similarly for a
generalized Horn $\fot$, as we can apply skolemization. In the second
case, there are no $\PIA$ axioms, but we can similarly apply skolemization and
obtain a disjunction-free theory that is Horn modulo modal
atoms. Skolemization does not work for
non-ground programs in this case, as previously unnamed individuals are
named by skolem terms.
% For illustration: take $\fot = \{ \exists x\st p(x), q(a) \}$ and
% $\lp =\{r \implied p(x)\}$. Then for $\mbox{}_{HP}$ and
% $\mbox{}_{HP^\lor}$, we have $r$ in the single stable expansion, while for
% but not for $\mbox{}_{EH}$ and  $\mbox{}_{EH^\lor}$.
%

We note that, combined with previous results, we can infer from
Proposition~\ref{prop:hp-eh-ground} that
Theorem~\ref{cor:corr-ground-atomic} generalizes from embeddings to
combinations with generalized Horn theories for the case of ground
logic programs.

For propositional theories, we obtain a result
symmetric to Proposition~\ref{prop:hpdis-ebdis-ground} for arbitrary logic programs.

\begin{proposition}
\label{prop:hpdis-ebdis-prop}
For every $(\fot, \lp) \in \propclass\mytimes\arbclassprg $,
$\combhpdis\combpar\anyexpequiv_{og}\combebdis\combpar$.
\end{proposition}
Intuitively,
this holds because a propositional $\fot$ cannot interfere with names
of individuals---it has no names. Therefore, as in the case
of the embeddings $\embhpdis$ and $\embebdis$, we can eliminate %with the
%$\PIA$ axioms
all modal atoms $\mop b$ from rule bodies in
$\combebdis\combpar$ to obtain $\combhpdis\combpar$.  For similar
reasons, also correspondence results for the embeddings $\embeb$ and
$\embeh$ extend to combinations with propositional theories.

\begin{proposition}
\label{prop:eb-eh-prop-merged}
For every $(\fot, \lp) \in \propclass\mytimes\narbclassprg $,
$\combeb\combpar\anyexpequiv_{og}\combeh\combpar$, and if $\lp \in
\nsafeclass$, then $\combeb\combpar\anyexpequiv\combeh\combpar$.
\end{proposition}

It turns out that the results in the preceding propositions cannot be
extended to more general classes of programs or theories, or larger
subsets of stable expansions.

\begin{proposition}
\label{prop:counterexamples}
There are pairs $ (\fot, \lp)$ in
\begin{flushleft}
\begin{enumerate}
\item\label{c1}
$ \propclass\mytimes\ngroundclass $  such that
$\combop\chi\combpar\not\anyexpequiv_{\oga}\combop\gamma\combpar$, for
$(\chi,\gamma)\in \{
(_{HP},_{EB}),$ $
(_{HP},_{HP}^\lor),$ $
(_{HP}^\lor,_{EH}^\lor)\}$;
\item\label{c2}
$ \emptyclass\mytimes\ngroundclass $ such that
$\combop\chi\combpar\not\anyexpequiv_{\og}\combop\gamma\combpar$, for
$(\chi,\gamma)\in\{(_{HP},_{EB}),$ $
(_{HP},_{HP}^\lor),$ $
(_{HP}^\lor,_{EH}^\lor)\}$;
\item\label{c3}
$  \emptyclass\mytimes\nsafeclass $ such that $\combhpdis(\fot, \lp)
\not \anyexpequiv\combebdis(\fot, \lp)$;
\item\label{c5}
$ \hornclass\mytimes\nsafeclass $ such that
$\combop\chi\combpar\not\anyexpequiv_{\oga}\combop\gamma\combpar$, for
$(\chi,\gamma)\in\{(_{HP}, \mbox{}_{EB}),$ $
(_{HP}, \mbox{}_{EH}),$ $
(_{EB}, \mbox{}_{EH}),$ $
(_{HP}, \mbox{}^\lor_{HP}),$ $
(_{HP}, \mbox{}^\lor_{EB}),$ $
(_{EB}, \mbox{}^\lor_{EB}),$ $
(_{EB}, \mbox{}^\lor_{HP}),$ $
(^\lor_{HP}, \mbox{}^\lor_{EB}),$ $
(_{EH}, \mbox{}^\lor_{EB}),$ $
(_{EH}, \mbox{}^\lor_{HP})\}$; and
\item\label{c8}
$ \emptyclass\mytimes\narbclassprg $ such that
$\combeb(\fot, \lp) \not \anyexpequiv\combeh(\fot, \lp)$.
\end{enumerate}
\end{flushleft}
\end{proposition}

\begin{proof}[of Theorem \ref{th:fo-ael-exp-correspondence}] Correctness
of the table (i.e., the
`$ \Leftarrow $' direction of the theorem) follows from
the fact that  $\embeh(\lp)$ and
$\embehdis(\lp)$ are identical for normal programs $\lp$
(thus $\combeh\combpar\,{=}\,\combehdis\combpar$), Theorem~\ref{cor:corr-ground-atomic},
Propositions~\ref{prop:eb-eh-safe}--\ref{prop:eb-eh-prop-merged},
the inclusion relations between the classes of programs and FO theories, and
the properties of the $\anyexpequiv_\mathrm{x}$ relation. Completeness (i.e., the `$ \Rightarrow $' direction) is shown
analogously, by exploiting the counterexamples in
Proposition~\ref{prop:counterexamples}.
One can verify with little yet tedious effort that the table is
complete and that no entries can be relaxed
(e.g., using a simple logic program).
\end{proof}

\section{Application to the Semantic Web}
\label{sec:SW-application}

The original motivation for our work was the interest of
combinations of rules and ontologies in the Semantic Web. Below
we illustrate how our results may be applied in this context.
Briefly, some uses are:
\begin{itemize}
\item to capture the semantics of proposed combinations in a uniform language;
\item to derive properties of such and other combinations;
\item to design semantics for %extensions of 
combinations%  to
%     richer languages
, such that the ontology and rule parts are faithfully captured, and 
controlling the effect of aspects like grounding and working with open vs.\
closed domains.
\end{itemize}
We note that, apart from an obvious relationship to the open vs.\ closed
domain issue, grounding rules is important from a practical
perspective, since many rule engines in use today employ grounding. In
fact, powerful rule engines like smodels%
\footnote{\url{http://www.tcs.hut.fi/Software/smodels/}},
DLV%
\footnote{\url{http://www.dbai.tuwien.ac.at/proj/dlv/}}, and clasp%
\footnote{\url{http://www.cs.uni-potsdam.de/clasp/}}, along with many others that offer
stable and/or well-founded semantics are essentially based on
evaluation of ground programs. Thus, aspects such as the invariance of a
combination with respect to grounding the rules prior to evaluation
(formally captured in Definition~\ref{def:grounding-invariance} below)
are also important from a practical perspective.

We concentrate  on two prominent Semantic Web languages, namely
(i) RDF and its extension RDF Schema (RDFS)
\cite{Klyne+Carroll-ResoDescFram04}, and (ii) OWL DL
(Version~2)
\cite{Motik+Patel-SchneiderETAL-OntoLangSpecFunc:09}. For reasons of clarity
we restrict ourselves to normal programs.

In the remainder of this section, we consider the combinations \comb, which are
defined as $\comb\combpar=\fot\cup\emb(\lp)$, with
$\embvar\in\{\mathit{HP},\mathit{EB},\mathit{EH}\}$, where \fot is an FO
theory (the ontology) and \lp is a normal logic program.

\subsection{Grounding Invariance and Closed Domains}
\label{sec:grounding-closed}

In order to state our
results for RDF and OWL concerning grounding and open vs.\ closed domains, we
first formally define grounding invariance and closed domain semantics. In the
following, $\mathcal{X}$ is a class of FO theories and $\mathcal{Y}$ is a
class of normal logic programs.

\begin{definition}[Grounding Invariance]
\label{def:grounding-invariance}
A combination $\comb$ \emph{fulfills grounding invariance} (or is
\emph{invariant under grounding}) for $\mathcal{X},\mathcal{Y}$, if
$\comb(\fot,\lp)\anyexpequiv_\oga\comb(\fot,\gr\lp)$, for every $(\fot,\lp)\in
\mathcal{X}\times\mathcal{Y}$.
\end{definition}
When speaking about open and closed domain semantics in the context of
combinations of rules and ontologies, we are interested in the effective
domain of quantification of the variables in the rules. In the open domain
semantics, variables quantify over arbitrary domains, while in the closed
domain semantics variables quantify over a fixed domain, e.g., the set of
ground terms obtained from the constants and function symbols appearing in the
rules or ontologies.

Recall that, given a normal program \lp and a rule $r\in\lp$, the embedding
$\emb(r)$ is a formula of the form $(\forall)\; b_r\limpl h_r$.

\begin{definition}[Closed Domain]
  \label{def:closed-domain}
  A combination $\comb$ is \emph{closed-domain} for
  $\mathcal{X},\mathcal{Y}$ if for every $(\fot,\lp)\in
  \mathcal{X}\times\mathcal{Y}$ and every stable expansion $T$ of
  $\comb\combpar$ the following property holds. For every interpretation \inter such that $\inter\anymodels_T
  \comb\combpar$ and variable assignment $\varass$, whenever
  $\inter,\varass\anymodels_T b_r$ for some rule $r\in\lp$, then $\varass$ assigns every variable $x$ in
  $r$ to a named individual, i.e., $x^\varass=t^\funsym$, for some name
  $t$.    
 Otherwise, $\comb$ is \emph{open-domain} for $\mathcal{X},\mathcal{Y}$.
\end{definition}
Essentially, a combination is closed domain if rules can only be applied
(i.e., the body is satisfied in a model and variable assignment) if all
variables are assigned to named individuals. 

We have that combinations involving only ground logic programs are trivially
closed-domain. 

\begin{proposition}
\label{cor:closed-domain-ground}
  Combinations defined as
  $\comb(\fot,\gr\lp)$ are closed-domain for every
  $\mathcal{X},\mathcal{Y}$.
\end{proposition}
Since combinations that are invariant under grounding are
equivalent (with respect to ground atoms) to the combination obtained by
grounding the program, grounding invariance essentially implies
closed-domain. The following observations follow straightforwardly from the
definitions of the respective combinations and the properties of autoepistemic
logic. 

\begin{proposition}
  \label{pro:closed-domain} The combinations \combeb and \combeh are
  closed-domain for FO theories and safe logic programs. The combination
  \combhp is open-domain already for empty theories and positive safe normal logic
  programs.
\end{proposition}
Of particular interest to combinations of rules and ontologies on the Semantic
Web are \emph{DL-safe} programs
\cite{Motik+SattlerETAL-QuerAnswOWL-with:05}, which yield grounding
invariance (for positive programs), thereby effectively imposing a
closed-domain semantics.  An atom $p(\vec{t})$ is a \emph{rule atom} if $p$
appears only in \lp. We call the program \lp \emph{DL-safe}, if \lp is safe
and every variable in every rule $r$ of $P$ appears in a rule atom in $B^+(r)$.

The next propositions follow straightforwardly from the proof of
Proposition~\ref{prop:positive-horn}.

\begin{proposition}
\label{pro:grounding-invariance-1}
Let \fot be a \ghornclass theory and \lp a DL-safe positive normal
 program. Then, $\combhp(\fot,\lp)\anyexpequiv_\oga\combeb(\fot,\lp)
\anyexpequiv_\oga\combeh(\fot,\lp)$.
\end{proposition}

\begin{proposition}
\label{pro:grounding-invariance-2-4} The following combinations
fulfill grounding invariance:
\begin{enumerate}
\item\label{itm:grounding-1} $\combeb$ and $\combeh$ for \ghornclass theories
  and safe positive normal programs;
\item\label{itm:grounding-2} $\combhp$, $\combeb$, and $\combeh$ for
  \hornclass theories and safe positive normal programs; and
\item\label{itm:grounding-3} $\combhp$, $\combeb$, and $\combeh$ for
  \ghornclass theories and DL-safe positive normal programs.
\end{enumerate}
\end{proposition}
Even under DL-safety, a generalization of
this result from positive to normal
logic programs fails, and the combinations $\combhp$,
$\combeb$, and $\combeh$ behave differently---this can be
shown by replacing in the proof of
Proposition~\ref{prop:counterexamples}(4) the theory $\fot$
with $\{ p(b),$ $p(a)\limpl q \}$; the program \lp is then
DL-safe.  From Theorem~\ref{th:fo-ael-exp-correspondence}, we
can then conclude the following:
%that for \hornclass{} theories and DL-safe
%normal logic programs, none of the combinations $\combhp$,
%$\combeb$, and $\combeh$ is invariant under grounding of the
%rules.
\begin{proposition}
\label{pro:grounding-non-invariance}
  The combinations \combhp, \combeb, and \combeh are not invariant under
  grounding for \hornclass theories and DL-safe normal logic programs.
\end{proposition}
Observe that this result, combined with Proposition \ref{pro:closed-domain},
shows that closed-domain does not imply grounding invariance. In contrast,
Proposition~\ref{cor:closed-domain-ground} shows that the converse effectively
holds, as long as one is interesting only in ground atoms.

A weaker notion of safety, namely \emph{weak DL-safety} \cite{rosati06}
(see also Section~\ref{sec:hybrid}) plays also an important role in the
Semantic Web context, because of the possibility to write conjunctive queries
over DL ontologies. A program \lp is \emph{weakly DL-safe}, if it is safe
and for every rule $r$ and every variable $x$ in $r$, $x$ either appears only
in non-rule atoms in $B^+(r)$ or $x$ appears in a rule atom in $B^+(r)$.  As
weak DL-safety is stronger than ordinary safety, clearly
Proposition~\ref{pro:grounding-invariance-2-4}(\ref{itm:grounding-1}) and
Proposition~\ref{pro:grounding-invariance-2-4}(\ref{itm:grounding-2}) extend to weakly DL-safe programs.
However, Proposition~\ref{pro:grounding-invariance-2-4}(\ref{itm:grounding-3}) does
not, and the same happens also to
Proposition~\ref{pro:grounding-invariance-1}.

\begin{proposition}
\label{pro:weak-dl-safe}
There is a pair $(\fot,\lp)\in\ghornclass\times\nsafeclass$ such that \lp is
weakly DL-safe and positive,
$\combhp\combpar\not\anyexpequiv_\oga\combeb\combpar$, and
$\combhp\combpar\not\anyexpequiv_\oga\combeh\combpar$.
\end{proposition}
\begin{proof}
  Consider $\fot =\{ \exists x\st p(x) \}$ and $\lp = \{ q \leftarrow p(x)
  \}$: \combhp\combpar allows to conclude $q$, whereas \combeb\combpar and
  \combeh\combpar do not.
\end{proof}

\subsection{RDF, RDF Schema, and Rules}
\label{sec:rdf} Recall that RDF is the basic data description
language of the Semantic Web, in which atomic statements have
the form
\triple{\textit{subject}}{\textit{predicate}}{\textit{object}}.
RDFS has further axioms about the meaning of certain triples;
for example, that the facts \triple{a}{\textrm{rdfs:subClassOf}}{b} and
\triple{b}{\textrm{rdfs:subClassOf}}{c} imply
\triple{a}{\textrm{rdfs:subClassOf}}{c}.

As shown by \citeN{Bruijn-Heymans-LogiFoun-07}, (finite) RDF graphs $S$ are
essentially \ghornclass theories of the form $\fot=\{\exists
\vec{x}. \bigwedge S\}\cup\Psi$, where the free variables in $S$ are among
$\vec{x}$ and $\Psi$ is a set of function-free Horn logic formulas, which capture
the RDFS semantics \cite{Hayes-Sema04}.

Combinations of RDF graphs with rules---e.g., the RIF RDF and OWL
compatibility recommendation~\cite{Bruijn-Comp:08} and Jena\footnote{\url{http://jena.sourceforge.net/}}---are
common, because of the flexibility to manipulate data that rules offer. Note
that in this context it is not possible to make a strict separation between
ontology and rule predicates, as the $triple$ predicate is ``defined'' by both
the ontology and the rules.

Current combinations of RDFS with rules are typically limited to positive Horn
rules---a notable exception being the work by
\citeN{DBLP:journals/jair/AnalytiADW08}. For example, the
RIF-RDFS~\cite{Bruijn-Comp:08} semantics essentially defines the combination
of an RDF graph \fot and a set of positive normal rules \lp as the first-order
logic theory $\combhp(\fot,\lp) = \fot\cup\embhp(\lp)$.%
\footnote{We avoid here to go into unnecessary and tedious detail
concerning the RIF-RDFS semantics specification, which does not
give further insight.}

This semantics can be straightforwardly extended to normal
rules \lp by interpreting the FO-AEL theory
$\fot\cup\embhp(\lp)$ using the any- or all-name semantics.
%The combination of RDF with rules is an example of a case in
%which it is not possible to make a strict separation between
%ontology and rule predicates, since both the ontology and the
%rules use the single $triple$ predicate.
Such an extension keeps the spirit of the RIF-RDFS semantics by having
an \emph{open domain}, i.e., not only the constants,
but also the existentially quantified variables in the RDF graphs
matter (see
also Proposition \ref{pro:closed-domain}).

However, rules typically have a closed domain semantics.  One
may thus argue that combinations should respect this semantics
and enforce a closed domain in the interaction between the RDF
statements and the rules; examples of such combinations are
$\fot\cup\emb(\gr\lp)$ and $\fot\cup\embeh(\lp)$: the former enforces closing
of the domain through grounding, while the latter forces closing through the
use of the modal operator \mop in the rules (cf.\
Proposition~\ref{cor:closed-domain-ground} and Proposition \ref{pro:closed-domain}). Note
that, by Theorem~\ref{th:fo-ael-exp-correspondence}, the
embeddings $\embhp$, \embeb, and \embeh may be used
interchangeably in the combination $\fot\cup\emb(\gr\lp)$ (as
long as we are interested only in ground atomic consequences),
as \fot is in \ghornclass and $\gr\lp$ is in \ngroundclass. The
following example illustrates the difference between
combinations with open and with closed domain semantics,
respectively.

\begin{example}
\label{ex:rdf-closed-domain} Consider the RDF graph
$$\fot=\{\exists  x.  \triple{$x$}{director}{TheGodfather}\}\cup\Psi$$
encoding the fact that there is a director of the film ``The
Godfather''. Consider also the program $\lp =
\{\mathit{hasDirector}(x)\leftarrow \triple{$y$}{director}{$x$}\}$
encoding that whenever someone directs a film, then this film
has a director. We have
\begin{align*}
\embhp(\lp) = & \{\forall
x,y.\triple{$y$}{director}{$x$}\limpl \mathit{hasDirector}(x) \}\text{ and}\\
\embeh(\lp) = & \{\forall x,y.\triple{$y$}{director}{$x$}\land
\mop\triple{$y$}{director}{$x$}\limpl \\
 & \phantom{\{\forall x,y. } \mathit{hasDirector}(x)\land
\mop \mathit{hasDirector}(x) \}\text{.}
\end{align*}
Clearly, $\mathit{hasDirector}(\text{TheGodFather})$ is a consequence of
$\fot\cup\embhp(\lp)$, but not of $\fot\cup\embeh(\lp)$, as
there is no constant $c$ such that
\triple{c}{director}{TheGodFather} is included in the single
stable expansion of $\fot\cup\embeh(\lp)$.

Similarly, $\mathit{hasDirector}(\text{TheGodFather})$ is not a
consequence of $\fot\cup\emb(\gr\lp)$, since there is no
constant representing the director.
\end{example}
%\comment{TE: By our results, for safe positive rules grounding should not
%matter? Please consider!}
%
%We note that,
%as long as \lp is positive, the combination
%$\fot\cup\embhp(\gr\lp)$ is an FOL theory, and reasoning can be
%done using standard Datalog technology, as demonstrated by
%\cite{Bruijn-Heymans-LogiFoun-07}.
%
Proposition~\ref{pro:grounding-invariance-2-4}
shows that grounding $\lp$ or not, prior to combination with
$\fot$, does not matter if $\fot$ is an RDF graph without blank
nodes (as then $\fot$ is in \hornclass) and $\lp$ is positive---in particular,
$\fot\cup\emb(\gr\lp)\anyexpequiv_\oga
\fot\cup\emb(\lp)$. Similarly for $\embeb$ and
$\embeh$, if $\fot$ is an arbitrary RDF graph, and $\lp$ is
safe and positive. If $\lp$ is moreover {DL-safe}
%\cite{Motik+SattlerETAL-QuerAnswOWL-with:05}
(see Section~\ref{sec:grounding-closed}), this invariance under
grounding for arbitrary RDF graphs and safe positive programs
also extends to the \embhp embedding, and thus \embeb, \embeh, and
\embhp are all interchangeable. We furthermore have that the open and
closed domain semantics coincide in the cases mentioned in this
paragraph.

If $S$ is an RDF graph, we define
$\comb(S,\lp)=\fot\cup\emb(\lp)$, where $\fot=\{\exists\vec{x}.\bigwedge
S\}\cup\Psi$, as before.  
An RDF graph $S$ is \emph{ground} if it does not contain free variables. Such
a graph is equivalent to the Horn theory $\fot=\{\bigwedge S\}\cup\Psi$.
%The following result follows 
From Propositions~\ref{pro:grounding-invariance-1} and
\ref{pro:grounding-invariance-2-4} we then obtain:

\begin{corollary}
  \label{cor:rdf-grounding-invariance}
  The combinations \combeb and \combeh fulfill grounding invariance for RDF
  graphs and safe positive normal programs, and \combhp fulfills it for ground
  RDF graphs and safe positive normal programs, as well as for RDF graphs and
  DL-safe positive normal programs. Moreover, for RDF graphs \fot and DL-safe
  positive normal programs \lp, it holds that $\combhp\combpar
  \anyexpequiv_\oga\combeb\combpar \anyexpequiv_\oga\combeh\combpar$.
\end{corollary}
A notable further consequence of
Proposition~\ref{pro:grounding-invariance-2-4} is the following
observation concerning the standard RIF-RDFS semantics
\cite{Bruijn-Comp:08}.
\begin{corollary}
  \label{cor:rif-rdf-grounding-invariance}
  The RIF-RDFS combination semantics fulfills grounding invariance for DL-safe
  positive normal programs.
\end{corollary}
For a possible use case scenario, suppose the ontology \fot is a ground RDF
graph. Now, suppose the user wants to add a set of DL-safe positive
normal rules and follow the standard RIF-RDFS combination semantics
\cite{Bruijn-Comp:08}. If the user is interested only in ground atomic
consequences, Corollary~\ref{cor:rif-rdf-grounding-invariance} tells
us that this semantics is invariant under grounding and thus
essentially closed-domain, by Proposition~\ref{cor:closed-domain-ground}. Even
when extending the graph with variables, the combination remains
invariant under grounding and thus closed-domain. However, grounding
invariance may be lost when extending the program with negation, by
Proposition~\ref{pro:grounding-non-invariance}.

\subsection{OWL DL and Rules}
\label{sec:owl}

The Web Ontology Language OWL DL is based on Description Logics (DLs);
Version~1 \cite{Patel-Schneider+HayesETAL-OntoLangSemaAbst03} is based on the
DL \shoin and Version~2 \cite{Motik+Patel-SchneiderETAL-OntoLangSpecFunc:09}
on the DL \sroiq.  Both DLs can be viewed as subsets of first-order logic
\cite{Sattler+CalvaneseETAL-RelawithOtheForm:03}. An influential proposal for
combining OWL DL ontologies with positive normal rules is the Semantic Web
Rules Language (SWRL) \cite{Horrocks+Patel-SchneiderETAL-SWRL04}, which gives
a standard first-order semantics to their union.

A SWRL theory \fot consists of a set of DL axioms and a set of Horn-like
formulas. We obtain the following correspondence with \combhp{} combinations.

\begin{proposition}
  \label{pro:swrl} Let \fot be a SWRL theory. Then, there are an FO theory
  $\fot'$ and safe positive normal logic program \lp such that
  $\fot\models\alpha$ iff $\alpha$ is a consequence of $\combhp(\fot,\lp)$,
  for every objective ground atom $\alpha$.
\end{proposition}

An approach similar to SWRL was adopted
by the RIF working group for positive normal RIF rules
\cite{Bruijn-Comp:08}. If \fot is the FOL-equivalent of an OWL DL
ontology and \lp is a set of positive normal rules, the semantics of
RIF-OWL DL combinations is given by the FO theory $\combhp(\fot,\lp) =
\fot\cup\embhp(\lp)$. By Proposition~\ref{pro:swrl}, this is equivalent to
SWRL.

Regarding open versus closed domains, similar considerations as
in Section~\ref{sec:rdf} apply to combinations of OWL DL with
rules: $\fot\cup\embhp(\lp)$ yields an open domain, while
$\fot\cup\emb(\gr\lp)$ and $\fot\cup\embeh(\lp)$ yield a closed
domain on the rule side (see
Example~\ref{ex:rdf-closed-domain}). However,
interchangeability and invariance of the embeddings $\embhp$,
$\embeb$, and $\embeh$  under grounding may not be guaranteed, as
$\fot$ need not be in $\ghornclass$.

\begin{example}
\label{ex:owl-ground-emb} Consider
$\fot = \{A(a),
\forall x.  A(x) \limpl B(x) \lor C(x) \}$, which captures a simple
OWL~DL ontology, and
$\lp = \{q \leftarrow B(x); q \leftarrow C(x) \}$.  Now,
\[
\begin{array}{r@{~}l}
\embhp(\gr\lp)=&\{B(a) \limpl q, C(a) \limpl q\}\text{ and} \\[1ex]
\embeh(\gr\lp)=&\{B(a)\land \mop B(a) \limpl q\land\mop q,
C(a)\land \mop C(a) \limpl q\land\mop q\}.
\end{array}
\]
We have that $q$ is a consequence of $\fot\cup\embhp(\gr\lp)$,
but not of $\fot\cup\embeh(\gr\lp)$, since neither $B(a)$ nor
$C(a)$ is included in the single stable expansion of
$\fot\cup\embeh(\gr\lp)$.
\end{example}
There are important fragments of OWL DL that are essentially included in
\ghornclass, such as the OWL~2 profiles EL, QL and RL
\cite{Motik+GrauETAL-OntoLang:08}, and the fragment corresponding to
Horn-\shiq \cite{Hustadt05datacomplexity}.  As was the case with RDF, when
considering the combination $\fot \cup \emb(\gr\lp)$ (see
Propositions~\ref{pro:grounding-invariance-1} and
\ref{pro:grounding-invariance-2-4}), the embeddings \embhp, \embeb, and \embeh
may be used interchangeably, and, for safe positive programs, $\embeb$ and
$\embeh$ are invariant under grounding.  Furthermore, OWL~2 RL is essentially
in \hornclass. Therefore, when considering combinations of OWL~2 RL ontologies
with positive normal programs, \combhp\combpar, \combeb\combpar, and
\combeh\combpar may be used interchangeably, by
Proposition~\ref{prop:positive-horn}, and the combinations are invariant under
grounding of the rules, by
Proposition~\ref{pro:grounding-invariance-2-4}. Moreover, they are
closed-domain.

%Grounding invariance holds for all mentioned OWL~2 profiles and
%Horn-\shiq under DL-safety, by Proposition
%\ref{pro:grounding-invariance-2-4}, and the fact that they fall in
%\ghornclass.
 As shown by
\citeN{Motik+SattlerETAL-QuerAnswOWL-with:05}, reasoning with OWL DL
plus DL-safe rules (i.e., SWRL having DL-safe rules) is
decidable. % An inspection of the proof of
% Proposition~\ref{prop:positive-horn} shows that for DL-safe
% positive normal programs, both equivalence between \combeb\
% and \combeh\ in the first part of the proposition and invariance
% under grounding extend to \combhp.
%The following corollary follows 
From Propositions~\ref{pro:grounding-invariance-1} and
\ref{pro:grounding-invariance-2-4}, we obtain the following corollary.
Here, OWL \ghornclass theories are theories of OWL~2 EL, OWL~2 QL, OWL~2 RL, or
Horn-\shiq.

\begin{corollary}
  \label{cor:owl-grounding-invariance} The combinations \combeb and \combeh{}
  fulfill grounding invariance for OWL \ghornclass theories and safe positive
  normal programs, and \combhp fulfills it for OWL~2 RL theories and safe
  positive normal programs, as well as for OWL \ghornclass theories and
  DL-safe positive normal programs. Moreover, if \fot is an OWL \ghornclass
  theory and \lp a DL-safe positive normal program, then $\combhp\combpar
  \anyexpequiv_\oga\combeb\combpar \anyexpequiv_\oga\combeh\combpar$.
\end{corollary}
Consider a scenario in which the ontology \fot is in both OWL~2 RL and
OWL~2 EL and one wants to add positive rules that are safe (but
not DL-safe), using the standard RIF-OWL combination semantics
\cite{Bruijn-Comp:08}. Corollary~\ref{cor:owl-grounding-invariance}
tells us that one may employ any of the considered
combinations \comb and may ground the rules, as long as one is only
interested in atomic formulas. However, if we were to extend the
ontology \fot towards full OWL~2 EL by introducing existentially
quantified variables, we may no longer use the \combeb or \combeh
combinations, if we want to remain faithful to the RIF-OWL
combination semantics, as illustrated by Example
\ref{ex:rdf-closed-domain}. In addition, we may not ground the
rules prior to reasoning. Therefore, if such a future extension
towards OWL~2 EL is likely, one should choose the \combhp
embedding rather than \combeb or \combeh, and should not rely on
grounding for reasoning.

% by suitably restricting the use of variables in rules to :
%In addition, by Proposition
%\ref{prop:positive-horn}, \embeb and \embeh may be used
%interchangeably in the combination $\fot \cup \emb(\lp)$, and
%when considering
%\embhp, \embeb, and \embeh may be used interchangeably.

% \subsection{Use Cases}

% We illustrate the applicability of the results in the preceding subsections
% using a number of small scenarios.

% \subsubsection{\hornclass and \ghornclass}

% [The idea is to say that there are scenarios like a) RDF graphs with no blank
% nodes b) RDF graphs with blank nodes   c) OWL2 EL ontologies d) OWL2
% ontologies, and one want to add rules to them, possibly with negation, and
% wants to choose an embedding. First question could be i) whether the choice
% matters (w/o disjunctive rules), second the grounding/open closed domain
% issue.

% In addition to this, a dynamic aspect: one has scenario a) at the moment but
% transition to b) may happen later. In the OWL domain, could be the conversion
% of OWL ontologies: say one moves from OWL2RL to OWL2EL (do some conversion)
% and recasts the ontology, perhaps not to have a 1-1 correspondence, but
% still. Rules exists aside, chosen with an embedding. The question is then
% whether the rule embedding causes additional problems (or would work in the
% same way).] 

\section{Discussion}
\label{sec:discussion}

In this section we discuss implications of the results in the
previous sections.  We
first discuss the implications of our results on the relationships between the
embeddings, and make a number of observations about those relationships.
We then discuss how the results in this paper can be used in
the context of combining classical theories (ontologies) with logic programs
(rules).  Specifically, how the embeddings studied in this paper can be used as
building blocks for such combinations.  Finally, we discuss our choice of
FO-AEL as the underlying formalism, and compare the semantics for
quantification (quantifying-in) with other approaches to quantifying-in in
autoepistemic logic \cite{levesque90:_all_i,levesque00:_logic_of_knowl_bases,kaminski02:_revis_quant_autoep_logic}.

\subsection{Relationships between the Embeddings}

Using the results obtained in Sections \ref{sec:prop-embedd} and
\ref{sec:prop-comb}, we can make a number of
observations about the embeddings:

\begin{longenum}
\item The differences between the embeddings by themselves
do not depend on the use of negation in the
program. Generally speaking, the differences originate from the
positive use of the modal operator in the antecedent and the
consequent, and the use of the $\PIA$ axioms. However, in combinations with FO theories, the interaction between names in \fot for which there are no \UNA axioms and negation in the rules gives rise to different behavior of the embeddings (see Proposition \ref{prop:counterexamples}(\ref{c5})).

\item The stable expansions of embeddings with and without the
$\PIA$s generally tend to differ. However, we can note that the former
are generally stronger in terms
of autoepistemic consequences (cf.~Figure~\ref{fig:rel-cons} and Example \ref{ex:pia}).

\item The embeddings $\embhp$ and $\embhpdis$ are generally
the strongest in terms of consequences (see Figure
\ref{fig:rel-cons}), when comparing to other embeddings without and with
$\PIA$s, respectively. They allow to derive the contrapositive of
rules (cf.\ Example~\ref{ex:contrapositive}) and the bodies of rules
are applicable to unnamed individuals, whereas the antecedents of
the axioms in the other embeddings are only applicable to named
individuals, because of the positive modal atoms in the bodies.

\item For unsafe programs, the embeddings $\embeh$ and $\embehdis$
are generally not comparable with the others; embeddings of unsafe
rules may result in axioms of the form $\forall x.\mop p(x)$
(cf.~Example~\ref{ex:general}), which result in all individuals being
named.

\item If names in \fot that lack \UNA axioms
 and rules in \lp (e.g., \fot is propositional or \lp is ground) do
not interact, then \embeb and \embeh
 are in most cases interchangeable.
\end{longenum}
Special care needs to be taken if one selects an embedding that
includes the $\PIA$ axioms (i.e., $\embhpdis$ and $\embebdis$). These axioms of the form
$\alpha \limpl \mop \alpha$ ensure that $\alpha$ or $\neg \alpha$ is included in every stable expansion, for
every ground atom of $\signature_\lp$. Note that the $\PIA$ axioms have
no effect when considering individuals that are not named by
ground terms in $\signature_\lp$.

The $\UNA$ axioms in embeddings, which serve to make individuals
different by default, may interact with the FO theory in a combination. For example,
consider $\lp = \{p(a);\; p(b)\}$ and $\fot = \{a \,{\neq}\, b \limpl
r,\; a\,{\neq}\, c \limpl s\}$.  Then, every stable expansion of
$\comb\combpar$, for any embedding $\emb$ we considered, contains $r$
as $a\neq b$ is concluded by default, but not $s$ (as $c$ is
unknown in $P$). To shortcut such (possibly undesired) inequality transfers from $\lp$ to $\fot$,
the unique names or even the standard names assumption may be
adopted {\em a priori}.  Recall that the results on the embeddings in
Section~\ref{sec:embedd-logic-progr} were obtained by stepping through
the standard names assumption, and thus they also hold
under the unique names or standard names assumption, as shown by \citeN{debr-etal-08}. On the one hand, this should
extend to the positive results about correspondences in
Sections~\ref{sec:prop-embedd} and \ref{sec:prop-comb}, whose proofs
rely on named interpretations and no equalities between individuals are
enforced. On the other hand, some counterexamples for correspondences
fail, including those for the first item in Lemma~\ref{lem:not-rel-cons} and
Proposition~\ref{prop:counterexamples}(\ref{c5}), and thus further
correspondences may hold. An in-depth study of the effect of
unique names and standard names assumptions on the correspondences and
differences between the embeddings is an interesting subject
for further work.

\subsection{Different Embeddings and Combinations}

Recall the general setting for combining a first-order theory \fot and a logic
program \lp in a unifying formalism (FO-AEL) that we sketched in the introduction.
The combination operator $ \combop $ takes as arguments the theory \fot and the
program \lp, and returns an FO-AEL theory $ \combop(\fot, \lp)$.  The operator
provides two embedding functions: $ \sigma $ and $ \tau $ map first-order theories,
respectively logic programs, to FO-AEL theories.  We also mentioned that in the simplest
case the combination is the union of the two individual embeddings: $
\combop(\fot, \lp)=\sigma(\fot)\cup \tau(\lp)$.

In Section \ref{sec:embedd-logic-progr} we investigated several candidates for
the embedding function for logic programs, $ \tau $. All these embedding
functions are faithful, in the sense that the stable models of the program \lp
correspond to the sets of objective ground atomic formulas in the stable
expansions of the embedding $ \tau(\lp) $.  In Section \ref{sec:prop-embedd} we
investigated the relationships between the stable expansions of these
embeddings when considering more general formulas.  It turned out that there
are already significant differences between the expansions when considering
non-ground or non-atomic formulas.

Now, in Section \ref{sec:prop-comb}, we investigated the relationships between
the expansions when considering combinations of the embeddings with first-order
theories.  We have found that, under certain circumstances---namely, when the
first-order theory and program are of particular shapes and we are interested
in a particular kind of formulas (e.g., ground formulas)---certain embeddings
can be used interchangeably (cf.\ Table~\ref{tab:summary-exp}).
%  on page \pageref{tab:summary-exp}).
For example, if the program is normal and ground ($
\lp \in \ngroundclass$), the theory is generalized Horn ($ \fot \in \ghornclass $),
and we are
interested in objective formulas, we can use the embeddings \embeb and \embeh
interchangeably: $ \combeb \expequiv\combeh $ for $ \lp \in \ngroundclass $
and $ \fot \in \arbclassth$, according to Table \ref{tab:summary-exp}, as $ \hornclass \subseteq
\arbclassth$ and the set of objective formulas is a subset of the set of formulas.

Our results are not limited to combinations of the form $ \combop(\fot, \lp)=\fot
\cup \embgen(\lp)$, where \embgen is one of the embeddings investigated in this
paper.  One could imagine adding axioms to \fot or rules to \lp to achieve the
desired interoperation between the two components, or even changing the axioms
or rules (e.g., by grounding), obtaining a first-order
theory $ \fot' $ and program $ \lp' $.
In this more general setting, the combination is defined as
$$ \combop(\fot, \lp)=\fot' \cup \embgen(\lp'),$$
where $ \fot' $ and $ \lp' $ are obtained from \fot and \lp by
adding and/or replacing axioms and rules. The results of
Section~\ref{sec:prop-comb} can be applied, provided that
$\fot'$ and $ \lp' $ are in the respective classes of theories
and programs, independent of the shapes of \fot and \lp (see
also Section~\ref{sec:SW-application}).

As discussed in Section~\ref{sec:embedd-logic-progr}, embeddings that
include the \UNA axioms are not modular in general, but only signature-modular.
This can be remedied by instead using the single axiom
$$ (\forall) \; \mop x=x \land \mop y=y \land \lnot \mop
x=y \limpl x \neq y$$
which has the same effect for embeddings. However, using
this axiom would entail default uniqueness on all names in a
combination, not only those from the signature of the program
(if desired, such default uniqueness can be easily accomplished by
just mentioning respective terms in the logic program). As
a consequence, also the combinations behave differently.

\subsection{Quantifying-in in First-Order Autoepistemic Logic}

We consider here FO-AEL, with the semantics
for quantifying-in as defined by \citeN{konolige91:_quant}, as an underlying
formalism for combinations of first-order theories and logic programs.
However, further semantics for quantifying-in have been proposed
in the literature.

\citeN{levesque90:_all_i} defined the logic of only knowing
(see also the subsequent work by
\citeN{levesque00:_logic_of_knowl_bases}), which is essentially
a superset of FO-AEL. Levesque's semantics for quantifying-in
is slightly different from the one of \citeN{konolige91:_quant}
that we used in this paper. He adopted a standard names
assumption that amounts to a special case of the notion in
Section \ref{sec:preliminaries-fol}; there is a countably
infinite number of constant symbols in the language, but there
are no (other) function symbols. Likewise, the variant of
FO-AEL by \citeN{kaminski02:_revis_quant_autoep_logic} also
employs a standard names assumption, although under a somewhat
different guise: the domain of every interpretation is an
extended Herbrand interpretation, i.e., it is a superset of the
set of constant symbols in the theory; function symbols are not
considered. The semantics of Konolige does not impose such
restrictions, e.g., the domain may be infinite, while the
number of constants is finite, and function symbols are
allowed.

It is well known that reasoning in standard first-order logic
can be reduced to reasoning in first-order logic with the standard names
assumption, as long as there are sufficiently many constant symbols
available~\cite{Fitting-FirsOrdeLogiAuto:96}.

Different from \citeN{levesque90:_all_i},
\citeN{kaminski02:_revis_quant_autoep_logic} did not consider
equality in the language. However, equality in first-order
logic with standard names behaves quite differently from
equality in standard first-order logic.  In the latter case,
two constant symbols may be interpreted as the same element in
the domain, whereas in the former case, all constant symbols
are interpreted distinctly, e.g., $ a=b $ cannot be satisfied
if $ a $ and $ b $ are distinct constant
symbols.\footnote{\citeN{levesque00:_logic_of_knowl_bases}
extend the logic of only knowing by allowing the use of
constants and function symbols different from standard names;
several ground terms may be associated with one \emph{standard
name}, and for any constant symbols $ a $ and $ b $ with this
property, $ a=b $ is satisfied.} It is, however, possible to
reduce reasoning in standard first-order logic with equality to
reasoning in first-order logic with standard names using a
special congruence predicate~\cite[Theorem
9.3.9]{Fitting-FirsOrdeLogiAuto:96}.
\citeN{Motik+Rosati-FaitInteDescLogi:07} use such a  predicate
in their variant of the logic MKNF
\cite{Lifschitz-NonmDataEpisQuer:91,Lifschitz-MiniBeliNegaFail:94},
as do \citeN{debr-etal-08} in a variant of FO-AEL with standard
names; see Section~\ref{sec:related-combinations} for further
discussion about this work.

\section{Related Work}
\label{sec:related-work}

We review here two areas of related work: extensions of logic programming
and description logic semantics with open domains and nonmonotonicity,
respectively, and approaches to combining rules and ontologies.

\subsection{Extensions of LP and DL Semantics}
\label{sec:lp-dl-extensions}
We have studied the combination of logic programs and
ontologies using embeddings in a unifying formalism
(FO-AEL). One could imagine, in contrast, extensions of the semantics
of logic programs or ontologies to incorporate (parts of) the other
formalism. One such extension of logic programming semantics is that
of open domains
\cite{gelfond93:_reason_open_domain}.
Such
extended semantics can be used to accommodate incomplete knowledge,
an important aspect of ontology languages.

\citeN{van97} define \emph{open logic programs}, which are
combinations of sets of rules and first-order logic formulas;
the set of predicate symbols is partitioned into a set of
\emph{open} and a set of \emph{closed} predicates. The
semantics of the program is the first-order theory consisting
of Clark's completion of the closed predicates and the
first-order formulas in the open program.  They then discuss
how description logics can be embedded in such open logic
programs and they discuss the correspondence between abduction
in open programs and reasoning in description logics.

\citeN{Heymans+NieuwenborghETAL-ConcLogiProg:06}  describe an
extension of the stable model semantics with open domains,
called \emph{open answer set programming} (\emph{OASP}).  They show
how the expressive DL \shiq{} can be embedded in this language
and \citeN{Heymans-BruijnETAL-GuarHybrKnowBase-08} show how
OASP can be used for combinations of rules and ontologies,
following the \dllog semantics \cite{rosati06} (see Section
\ref{sec:related-combinations}).

Recently, \citeN{Cali+GottlobETAL-:09} presented Datalog$^\pm$
as a language that, similarly as OASP, can be used to enhance
ontologies with rules. In essence, Datalog$^\pm$ amounts to a
skolemized form of \ghornclass in a relational setting, where
for decidability rules must satisfy a guardedness condition. As
reported by \citeANP{Cali+GottlobETAL-:09}, various DLs can be
encoded into Datalog$^\pm$, and thus, like in OASP, combination
of rules and ontologies can be achieved by adding rules to this
encoding. Furthermore, \citeANP{Cali+GottlobETAL-:09} present a
semantics for Datalog$^\pm$ programs with stratified negation
that generalizes the usual notion of stratified programs, which
thus enables combinations with nonmonotonic rules. An embedding
of (stratified) Datalog$^\pm$ into FO-AEL via the embedding
$\embhp$ seems easily possible, such that its (operational)
semantics can be reconstructed in logical terms, as well of the
combination with the DLs described. Moreover, the embedding can
be used to give semantics to unstratified Datalog$^\pm$
programs via FO-AEL, and the results of this paper can be
exploited to derive properties. Investigating this in detail
remains for future work.

Several nonmonotonic extensions
of description logics
have been defined in the literature
\cite{baader95:_embed,Donini+LenzeriniETAL-episoperdesclogi:98,Donini+NardiETAL-DescLogiMiniKnow:02,bonatti06:_expres_non_monot_descr_logic_based_circum}.
These might be further extended to accommodate logic programs by well-known
correspondences of the latter to nonmonotonic formalisms.
In more detail, extensions of DL semantics with defaults and circumscription have been
described by \citeN{baader95:_embed} and
\citeN{bonatti06:_expres_non_monot_descr_logic_based_circum}, respectively.
Extensions with nonmonotonic modal operators, inspired by the logic MKNF \cite{Lifschitz-NonmDataEpisQuer:91}, have
been described by
\citeANP{Donini+LenzeriniETAL-episoperdesclogi:98} \citeNN{Donini+LenzeriniETAL-episoperdesclogi:98,Donini+NardiETAL-DescLogiMiniKnow:02}.
Both works mention a notion of
procedural or  default rules, which are rules involving description logic
concepts. \citeN{Donini+LenzeriniETAL-episoperdesclogi:98} allow rules of the
form $ C \Rightarrow D $, where $ C $ and $ D $ are DL concepts (i.e., unary
predicates); such rules are intuitively read ``if  an
individual  is  proved  to  be  an  instance  of  $C$,  then  derive  that  it
is  also  an  instance  of  $D$''. The default rules considered by
\citeN{Donini+NardiETAL-DescLogiMiniKnow:02} are a generalization; they are of
the form $C_0, \dnot C_1, \dots, \dnot C_n \Rightarrow D $, $n\geq 0$,
where all $ C_i$ and $D$ are DL concepts.  Intuitively, ``if  an
individual  is  proved  to  be  an  instance  of  $C_0$ and is not proved to be
an instance of $ C_1 $, \dots, or $ C_n $,  then  derive  that  it  is  also
an  instance  of  $D$''.
The work of \citeANP{Donini+LenzeriniETAL-episoperdesclogi:98}
inspired some more advanced formalisms for combining rules and
ontologies, which we consider next.

\subsection{Combinations of Rules and Ontologies}
\label{sec:related-combinations}

Roughly speaking, we can distinguish between three kinds of
combinations of rules and ontologies: (1) \emph{uniform combinations}
(e.g., CARIN \cite{Levy+Rousset-CombHornruledesc:98} and SWRL
\cite{Horrocks+Patel-SchneiderETAL-rule:05}), (2) \emph{hybrid
combinations} (e.g., dl-programs
\cite{eite-etal-07} and \dllog
\cite{rosati06}), and (3) \emph{embedding combinations} (e.g., the
MKNF combination by \citeN{Motik+Rosati-FaitInteDescLogi:07} and a
combination based on quantified equilibrium logic \cite{Bruijn-PearceETAL-QuanEquiLogiHybr-07}); for
more discussion, see, e.g., the works of \citeN{eite-etal-07} and \citeN{bruijn06-repres-issues-about-combin-of}.
We also note the recent approach by \citeN{debr-etal-08} for embeddings of dl-programs, \dllog, and MKNF into FO-AEL.

\subsubsection{Uniform Combinations}

With uniform combinations we mean combinations of ontologies that are
essentially classical first-order theories and of Horn logic
formulas that are essentially positive rules.  The combined theory,
which is the set-theoretic union of the formulas in the ontology and
the Horn formulas, is interpreted under the standard first-order logic
semantics.

In the CARIN approach \cite{Levy+Rousset-CombHornruledesc:98},
the ontologies are theories of the description logic
$\mathcal{ALCNR}$ and the rules are Datalog rules, i.e., safe
positive normal rules as defined in Section~\ref{sec:preliminaries-programs}, with the further restriction
that predicates which occur in the ontology may not be used in
rule heads.  \citeANP{Levy+Rousset-CombHornruledesc:98} show
that reasoning with these combinations is undecidable in
general, but becomes decidable when suitably restricting either
the ontology or the rules. As discussed in Section~\ref{sec:SW-application},
\citeN{Motik+SattlerETAL-QuerAnswOWL-with:05} demonstrated
decidability of SWRL---the combination of OWL DL with normal
positive rules---restricted to DL-safe rules.

\subsubsection{Hybrid Combinations}
\label{sec:hybrid}

Hybrid approaches combine logic programs with nonmonotonic
negation (usually, under the stable model semantics or the
well-founded semantics) with a description logic knowledge base
or, in more abstract terms, theories in first-order logic.  The
two most prominent such approaches are dl-programs
\cite{eite-etal-07} and \dllog \cite{rosati06}. The main
difference between them is the way in which the interaction
between the individual components (the logic program and the
ontology) is managed.  For both, we assume that the ontology
component is a DL theory and the logic program is function-free
and safe.

In dl-programs, the interoperation between the program and the
ontology is achieved by \emph{DL queries}, which are queries to
the DL ontology, in the bodies of the rules; prior to
evaluation, information from the program may be temporarily
added to the ontology for a query. Eiter et
al.~\citeyear{eite-etal-07} show that query answering in
dl-programs is decidable as long as reasoning in the individual
components (ontology and logic program) is decidable.
HEX-programs \cite{eiter05:_unifor_integ_higher_order_reason}
generalize dl-programs to more general external evaluations
that are not limited to queries on DL ontologies.

\dllog
 makes a distinction between
 ontology and rules predicates;
rules predicates may not occur in the ontology, but the ontology
predicates may occur in the rules.  The combination is interpreted by
a single first-order interpretation, but the part of the
interpretation concerned with the rules predicates is subject to
stability conditions corresponding to the usual definition of stable
models.
%  (cf.\ Section~\ref{sec:preliminaries-programs} and the seminal paper by \citeN{Gelfond+Lifschitz-StabModeSemaLogi:88}).
Thus, the
interoperation is based on single models, resulting in a broad
interface between the program and the ontology.  \citeN{rosati05}
shows that if the rules are DL-safe and satisfiability checking in the
ontology component is decidable, then reasoning with the combination
is decidable.  \citeN{rosati06} shows that reasoning is decidable if
the problem of containment of conjunctive queries in unions of
conjunctive queries is decidable for the underlying DL, provided that
the rules are \emph{weakly DL-safe}; this notion dispenses DL-safety
for variables that occur only in ontology predicates in rule bodies,
which makes it possible to access unnamed individuals in rules.
\allog \cite{Donini+LenzeriniETAL-AL-l:98} can be seen as a precursor
of \dllog that considers only positive programs and that allows
(unary) ontology predicates only in rule bodies and effectively
requires DL-safety.
The differences between the underlying principles of dl-programs
and \dllog are discussed in more detail by \citeN{bruijn06-repres-issues-about-combin-of}.

Since we did not distinguish between rule and ontology
predicates in our em\-bed\-dings---indeed, in the introduction we claimed
this is undesirable---there is no straightforward correspondence
between any of the embeddings we considered and the mentioned hybrid
approaches.
The embeddings we considered in this paper can be used to construct
combinations that have a tight integration between the components and
that do not have a separation between ontology and rules
predicates. In fact, the \dllog approach can be reconstructed by an
extension of simple combinations \comb\combpar = $\fot\cup \lp$ with {\em
classical interpretation axioms}, which, loosely speaking, fix the
value of classical predicates for stable expansions; we refer to
\citeN{debr-etal-08} for details.

\subsubsection{Embedding Combinations}

\noindent \citeN{Motik+Rosati-FaitInteDescLogi:07} propose a combination of DL
ontologies and nonmonotonic logic programs through an embedding into the bimodal
nonmonotonic logic MKNF \cite{Lifschitz-NonmDataEpisQuer:91}, which uses
the modal operators \mopk, which
stands for ``knowledge'', and \mopnot, which stands for ``negation as
failure''.
The
variant of MKNF used by Motik and Rosati
employs a standard names assumption similar to the approach of \citeN{levesque90:_all_i}:
there is a one-to-one correspondence between the countably many constant symbols in the
language and elements in the domains of interpretations (functions
symbols are not considered).
The equality symbol of
first-order logic ($=$) is embedded using a special binary predicate symbol $
\approx $ and the usual congruence axioms \cite[Chapter
9]{Fitting-FirsOrdeLogiAuto:96} are added.  Logic programs are
embedded into MKNF using the transformation described by
\citeN{Lifschitz-MiniBeliNegaFail:94}: a rule $r$ of form (\ref{eq:rule})
% on page \pageref{eq:rule})
is embedded as the formula
\[
\textstyle \tau_{MKNF}(r)=\bigwedge_i \mopk b_i \land \bigwedge_j \mopnot\, c_j \limpl
\bigvee_k \mopk h_k.
\]
A classical theory \fot is
embedded as a conjunction comprising all the formulas in the theory, preceded
by the modal operator \mopk: $ \sigma_{MKNF}(\fot)=\mopk(\bigwedge \fot) $.
Finally, the combination of the logic program $\lp$ and the
first-order theory $\fot$ is simply $\combop_{MKNF}=\tau_{MKNF}(\lp)\cup \{\sigma_{MKNF}(\fot) \}$.

Comparing % the embedding
$ \tau_{MKNF}$ to the embeddings in
Section~\ref{sec:embedd-logic-progr}, we can see that it is close in spirit to
the embedding \embehdis; both embeddings feature modal belief operators in
front of positive atoms in both the body and the head of the rule. In fact,
it turns out that, when using a variant of FO-AEL with standard names, there
is a one-to-one correspondence between the stable expansions of $
\tau^{\lor-}_{EH}(\lp)$ and the MKNF models of $
\tau_{MKNF}(\lp)$ (recall that $ \tau^{\lor-}_{EH}(\lp)$ is $
\tau^{\lor}_{EH}(\lp)$ without the $\UNA$ axioms); however, this correspondence does
not extend to combinations with FO theories, as shown by \citeN{debr-etal-08}.

Besides the obvious differences between MKNF and autoepistemic logic---il\-lus\-trat\-ed
by the differences between the $ \tau_{MKNF}$ and $
\tau^{\lor-}_{EH}$ embedding functions---there is a difference in the
semantics for quantifying-in between the variant of MKNF used by
\citeN{Motik+Rosati-FaitInteDescLogi:07} and Konolige's any- and all-name
semantics that we used in this paper.  Since FO-AEL permits arbitrary
interpretations, we needed to utilize $\UNA$ axioms.
Motik and Rosati employ the standard names assumption and thus do not need such
axioms.

\medskip

\noindent
As already pointed out, \citeN{Bruijn-PearceETAL-QuanEquiLogiHybr-07} used another nonmonotonic logic for combining ontologies
and logic programs, namely {quantified equilibrium logic}\/ (QEL)
\cite{PV05}. While FO-AEL and MKNF are nonmonotonic modal logics, QEL
is based on the nonclassical \emph{logic of here-and-there}, which is
an intermediate logic between classical and intuitionistic logic.
Negation in QEL is nonmonotonic; however, by axiomatizing the law of
the excluded middle (LEM) through $ \forall \vec{x}( p(\vec{x})\lor
\lnot p(\vec{x}))$, one can enforce that a predicate $ p $ is
interpreted classically, and negation of this predicates becomes
classical.
Actually, \citeN{Bruijn-PearceETAL-QuanEquiLogiHybr-07} used a slightly generalized version of QEL that does not assume
uniqueness of names and includes equality
to show that the QEL theory obtained by adding such LEM axioms to the combination
$\combop\combpar = \fot\cup\lp$ of a FO theory  \fot and a logic
program \lp yields the \dllog semantics.

\section{Conclusion}
\label{sec:conclusions}

We have defined various embeddings of non-ground
programs into first-order autoepistemic logic (FO-AEL) that
generalize respective embeddings of propositional logic programs into
standard AEL, and we have investigated their semantic properties. We
have shown that these embeddings are faithful, in the sense that the
stable models (or answer sets) of a given non-ground logic program \lp
are in one-to-one correspondence to the stable expansions of the
embeddings $ \embgen(\lp)$ with respect to objective ground atomic
formulas. Furthermore, we have analyzed the correspondences
between the embeddings at more fine-grained levels, revealing their
commonalities and differences.

Our results provide a basis and a stepping stone for the more
complex endeavor to combine classical knowledge bases and
non-ground logic programs in a uniform logical formalism (which
is one of the targets of the Semantic Web architecture), namely
the well-known and amply studied formalism of autoepistemic
logic. Indeed, since the combination of positive RIF rules with
RDF and OWL DL \cite{Bruijn-Comp:08} corresponds to one of the
combinations we studied, our results are directly applicable to
such combinations.

In this direction, we have investigated correspondences between simple
combinations of embeddings of logic programs with FO theories
for various classes of logic programs
and FO theories. The results of our investigation provide useful
insights into the behavior of different embeddings for logic programs
with respect to a context, given by a first-order theory,
and allows some conclusions about the
replaceability of one embedding by another without altering the behavior of
the combination. Based on the results in the present paper, more elaborated
combinations of logic programs with FO theories are investigated by
\citeN{debr-etal-08}, who show how well-known approaches to combining
rules and ontologies in the Semantic Web context can be embedded into
FO-AEL, like those of
\citeN{eite-etal-07}, \citeN{rosati06}, and \citeN{Motik+Rosati-FaitInteDescLogi:07}.
Notably, the \dllog approach
can be embedded into FO-AEL by adding
further axioms to the simple combination that we have considered here.%

Several issues remain for future work.  In the present paper,
we focused on semantic aspects of embeddings of logic programs,
but we did not address computational issues.  Since the
embeddings are easily computed, they may be exploited to
establish decidable fragments of combinations of rules and
ontologies, and to craft sound (but possibly incomplete)
algorithms for specific reasoning tasks for such combinations.
There are several promising starting points for devising
algorithms for computing stable expansions and/or autoepistemic
consequences in FO-AEL. \citeN{Niemelae-decicompautoreas:92}
presents a general procedure for computing stable expansions in
FO-AEL without quantifying-in.
\citeN{levesque00:_logic_of_knowl_bases} present a sound, but
incomplete proof theory   for the logic of only knowing, which
extends FO-AEL with standard names. Finally,
\citeN{Rosati-TowaFirsNonmReas:99} presents techniques for
reasoning with first-order MKNF (with standard names) with a
limited form of quantifying-in; the \mopnot operator in MKNF is
equivalent to $ \lnot \mop $ in autoepistemic logic
\cite{rosati97embedding}.

Other issues are extensions of the language used for logic programs.
Adding classical negation to the \embeb and \embeh is routine, and has
been done by \citeN{debr-etal-08} for FO-AEL with standard names.
Other interesting extensions include nesting
\cite{Lifschitz+TangETAL-NestExprLogiProg:99}, where the closeness
between nesting in logic programs and the logic MKNF suggests
that an embedding is straightforward, and aggregates
\cite{fabe-etal-jelia04,ferr-2005,DBLP:journals/tplp/PelovDB07,DBLP:journals/tplp/SonP07}.

Furthermore, in the present work, we considered embeddings of
logic programs interpreted under the stable model semantics,
which adopts a two-valued semantics.  It would be interesting
to consider also embedding of logic programs under many-valued
semantics, most importantly under the well-founded semantics
\cite{Gelder+RossETAL-WellSemaGeneLogi:91}, which is a
three-valued semantics for logic programs with negation that
has also been considered for combination of rules and
ontologies \cite{knor-etal-08,DBLP:conf/rr/DrabentHM07}.
Three-valued extensions of autoepistemic logic \cite{denecker03:_unifor,bonatti95:_autoep_logic_unify_framew_seman_logic_progr,Przymusinski-ThreNonmFormSema:91}
may be used as a starting point.

%\begin{figure}
%\center
%\includegraphics[width=0.5\linewidth]{layerCake-4.png}
%\caption{The most recent incarnation of W3C's famous ``Semantic Web Layer cake''.}
%\label{figsestack}
%\end{figure}

Lastly, the initial motivation for our work has been the
application to Semantic Web languages. Combinations of positive
RIF rules with RDF and OWL DL ontologies, as we have discussed
in the present paper, are just a first step. Nonmonotonic
extensions of RIF~\cite{Kifer-rif08}, and also the RDF Query
Language SPARQL~\citeyear{sparql}, are instances of the
combination problems we have sketched in the present paper. The
semantics of both nonmonotonic RIF and SPARQL can be expressed
in terms of nonmonotonic logic
programs~\cite{Kifer-rif08,angl-guti-2008,poll-2007}, but their
combination with OWL ontologies is still an open issue on W3C's
agenda in completing the Semantic Web
architecture~\cite{w3c2007}. We expect that our results can be
used to provide valuable insights towards the definition of a
\emph{unifying logic} encompassing the Semantic Web Ontology
(OWL, RDFS), Rules (RIF) and Query languages (SPARQL).

\appendix

\section{Appendix: Proofs of Sections \ref{sec:prop-embedd} and \ref{sec:prop-comb}}
\label{sec:proofs}

This appendix contains the remaining proofs of the results stated in
Sections~\ref{sec:prop-embedd} and \ref{sec:prop-comb}.
We remind the reader
that we omit the superscript $\mathrm{X}$ from
$ \xmodels$, $\xexpequiv_\gamma$, $ Cn^{\mathrm{X}}$, \xstable, etc.\
if the stated property holds regardless of whether $ X=E $ or $ X=A $.

\subsection*{Proofs of Section \ref{sec:prop-embedd}}
\label{sec:proofs-5}

\begin{proof}[of Proposition \ref{prop:corr-ground-formulas}]
% JB: the proof of the first part of the proposition is redundant with the proof of proposition 6.8. However, the proposition is only stated for the any-name semantics, so I guess that for now it's okay to keep the proof as it is
Let $T$ be a  stable expansion of $\embeb(\lp)$ or $\embeh(\lp)$
(resp., of $\embhpdis(\lp)$ or $\embebdis(\lp)$).
We only consider $ T_\oga$ in the following; so, by Theorem~\ref{cor:corr-ground-atomic},
the choice between \embeb and \embeh (resp., $\embhpdis$ and $\embebdis$) is immaterial.
Define
$$\Gamma^\embop = \{\phi \in \flang_{g} \mid \embop(\lp)\models_{T_{\oga}}\phi\},
\ \embop\in\{\embeb,\embeh,\embhpdis,\embebdis\}.
$$
 By Proposition
\ref{prop:atoms-expansions}, $\embop(\lp)$ has a stable expansion
$T^\embop$ such that $\Gamma^\embop=T^\embop_{og}$.

\smallskip

\noindent First
we will show that $T^{\embeb}_{og} = T^{\embeh}_{og}$ by establishing
\begin{equation}
\label{eq:prop-4} T^{\embop}_{og}=\{\phi\in\flang_{g}\mid
T_{\oga}\models\phi\}.
\end{equation}
$T^{\embeb}_{og} = T^{\embeh}_{og}$
follows from this claim, thereby establishing the first part of the proposition,
concerning the embeddings \embeb and \embeh.

Every entailed objective ground formula is equivalent to a conjunction
of ground clauses $c=l_1\lor\dots\lor l_k$, where each
$l_i$ is either an atom $p_i$ or a negated atom $\lnot p_i$.
Clearly, $T_{\oga}\models c $ iff $l_i\in T_{\oga}$ for some $l_i$ in
$c$. To prove
(\ref{eq:prop-4}),
clearly $T^{\embop}_{og} \supseteq \{\phi\in\flang_{g}\mid
T_{\oga}\models\phi\} $. For the other inclusion,
suppose that $\embop(\lp)\models_{T_{\oga}}c$ (hence
$c\in T_{og}^\embop$), but $T_{\oga}\not\models c$. Hence, $l_i\notin
T_{\oga}$
for every $l_i$ in $c$. Consider an arbitrary
interpretation $\inter$ such that
$\inter\models_{T_{\oga}}\embop(\lp)$; then,
$\inter\models_{T_{\oga}}c$. Let $\inter'$
result from $\inter$ by flipping the truth value of the atom
$p_i$ of each literal $l_i $ in $c$ such that $\inter \models_{T_{\oga}}l_i$; clearly,
$\inter'\nmodels_{T_{\oga}}c$.  We now show that
$\inter'\models_{T_{\oga}}\embop(\lp)$; this contradicts
the assumption $\embop(\lp)\models_{T_{\oga}}c$ and
proves (\ref{eq:prop-4}).

Let $\alpha \in \embop(\lp)$ be an instance of an axiom that
originates from a rule in \lp. We show that
$\inter'\models_{T_{\oga}}\alpha$. Suppose first that some flipped atom $p_i$
occurs in the antecedent $\alpha_a$ of $\alpha$.  If
$\inter'\models_{T_{\oga}}\neg p_i$, then clearly
$\inter'\models_{T_{\oga}}\alpha$.  Otherwise,
$\inter\models_{T_{\oga}} \neg p_i$ and thus $\mop p_i$, which also
occurs in the antecedent of $\alpha$, is false in $\inter'$.  Hence,
$\inter' \models_{T_{\oga}}\alpha$.
Suppose then that $\alpha$ has no flipped $p_i$
in $\alpha_a$, and that
$\inter' \nmodels_{T_{\oga}} \alpha$. Hence,
$\inter' \models_{T_{\oga}} \alpha_a$ and $\inter'
\models_{T_{\oga}} \neg p_i$ for some
flipped $p_i$ that occurs in the consequent of $\alpha$.
As $\inter \models_{T_{\oga}} \alpha_a$, it follows
$p_i\in T_{\oga}$; as  $l_i\notin T_{\oga}$, we have
$l_i = \lnot p_i$ and thus $\inter \models_{T_{\oga}} \neg p_i$,
which implies $\inter' \nmodels_{T_{\oga}} \neg p_i$ by definition;
this is a contradiction. This proves $\inter' \models_{T_{\oga}} \alpha$.

As unique names axioms in $\embop(\lp)$ are clearly satisfied in
$\inter'$, $\inter'\models_{T_{\oga}}\embop(\lp)$, thereby establishing the claim (\ref{eq:prop-4}) and thus the first part of the proposition.

\medskip
\noindent  For the second part concerning
  $\embhpdis(\lp)$ and $\embebdis(\lp)$, we  exploit
the $\PIA$ axioms: thanks to them
each objective ground atom $\alpha $ or
its negation $\lnot\alpha $ is included in the stable expansion
$T^\tau$, for $\tau\in\{\embhpdis,\embebdis\}$. Thus,
$T^\tau_{og}=\{\phi\in\flang_{g}\mid T_{\oga}\models\phi\}$ (by
structural induction), from which
$T^{\embhpdis}_{og} = T^{\embebdis}_{og}$ follows immediately.
\end{proof}

\begin{proof}[of Proposition \ref{prop:corr-safe-formulas}]
Let
$T$ be a stable expansion
of $\embeb(\lp)$ (resp., $ \embeh(\lp)$).  By
Theorem~\ref{cor:corr-ground-atomic}, we know that
\begin{equation}
\label{eq:proof-eb-eh-safe-1}
\embeb(\lp) \expequiv_{\oga} \embeh(\lp),
\end{equation}
hence, $ T_{\oga} = T'_{\oga}$ for some stable expansion $T'$  of $ \embeh(\lp)$ (resp., $ \embeb(\lp)$).

We proceed as follows.  We first  claim that (\dag) given a safe
program \lp and an interpretation \inter,
one can construct a named interpretation $\inter'$
such that
$\inter\models_{T_{\oga}} \embop(\lp)$ iff
$\inter'\models_{T_{\oga}}\embop(\lp)$, where $\embop\in \{\embeb, \embeh\}$. It follows that,
for every formula $ \phi$,
\begin{multline}
\label{eq-proof-eb-eh-safe-2}
\embop(\lp)\models_{T_{\oga}} \phi
\text{ iff for every named interpretation } \inter, \\
\inter
\models_{T_{\oga}}\embop(\lp) \text{ implies } \inter \models_{T_{\oga}}\phi
\end{multline}
We denote the latter---i.e., entailment restricted to named
interpretations---with the symbol
 $\models^{\names}_{T_{\oga}}$.

We then
claim that (\ddag) for every formula $ \phi \in \fmodl $, $
\embeb(\lp)\models^{\names}_{T_{\oga}}\phi $ iff $ \embeh(\lp)\models^{\names}_{T_{\oga}}\phi $;
combined with \eqref{eq:proof-eb-eh-safe-1} and \eqref{eq-proof-eb-eh-safe-2},
this establishes
$\embeb(\lp)\expequiv\embeh(\lp)$. We now proceed to prove
the individual claims (\dag) and (\ddag).

\medskip
\noindent (\dag)\quad
Let
\interdef be an interpretation.  Let the named interpretation $\inter^\names=\langle
\domain^\names,\cdot^{\funsym^\names}\rangle$ be as follows:
$\domain^\names=\{t^\funsym \mid t \in \names\}$;
for
every $n$-ary predicate symbol $p$, $p^{\funsym^\names}=
p^\funsym\cap(\domain^\names)^n$;
and for every $n$-ary function symbol $f$ and
$\vec{k}\in (\domain^\names)^n$, $f^{\funsym^\names}(\vec{k})=f^{\funsym}(\vec{k})$.

Consider $\embeb(\lp)$, which contains two kinds of formulas:
$\UNA$ axioms and axioms of the form $(\forall)\bigwedge(b_i\land\mop b_i)\land
\bigwedge\lnot\mop c_j\limpl h$. The former
are obviously satisfied in $\langle \inter^\names,T_{\oga}\rangle$ iff they are satisfied in
 $\langle \inter,T_{\oga}\rangle$, because they are variable-free. Consider
an open formula $\alpha = \bigwedge(b_i\land\mop b_i)\land \bigwedge\lnot\mop c_j\limpl
h$ and a variable assignment $B$ of $\inter$.  Since \lp is
safe,
every variable in $\alpha$ occurs in some $b_i$.  Therefore, if $B$
assigns any variable in $\alpha$ to an
unnamed individual, $(\inter,B)\nmodels_{T_{\oga}} \mop b_i $; hence  $(\inter,B)\models_{T_{\oga}} \alpha$.
If $B$ assigns all variables in $\alpha$
to named individuals,
then $(\inter,B) \models_{T_{\oga}} \alpha$
iff  $(\inter^\names,B) \models_{T_{\oga}} \alpha$, by construction of $\inter^\names$. The proof for the case of
$\embeh(\lp)$ is analogous; the antecedents of the implications in $
\embeb(\lp)$ and $ \embeh(\lp)$ are the same. This proves the claim (\dag).

\medskip

\noindent It remains to prove (\ddag).
Let $T$ be a stable expansion of $\embeb(\lp)$.\\
\noindent ($ \Rightarrow $)
By the shape of the formulas in $
\embeb(\lp)$ and $ \embeh(\lp)$, for every interpretation \inter
clearly $ \inter \models^\names_{T_{\oga}}\embeh(\lp)$ implies $ \inter
\models^\names_{T_{\oga}}\embeb(\lp)$.
Hence, for every formula $ \phi\in \fmodl $
it holds that $\embeb(\lp)\models^\names_{T_{\oga}}\phi $ implies
$\embeh(\lp)\models^\names_{T_{\oga}}\phi $.

\smallskip

\noindent
($ \Leftarrow $) We proceed
by contradiction. Suppose that $ \embeh(\lp)\models^\names_{T_{\oga}}\phi $, but $
\embeb(\lp)\nmodels^\names_{T_{\oga}}\phi $. Hence there must be a
named interpretation $ \inter $ such that $ \inter
\models_{T_{\oga}}\embeb(\lp) $ and $\inter\nmodels_{T_{\oga}}\phi$. We
will show that $ \inter \models_{T_{\oga}}\embeh(\lp)$, which contradicts the
assumption.

The $\UNA$ axioms in $ \embeh(\lp)$ are obviously satisfied in \inter, since they
are also in $ \embeb(\lp)$. Consider a formula $(\forall)\,\alpha \limpl
h \land \mop h \in\embeh(\lp)$ that is not satisfied in
$ \langle \inter, T_\oga  \rangle $, where $\alpha = \bigwedge (b_i\land
\mop b_i) \land \bigwedge \lnot \mop c_j$.
Since $\inter\models_{T_{\oga}}(\forall)\,\alpha\limpl h \in \embeb(\lp)$,
we have that, for some variable assignment $ B $,
$ (\inter, B) \models_{T_{\oga}}\alpha$ and
$ (\inter, B) \not \anymodels_{T_{\oga}}\mop
h$ (resp., $ (\inter, B) \not \allmodels_{T_{\oga}}\mop
h$); hence, $ h\varsub \notin T_{\oga}$ for all (resp., for some) name substitution(s)
\varsub associated with $ B $.  However, since $ (\inter, B)
\models_{T_{\oga}}\alpha$, and thus $(\inter, B) \anymodels_{T_{\oga}} \mop b_i$ and $(\inter, B) \not\anymodels_{T_{\oga}} \mop c_j$
(resp., $(\inter, B) \allmodels_{T_{\oga}} \mop b_i$ and $(\inter, B) \not\allmodels_{T_{\oga}} \mop c_j$),
it follows that $b_i\varsub_{b_i} \in T_{\oga} $ and $ c_j\varsub_{c_j} \notin
T_{\oga}$\footnote{Observe that $ b_i \varsub_{b_i} $ and $ c_j \varsub_{c_j} $ are
well-defined, because all individuals in \inter are named.}
for some/all
(resp., for all/some) name substitutions
$\varsub_{b_i},\varsub_{c_j}$ associated with $ B $.
Now, since $T$ is a
stable expansion of $ \embeb(\lp)$,
from
(\ref{eq-proof-eb-eh-safe-2}) and the $\UNA$ axioms contained in $\embeb(\lp)$ it follows that
$ h\varsub \in T_{\oga}$ for some (resp., for all) name substitution(s)
\varsub associated with $ B $.  Therefore,
$\inter\models_{T_{\oga}} (\forall)\,\alpha\limpl h \land \mop h $ and
$\inter \models_{T_{\oga}}\embeh(\lp)$, establishing the desired
contradiction.
This proves that $\embeh(\lp)\models^\names_{T_{\oga}} \phi $ implies $ \embeb(\lp)\models^\names_{T_{\oga}} \phi$.

\smallskip
\noindent The case of $T$ being a stable expansion of $\embeh(\lp)$
is analogous.
\end{proof}

\begin{proof}[of Proposition \ref{prop:corr-ground-programs}]
By the $\PIA$ axioms $b_i \limpl\mop b_i$ we can eliminate the modal atoms
of the form
    $\mop b_i$ from the antecedents of the axioms in $\embebdis(\lp)$
that originate from rules in \lp.
The remaining theory is the same as $\embhpdis(\lp)$ and thus the stable expansions
    correspond.
  \end{proof}

\begin{proof}[of Proposition \ref{prop:models-norm}]
As $\embeb(\lp)$ and $\embeh(\lp)$ differ only in that
the latter has a conjunction $h_1\land \mop h_1$ in the consequent of embedded rules
$\embeh(r)$ while the former simply has $h_1$, the first implication is
immediate. The second implication is argued similarly; each $\embeb(r)$
has a stronger antecedent than $\embhp(r)$.
  \end{proof}

\begin{proof}[of Proposition \ref{prop:models-disj}]
Similarly as $\embebdis(\lp)$ and $\embhpdis(\lp)$ in Proposition \ref{prop:models-norm},
$\embebdis(\lp)$ and $\embhpdis(\lp)$ only differ in
that the embedding $\embebdis(r)$ of each rule $r$ has a stronger antecedent than the
corresponding $\embhpdis(\lp)$, from which the first part follows.

For the second part, as in
Proposition \ref{prop:corr-safe-formulas}
it suffices to consider
the case of named individuals.
Consider any interpretation
$\langle \inter,T \rangle$ such that \inter
is named and $\inter \models_T \embebdis(\lp)$.
By the $\PIA$ axioms, for every variable assignment $B$,  associated
name substitution \varsub, and $a \,{\in}\,\flang_{\oga}$,
we have that $(\inter,B) \models_{T} a$
implies $\inter \models_T \mop a\varsub$. Hence, $\inter \models_T
\embehdis(\lp)$.
\end{proof}

\begin{proof}[of Lemma \ref{lem:rel-cons}]
Recall
that a formula $\phi\in \fmodl$ is an autoepistemic consequence of a theory
$\fot\subseteq \fmodl$ if
$\phi$ is included in every stable expansion of \fot. Recall also that, by
Theorem~\ref{cor:corr-ground-atomic}, the stable expansions of all
embeddings $\emb^{(\vee)}(\lp)$ correspond with respect to
$\flang_{\oga}$.  From this and the fact that, by Proposition
\ref{prop:atoms-expansions}, every stable expansion $T$
is determined by $T_{\oga}$ and
$\emb^{(\vee)}(\lp)$, we can conclude  $C_\chi^{(\vee)}
\subseteq C_\gamma^{(\vee)}$
if
$\inter \models_{T_{\oga}}\tau_\gamma^{(\vee)}(\lp)$ implies
$\inter \models_{T_{\oga}}\emb^{(\vee)}(\lp)$
for every interpretation
$\langle\inter,T_{\oga}\rangle$.
\begin{itemize}
\item $C_{EH} \longrightarrow C_{EH}^\vee$, $C_{EH}^\vee
\longrightarrow C_{EH}$:
follow from the fact that, if $\lp\in \narbclassprg$,
$\embeh(\lp)$ and  $\embehdis(\lp)$ coincide.

\item  $C_{EH} \dashrightarrow C_{EB}$:  follows from
Proposition~\ref{prop:corr-safe-formulas}.

\item  $C_{EB} \longrightarrow C_{EH}$, $C_{EB} \rightarrow C_{HP}$:
follows from
Proposition \ref{prop:models-norm}.

\item  $C_{HP} \longrightarrow C_{HP}^\vee$, $C_{EB} \rightarrow
  C_{EB}^\vee$:
if $\lp\in \narbclassprg$, then clearly $ \embhp(\lp)\subseteq \embhpdis(\lp)$,
so $\inter \models_{T_{\oga}}\embhpdis(\lp)$ implies
$\inter \models_{T_{\oga}}\embhp(\lp)$. Likewise for $ \embeb(\lp)$ and
$ \embebdis(\lp)$.

\item   $C_{EB}^\vee \longrightarrow  C_{HP}^\vee$:   $C_{EB}^\vee \subseteq C_{HP}^\vee$ follows
from Proposition \ref{prop:models-disj}.

  \item   $C_{EH}^\vee \dashrightarrow C_{EB}^\vee$: follows from
  Proposition~\ref{prop:models-disj}.
 \end{itemize}
\end{proof}

\begin{proof}[of Lemma \ref{lem:not-rel-cons}]
\mbox{}\\

\vspace*{-1.25\baselineskip}

\begin{itemize}
\item $C^{\lor}_{EH} \not\subseteq C^{\lor}_{HP}$: Consider $
\lp=\{p(x);\; q(c) \mid c \in \names\}$, which is not safe.
Then $\embehdis(\lp)$ has a single stable expansion $T$, and
since $\forall x.p(x)\land\mop p(x)$ is in $T$, in every model
of $T$ all individuals must be named. Hence $T$, and thus $ C^{\lor}_{EH}$, contains $\forall x.q(x)$.
Also $\embhpdis(\lp)$ has a single stable expansion,
but individuals may be unnamed in its models, and so $\forall
x.q(x)\notin C^{\lor}_{HP}$ holds.

\item $ C_{EB}^\lor \not\subseteq C_{HP}$: Consider $ \lp =
\{b \leftarrow a \}$, which is clearly safe. Then $ \embebdis(\lp)$ has a single
stable expansion, and $a\notin T$, which implies
$\lnot \mop a \in T$. As $T$ contains the contrapositive $ \lnot \mop
a \limpl \lnot a $ of the $\PIA$
axiom $a \limpl \mop a $, $T$ contains also $ \lnot a $.  Consequently, $ \lnot a
\in C_{EB}^\lor$. Since $ \embhp(\lp)$ contains
no $\PIA$ axioms, $ \lnot a \not \in C_{HP}$.

\item $ C_{HP} \not\subseteq C^{\lor}_{EB}$: Consider the safe program $ \lp =
\{q(x)\leftarrow p(x)\}$.  The single stable expansion of $
\embhp(\lp)$ contains $ \forall x( p(x)\limpl q(x))$, while the single
stable expansion of $ \embebdis(\lp) $ does not.  Hence, $
\forall x( p(x)\limpl q(x)) \in C_{HP}$, but $ \forall x (p(x)\limpl
q(x)) \not \in C^{\lor}_{EB}$.
\hfill\mbox{} \qed
\end{itemize}
\end{proof}

\subsection*{Proofs of Section \ref{sec:prop-comb}}
\label{sec:proofs-6}

%Recall that we consider only the any-name semantics in Section \ref{sec:prop-comb}, and so
%when writing stable expansions we mean \anystable expansions and with $\models$ we mean
%$\anymodels$.

\begin{proof}[of Proposition \ref{prop:positive-horn}]
We reduce the first part of the proposition to the second part
through skolemization.

Let  $ \lp \in \nsafeclass$, $ \fot \in \ghornclass $ and let
$\signature = \signature_\fot\cup \signature_\lp$ be the
combined signature. To ease argumentation, we use a fresh unary
predicate $ d $ as a domain predicate in $\lp$: let $ \lp' $
result from \lp by adding for each variable $x$ occurring in
rule $ r \in \lp $ the atom $ d(x) $ to the body of $r$, and
let $\fot' = \fot \cup \{d(t) \mid t \in \names_\signature \}
$.

As an easy lemma,  the stable expansions of \comb\combpar and $
\comb(\fot', \lp') $ ($ \combvar \in \{{EB},$ $ {EH} \} $)
correspond and agree on formulas in $\lang= \lang_\signature$.
Indeed, as every rule $ r \in \lp $ is safe, every variable $x$
in $ \emb(r) $ occurs in a positive modal atom $\mop b_i$ in
the antecedent; thus, given a model \ainter of \comb\combpar
and a variable assignment $B$, $(\inter,B)\anymodels_\Gamma
\mop b_i$ holds only if $x^B$ is a named individual, and thus
$(\inter', B)\models d(x)$, where $ \inter' $ is obtained from
\inter by setting the extension of $ d $ to the set of named
individuals. On the other hand, as $d$ occurs in $ \comb(\fot',
\lp') $  only in facts and positively in antecedents of
embedded rules $ \emb(r) $, we can, in every model \ainter of $
\comb(\fot', \lp') $, decrease the extension of $d$ to the set
of named individuals without compromising satisfaction. It
follows that if $T$  is a stable expansion of $ \comb(\fot,
\lp) $, then there exists a stable expansion $S$ of
 $ \comb(\fot', \lp') $ such that $S_\oga = T_\oga \cup \{d(t) \mid t
\in \names_\signature \}$, and vice versa.

Now let $ \fot''$ be a skolemization of $\fot'$,  let $
\signature' $ (resp., $ \signature''$) be the combined
signature of $\lp' $ and $ \fot' $ (resp., $\lp' $ and $ \fot''
$), and let $\lang'$ (resp., $\lang''$) be the resulting FO
language. We now establish correspondence between the stable
expansions of $\Delta' = \comb(\fot', \lp')$ and $\Delta'' =
\comb(\fot'', \lp')$ with respect to $\lang'_{\ga} $. As
$(\fot'', \lp') \in \hornclass{\times}\narbclassprg$ and $\lp'$
is positive, the first part of the proposition then follows
from second.

($ \Leftarrow $) Let $T''$ be a stable expansion of $\Delta''$. %We show that  $\Delta'$ has a stable expansion $T'$ such that $T'_\oga =
%T''\cap\lang'_{\ga}$.
We note an important fact: (\dag) $d(t')\in T''$ such that
$t'\notin\names_{\Sigma'}$ (i.e., $t'$ is a Skolem term)
implies that $t=t'\in T''$ for some $t\in \names_{\Sigma'}$ (as
$d$ occurs in $\Delta''$  only in ground facts and in
antecedents of embedded rules). We show that, for every $\alpha
\in \lang'_\ga$, it holds that $\Delta'\anymodels_{T''_\oga}
\alpha$ iff $\alpha\in {T''_\oga}$ (i.e.,
$\Delta''\anymodels_{T''_\oga} \alpha$); by Proposition
\ref{prop:atoms-expansions} this proves that there is a stable
expansion $T'$ of $ \Delta' $ such that $T'_\oga =
T''\cap\lang'_{\ga}$.

Suppose first that $\inter\anymodels_{T''_\oga}
\Delta'\cup\{\neg\alpha\}$. Then we can extend $\inter$ by
suitable interpreting the Skolem functions, thereby obtaining
some $\inter'$ such that $\inter'\models \fot''$. Clearly,
$\inter' \anymodels_{T''_\oga} \UNA_{\signature_\lp}$. Now
suppose that for some variable assignment $B$ and some variable
$x$ that occurs in $\emb(r)$ ($r\in\lp'$),
$(\inter',B)\anymodels_{T''_\oga} \mop d(x)$. Then, by (\dag),
$x^B$  is named by some  $t\in\names_{\Sigma'}$ and hence
$(\inter,B)\anymodels_{T''_\oga} \mop d(x)$ and the antecedent
of $\emb(r)$ is satisfied by $ \langle \inter, T''_\oga \rangle
$  with respect to $ B $.
 It follows that $\inter' \anymodels_{T''_\oga}\emb(r)$ and thus $\inter'\anymodels_{T''_\oga} \Delta''\cup\{\neg\alpha\}$.

Conversely, suppose $\inter'\anymodels_{T''_\oga}
\Delta''\cup\{\neg\alpha\}$. Let $\inter$ be the restriction of
$\inter'$ to $\Sigma'$. Then  clearly
$\inter\anymodels_{T''_\oga} \fot \cup \UNA_{\signature_\lp}$.
Furthermore, if $(\inter,B)\anymodels_{T''_\oga} \mop d(x)$ for
some variable $x$ that occurs in $\emb(r)$, then
$(\inter',B)\anymodels_{T''_\oga} \mop d(x)$; as
$\inter'\anymodels_{T''_\oga} \emb(r)$, it follows that
$\inter\anymodels_{T''_\oga} \emb(r)$. Hence,
$\inter\anymodels_{T''_\oga} \Delta'\cup\{\neg\alpha\}$.

($ \Rightarrow $) Now let $T'$ be a stable expansion of
$\Delta'$. We show that  $\Delta''$ has a  stable expansion
$T''$ such that $T'_\oga = T''\cap\lang'_{\ga}$.  Let $\Gamma$
be the smallest set $\Gamma\subseteq\lang''_\ga$ such that (a)
$T'_\oga\subseteq \Gamma$ and (b) $\Delta''\anymodels_\Gamma
\Gamma$; the set $\Gamma$ exists since
$\Delta''\anymodels_{T'_\oga}\beta$ for each $\beta \in
T'_{\oga}$ and $\Delta''\anymodels_{S}\beta$ implies
$\Delta''\anymodels_{S'}\beta$ for each $\beta \in
\lang''_\oga$ and $S\subseteq S'\subseteq \lang''_\oga$. Hence,
$\Delta''$ has a stable expansion $T''$ such that
$T''_\oga=\Gamma$. As $T'_\oga \subseteq \Gamma$, it remains to
show that $\Gamma\cap\lang'_\ga \subseteq T'_\oga$.

We can obtain $\Gamma$ as $\Gamma = \bigcup_{i\geq 0}
\Gamma^i$, where $\Gamma^0=T'_{oga}$ and $\Gamma^{i+1} = \{
\alpha \in \lang''_\ga \mid \Delta'' \anymodels_{\Gamma^i}
\alpha\}$, $i\geq 0$; furthermore, $d(t') \in \Gamma^i$ implies
that $t=t' \in \Gamma^i$ for some $t'\in\names_{\Sigma'}$ by
induction on $i\geq 0$ (note that $\fot$ is Horn). Now let $i$
be the least index such that $\Gamma^{i+1}\cap
\lang'_\ga\not\subseteq T'_\oga$. Hence, $ \Delta''
\anymodels_{\Gamma^i} \alpha$ for some $\alpha \in
\lang'_\ga\setminus \Gamma^i$,  and for each model
$(\inter,\Gamma^i)$ of $ \Delta''$ there is some variable
assignment $B$ and axiom $\emb(r)\in \Delta''$ such that
$(\inter,B)\anymodels_{\Gamma^i} b\land \mop b$ for each $b$ in
the antecedent of $\emb(r)$, in particular for each $d(x)$, and
$\emb(r)$ has consequent $h$ such that $\alpha=h\varsub$ for
some name substitution $\varsub$ associated with
 $B$. From minimality of $i$ and safety of $r$, we
may assume that $\varsub$ has only names from
$\names_{\signature'}$. But then
$\Delta'\anymodels_{T'_\oga}\emb(r)$ implies that $\alpha \in
T'_\oga$, which is a contradiction. Hence $\Gamma^{i+1}\cap
\lang'_\ga\subseteq T'_\oga$ holds for all $i\geq 0$ and thus
$\Gamma\cap\lang'_\ga \subseteq T'_\oga$.

\medskip

\noindent For the second part, we show that for any stable
expansion $ T $ of $\Delta\combvar=\comb\combpar$, $\combvar
\,{\in}\,\{\mbox{}_{HP}, \mbox{}_{EB}, \mbox{}_{EH} \} $, there
exists a stable expansion $ T' $ of
$\Delta'\combvar=\comb(\fot,gr_{\lang}(\lp)) $ such that $
T'_{oga}{=}\,T_{oga}$, and vice versa. The result then follows
by Propositions~\ref{prop:hp-eh-ground}
and~\ref{prop:eb-eh-safe}.

In what follows, let ${\signature_\lp}$ be the signature of
\lp, let ${\signature}$ be the combined signature of \fot and
\lp, and let $\Delta^{(')}\combvar[\Sigma_\lp]$ (resp.,
$\Delta^{(')}\combvar[\signature]$) result from
$\Delta^{(')}\combvar$ by removing the $\UNA_{\Sigma_\lp}$
axioms and adding, for each $a,b\in \names_{\signature_\lp}$
(resp., $a,b\in \names_{\signature}$), the atom $a\,{=}\,b$ if
$a\,{=}\,b \in T$ and the negated atom $\neg(a\,{=}\,b)$ if
$a\,{=}\,b\notin T$.

For the case $ x=HP $, we note that $\Delta_{HP}[{\Sigma_\lp}]$
and $\Delta_{HP}[{\Sigma}]$ are classical FO Horn theories. Let
$\alpha \in \lang_\ga$. Given that $T$ is a stable expansion of
$\Delta_{HP}$, Proposition~\ref{prop:atoms-expansions} implies
$\alpha \in T_\oga$ iff (a) $\Delta_{HP} \anymodels_{T_\oga}
\alpha$; the latter is equivalent to (b)
$\Delta_{HP}[{\Sigma_\lp}] \anymodels_{T_\oga} \alpha$. By
well-known minimal model properties of Horn theories (cf.\
Lemma~\ref{lem:intersection}), (b) iff (c)
%$\alpha \in T_\oga$ iff
$\Delta_{HP}[{\Sigma}] \anymodels_{T_\oga} \alpha$;
furthermore, if we replace all embedded rules $\embhp(r)$,
$r\,{\in}\,\lp$, in $\Delta_{HP}[{\Sigma}]$ with their
groundings with respect to $\names_{\Sigma}$,
we obtain $ \Delta'_{HP}[\Sigma]$, and we that (c) iff (d) %$\alpha \in T_\oga$ iff
$\Delta'_{HP}[{\Sigma}] \anymodels_{T_\oga} \alpha$.  The
latter is clearly equivalent to $\Delta'_{HP}
\anymodels_{T_\oga} \alpha$, which implies that $\Delta'_{HP}$
has some stable expansion $T'$ such that $T_\oga = T'_\oga$.
The converse direction is argued similarly.

In the cases $\combvar=E B $ and $\combvar=E H $,
$\Delta^{(')}_\combvar[{\Sigma_\lp}]$ and
$\Delta^{(')}_\combvar[{\Sigma}]$ are not classical FO Horn
theories. However, we can turn them into such theories: for
each modal atom $\mop p(\vec{x})$ in a rule, we view $\mop p$
as a fresh classical predicate, and add, for each ground atomic
formula $\alpha$ with predicate $p$,
 the
formula $\mop \alpha$ to \fot if $\alpha \in T$ and we add the
negated atom $\neg\mop\alpha$ to \fot if $\alpha \notin T$. One
can then straightforwardly verify that, for any $ \alpha \in
\lang_{ga} $, the resulting theory entails $ \alpha $ iff the
original theory entails $ \alpha $ with respect to $T$, i.e.,
if $ \Psi$ is the original theory, $ \Psi \anymodels_T \alpha
$. The argument then proceeds analogously to the case $x=HP$.
\end{proof}

  \begin{proof}[of Proposition \ref{prop:eb-eh-safe}]
We show that, given an FO interpretation \inter and a
stable expansion $T$ of \comb\combpar, with $\combvar\in\{EB,EH\} $,
$\inter\anymodels_T\combeb\combpar$ iff $\inter\anymodels_T \combeh\combpar$.
Recall that
\combeb\combpar and \combeh\combpar
only differ in embedding of rules $r$ from \lp, which are mapped to
$\embeb(r) = \alpha\limpl h $
in \combeb\combpar and to
\begin{equation}
\label{eq:embeh} \embeh(r) = \alpha\limpl (h\land \mop
h)
\end{equation}
in \combeh\combpar, where $\alpha = \textstyle{\bigwedge} (b_i
\land \mop b_i) \land \textstyle{\bigwedge} \lnot \mop c_j$.
Consider first $\combvar = EB$.

\smallskip

\noindent
($ \Leftarrow $) Clearly, every model of \combeh\combpar is a model of
\combeb\combpar.

\smallskip

\noindent
($ \Rightarrow $) Suppose $\inter\anymodels_T\combeb\combpar$.
Consider an axiom (\ref{eq:embeh}) such that $(\inter,\varass)\anymodels_T \alpha$.
It
follows that $b_i\in T$ and $ c_j \not\in T$, and, since
$\alpha \limpl h$ is in \combeb\combpar and $T$ is a stable
expansion thereof (and consequently closed under first-order entailment),
that $ h \in T$.
This proves $\inter\anymodels_T\combeh\combpar$.

\smallskip

\noindent The case $\combvar= EH$, i.e., $T$ is a stable
expansion of \combeh\combpar, is analogous.
\end{proof}

\begin{proof}[of Proposition \ref{prop:hpdis-ebdis-ground}]
Follows from the proof of Proposition
\ref{prop:corr-ground-programs}.
\end{proof}

\begin{proof}[of Proposition \ref{prop:hp-eh-ground}]
We first prove the result for $\fot\in  \hornclass$, and then extend
the result to $\fot\in \ghornclass$. We use the following Lemma.

\begin{lemma}
\label{lem:intersection}
Let $J$ be a {\rm (}countable{\rm )} index set, let $\inter_j = \langle
U_j,\interfunj\rangle$, for $j \in J$, be models of an FO theory
\fot in which all individuals are named, and let
\interdef be the ``intersection'' of all $\inter_j$,
defined as follows:
\begin{itemize}
  \item $U = \{ [ t^{I_1}, \ldots, t^{I_j},\ldots] \mid
  t\in \names \}$ {\rm (}recall that $\names$ is a set of
  ground terms{\rm )};
  \item for $n$-ary function symbols $f$, $f^I(u_1,\ldots,u_n) =
   [ v_1,\ldots,v_j,\ldots]$ such that $v_j
   = f^{I_j}(u_{1,j},\ldots,u_{n,j})$, for all $j\in J$, and $u_i=[
   u_{i,1},\ldots,u_{i,j},\ldots]$, for $i=1,\ldots,n$;
 \item for $n$-ary predicate symbols $p$, $\langle u_1,\ldots,u_n\rangle \in p^I$ iff
$\langle u_{1,j},\ldots, u_{n,j}\rangle \in p^{I_j}$, for all $j \in J$, where
$u_i=[ u_{i,1},\ldots,u_{i,j},\ldots]$, for $i=1,\ldots,n$.
   \end{itemize}
Then, if  \fot is Horn, $\inter$ is a model of \fot.
\end{lemma}
This lemma generalizes the well-known intersection property
of Herbrand models of Horn theories.  Observe that for ground terms $t,t'\in \names$ we
have $\inter \models t = t'$ iff $\inter_j \models t = t'$ for every $j\in J$.

To show that
$\combhpdis\combpar\anyexpequiv_\oga \combehdis\combpar$,
it is by Proposition~\ref{prop:atoms-expansions} sufficient
to show that given any stable expansion $T$  of (1) $\combehdis\combpar$, resp.\ (2)
$\combhpdis\combpar $,
for each $\alpha\in\flang_\ga$ it holds that
$\combhpdis\combpar \anymodels_T \alpha$ iff $\combehdis\combpar \anymodels_T
\alpha$. We leave out the parameters \combpar~in the remainder.

The \emph{if}
direction always holds, as every model \ainter of \combhpdis~is also a model of $\combehdis$:
both theories
contain $\fot$ and the same \UNA axioms; by the \PIA axioms we have that ($*$) for any ground atom $ \alpha $ occurring
in \lp, $\inter\models\alpha $ iff $ \alpha \in T $ and, consequently, for every rule $ r $ (which must be
ground, as \lp is ground),
$\inter\anymodels_T \embhpdis(r)$ iff
$\inter \anymodels_T \embehdis(r)$.
Thus, it remains to consider the \emph{only-if} direction.

Case 1. Let $T$ be a stable expansion of $\combehdis$.  Suppose
$\combhpdis\anymodels_T \alpha$, but $\combehdis\nmodels_T \alpha$. Let $\{
\alpha_j \mid j\in J\} = \{ \alpha' \in \flang_\ga \mid \combehdis\nmodels_T
\alpha'\}$ and let for each $\alpha_j$, $\inter_j$ be an interpretation such
that $\inter_j\anymodels_T \combehdis \cup \{\lnot\alpha_j\}$; such a $\inter_j$
must exist and, without loss of generality, $\inter_j$ has only named individuals (as $\combehdis$ is
a universal theory). Note that $\alpha = \alpha_j$ for some index
$j\in J$.

Let $\inter^*$ be the ``intersection'' of all $\inter_j$, constructed
as in Lemma~\ref{lem:intersection}.
We first establish $\inter^* \anymodels_T \combehdis$: by the lemma,
$\inter^* \anymodels_T \fot$; $\inter^*$ does not satisfy equalities that are not
satisfied by the $\inter_j$, satisfaction of the \UNA axioms follows; finally,
$\langle\inter^*,T\rangle$ does not satisfy antecedents of any $\embehdis(r)$,
for $ r \in \lp $,
that are not satisfied by some $ \langle \inter_j, T \rangle $, and so
$ \inter^*\anymodels_T \embehdis(r)$.

We now establish $\inter^* \anymodels_T \combhpdis$.
As $T$ is a stable expansion of \combehdis, every
$\alpha'\in\flang_\ga$ such that $\alpha'\in T$
is satisfied in $\inter^*$.
The $\PIA$ axioms ($\alpha' \limpl
\mop\alpha'$) are satisfied: if $\alpha'\notin T$ (i.e., $\mop\alpha'$ is false),
then $\alpha'$ is false in $\inter^*$, as for each $\alpha'\notin T$, there is
a $\inter_j $ such that $ \inter_j \not\models \alpha' $.
Satisfaction of the rule embeddings $ \embehdis(r)$  follows by the same argument as
in the \emph{if} direction.
Therefore, $\inter^*\anymodels_T
\combhpdis$. However, as $\alpha=\alpha_j$ for some $j$, $\inter^*\not \models\alpha$.
This implies $\combhpdis\nmodels_T \alpha$, a contradiction.

Case 2. Let $T$ be a stable expansion of $\combhpdis$.
Suppose again that $\combhpdis\anymodels_T \alpha$, but $\combehdis\not\anymodels_T
\alpha$. Hence there is some model $\langle \inter',T\rangle$
of $\combehdis \cup \{\lnot \alpha\}$ such that $\inter'$ is named. As
$\inter'\nmodels_T\combhpdis$ must hold, the difference of $\combhpdis$ and $\combehdis$
(which is only in rule embeddings) together with groundness of \lp
implies that $\alpha$ must occur in some rule head of \lp.
Let
$\langle \inter'',T\rangle$ be a model of $\combhpdis$ such that $\inter''$
is named. Such a model must exist: if $\combhpdis$ would have no such model, then
$\flang_\ga\subseteq T$ would hold as $T$ is a stable expansion of
\combhpdis. As \lp is ground, we can replace all non-negated modal atoms in $\PIA$
axioms and in rule embeddings $\embehdis(r)$ with $\top$. Modulo simplifications,
we then obtain $\combhpdis$. But then $\inter' \anymodels_T \combhpdis$, which is a contradiction.

Let $\inter$ be the ``intersection'' of $\inter'$ and $\inter''$ as in
Lemma \ref{lem:intersection}. We claim that $\inter\anymodels_T \combhpdis
$: $\langle \inter,T\rangle$ satisfies \fot, the $\UNA$, and the \PIA axioms
by  arguments analogous to Case~1. Satisfaction of the rule embeddings
$\embhpdis(r)$ follows from Property~($*$).

\medskip

\noindent
For establishing
$\combhp\combpar = \combeh\combpar$ for $\lp\in\ngroundclass$, the argument
is similar but slightly diverging, as there are no
$\PIA$ axioms. Even if Property~($*$) does not
hold, in (2) some $\inter''$ exists such that $\langle \inter'',T\rangle$
satisfies $\bigwedge_i b_i\land \mop h_i$, resp.\ $h_{k'}\land \mop
h_{k'}$, in an embedded rule $\embeh(r)$ iff
$\inter''\models \bigwedge_i b_i$, resp.\ $\inter''\models h_{k'}$, as
$EH$ is, modulo modal atoms $\mop\alpha$---which can be eliminated from
it, see below---, an FO Horn theory, and so
Lemma~\ref{lem:intersection} applies.

This proves the result for $\fot\in \hornclass$. To generalize it to $\fot\in
\ghornclass$, we exploit standard skolemization~\cite{Fitting-FirsOrdeLogiAuto:96}, by which we obtain from $\fot$ an
equi-satisfiable FO Horn theory $\fot'$.  We show that skolemization of
$\fot$ commutes with its combination with the ground program \lp.

In detail, as already shown,
$\combehdis(\fot',\lp)\anyexpequiv_{\oga}\combhpdis(\fot',\lp)$
and  $\combeh(\fot',\lp)\anyexpequiv_{\oga}\combhp(\fot',\lp)$
hold. Thus it remains to show for each $\chi\in \{ \mbox{}_{HP}^\lor, \mbox{}_{EH}^\lor, \mbox{}_{HP}, \mbox{}_{EH}\}$
that for every stable expansion $T$ of
$\combop\chi\combpar\subseteq \fmodl$, there exists some stable
expansion $ T'$ of $\combop\chi(\fot',\lp)\subseteq \fmodl'$  such that
$T\cap\flang_{ga}=T'\cap\flang_{ga} $, and vice versa.

We transform $\combop\chi\combpar$ to an FO theory $J_\chi(T)$
as follows: we replace each modal atom
$\mop\alpha$ in  $\combop\chi\combpar$ with $\top$ if $\alpha\in T$, and with
$\bot$ if $\alpha\notin T$. Intuitively, in $J_\chi(T)$ we fix the value of
$\mop\alpha$ according to $T$. Since $\lp$ is ground, $J_\chi(T)$ is
then indeed an FO theory.
% We now have that for any objective ground atom
% $\alpha\in\flang_{ga}$,
Clearly, for every  $\alpha\in\flang_{ga}$, $\combop\chi\combpar\anymodels_T \alpha$
iff $J_\chi(T)\models\alpha $.

We construct $J'_\chi(T) $ from $\combop\chi(\fot',\lp)$ in the same way, and
obtain that for every $\alpha\in\flang_{ga}$, $\combop\chi(\fot',\lp)\anymodels_T
\alpha $ iff $J'_\chi(T)\models\alpha $. Now observe that $J'_\chi(T)$ is a
skolemization of $J_\chi(T)$, and thus equi-satisfiable; the same
holds for $ J_\chi(T) \cup \{\lnot \alpha \}$ and $
J'_\chi(T) \cup \{\lnot \alpha \}$, for arbitrary
$\alpha\in\flang_{ga}$, and thus
$ J_\chi(T) \models \alpha $ iff $ J'_\chi(T) \models \alpha $;
equivalently,
$$\combop\chi(\fot,\lp)\anymodels_{T_\oga} \alpha\text{ iff
}\combop\chi(\fot',\lp)\anymodels_{T_\oga}\alpha\text{.} $$
From Proposition \ref{prop:atoms-expansions}, we
can conclude that $ \combop\chi(\fot',\lp)$ has
some stable expansion
$T' $ such that $ T_{\oga}=T'\cap \lang_{ga}$.
(Note that $ T' $ might contain objective ground atoms that are
not in $ T_{\oga}$, due to the Skolem functions.) The converse
is obtained analogously.
\end{proof}

\begin{proof}[of Proposition \ref{prop:hpdis-ebdis-prop}]
Follows from the proof of Proposition
\ref{prop:corr-ground-formulas} and the fact that there is a $\PIA$
axiom for
each propositional symbol that occurs in any formula in
$\embhpdis(\lp)$ or $\embebdis(\lp)$.
\end{proof}

\begin{proof}[of Proposition \ref{prop:eb-eh-prop-merged}]
To show the first part, let $T$ be a stable expansion of \combeb\combpar. We show that for every $\phi \in \lang_{g}$,
$\combeb\combpar\anymodels_T \phi$ iff $\combeh\combpar\anymodels_T \phi$.
\smallskip

\noindent
($\Rightarrow$) Trivial, as every model of \combeh\combpar is a model of
\combeb\combpar (see also Proposition~\ref{prop:models-norm}).

\smallskip

\noindent
($\Leftarrow$) Let $\phi \in \lang_{g}$. Suppose $\combeh\combpar\anymodels_T\phi$, but
$\combeb\combpar\nmodels_T\phi$, and so there is a \interdef such that
$\inter\anymodels_T\combeb\combpar$ and $\inter\nmodels_T\phi$.
Without loss of generality, we may assume that \inter is named
by straightforward adaptation of Lemma~\ref{lem:named-model}; note
that individuals do not affect satisfaction of propositional formulas,
and unnamed individuals to not affect satisfaction of ground formulas.
By analogous argument to the proof of the \emph{if} part of (\ddag) in
Proposition~\ref{prop:corr-safe-formulas}, we can conclude
$\inter\anymodels_T\combeh\combpar$, a contradiction.
\smallskip

\noindent
The case of $T$ being a stable expansion of \combeh\combpar is analogous.
The proof of the second part is obtained from the proof of Proposition \ref{prop:corr-safe-formulas}
by replacing $ \emb(\lp) $ with $\comb(\fot,\lp) $ and the observation that
the domain of interpretation $ \domain $ has no bearing on the satisfaction
of propositional formulas.
\end{proof}

\begin{proof}[of Proposition \ref{prop:counterexamples}]\mbox{} \\
\noindent(\ref{c1}):\quad Consider the program $\lp=\{
r  \leftarrow p;
r  \leftarrow q;
s  \leftarrow s\} $
and the theory $ \fot=\{p \lor q, \lnot s \limpl t \} $.
The combinations $ \combop\chi\combpar $, with $ \chi \in
\{\mbox{}_{HP},\mbox{}_{EB},\mbox{}^\lor_{HP},\mbox{}^\lor_{EH} \} $, all have
one stable expansion $ T \chi $. Now $ T_{HP} $ contains $ r $ while $
T_{EB} $ does not, due to the modal operators in the antecedents in $ \embeb(\lp)$.
We have that $ s $ is not included in any $ T \chi $.  So, by the $\PIA$ axioms,
$
\lnot s $ and, consequently, $ t $ are included in $ T^\lor_{HP} $.  Neither is
included in $ T_{HP} $ or $ T^\lor_{EH} $.

\smallskip

\noindent(\ref{c2}):\quad Consider $ \lp=\{q \leftarrow p \} $.
The stable expansions of both $\combhp(\emptyset, \lp)$ and
$\combhpdis(\emptyset, \lp)$ contain $ p \limpl q $, whereas
the expansions of $\combeb(\emptyset, \lp)$ and
$\combehdis(\emptyset, \lp)$ do not; see also Example
\ref{ex:contrapositive}. In addition, $ \lnot p $ and $ \lnot q
$ are in the expansion of $\combhpdis(\emptyset, \lp)$, due to
the $\PIA$ axioms, but not in the expansions of
$\combhp(\emptyset, \lp)$ or $\combehdis(\emptyset, \lp)$.

\smallskip

\noindent(\ref{c3}):\quad Consider $ \lp=\{p(a);q(x)\leftarrow p(x)\}$. The single
expansion of $ \embhpdis(\lp) $ contains $\forall x \st p(x)\limpl q(x)$, whereas the expansion of
$ \embebdis(\lp) $ does not.

\smallskip

\noindent(\ref{c5}):\quad Consider $ \lp=\{ q'\leftarrow r(x), \dnot p(x);\,
s\leftarrow q;$ $s\leftarrow q';$ $r(a)\} $ and the theory $ \fot=\{p(a), p(b)\limpl q,r(b)\} $.  We have
\begin{align*}
\combhp\combpar = \{ & p(a), r(b), r(a), \\
                 & p(b) \limpl q,\; q  \limpl s,\; q' \limpl s, \\
                 & \forall x.r(x)\land\lnot \mop p(x) \limpl q'  \}.
\end{align*}
A given model \ainter of $\embhp(\lp)$ either satisfies $ a=b $
or $ a \neq b $---$\embhp(\lp)$ does not include UNA axioms
involving $b$. In the former case, $ p(b)$, and thus $ q $ and
$ s $, must be satisfied. In the latter case, assuming that $
p(a) \in T $ and $p(b) \notin T $, $ \mop p(x)$ is not
satisfied in case $ x $ is assigned to the individual denoted
by $ b $. Consequently, $ q' $ and $ s $ are satisfied.  One
can then verify that the single consistent stable expansion $ T
$ of \combhp\combpar is such that $ T_\oga=\{r(a), r(b), p(a),
s \} $.

The combination \combeb\combpar is like \combhp\combpar, except that
the axioms for $s$ are:
\begin{align*}
                 & \mop q \land q  \limpl s,\;  \mop q' \land q' \limpl s.
\end{align*}
The single stable expansion of \combeb\combpar does not contain $ s $, as neither
$ q $ nor $ q' $ can be derived. So, $q,q' \not \in T$, and thus $ s $ cannot
be derived. This disproves $\combhp\combpar\anyexpequiv_{\oga}\combeb\combpar$.

Consider the axiom corresponding to the first rule in \lp in
\combeh\combpar:
\begin{align*}
                 & \forall x \st\mop r(x) \land r(x)\land \lnot \mop p(x)  \limpl \mop q' \land q'.
\end{align*}
Let $ T' $ be a stable expansion of \combeh\combpar such that $ q',p(b) \notin T' $;
clearly, $ r(b), p(a) \in T' $.
Consequently, $\mop r(x) \land r(x)\land \lnot \mop p(x)$
is not satisfied in any model $ \langle \inter, T' \rangle $
for any variable assignment $ B $.
Consider $ B $ such that $ x^B=b^\inter $. Clearly, $ (\inter, B) \anymodels_{T'} \mop r(x)
\land r(x) $, and thus it must be the case that $ (\inter, B) \anymodels_{T'} \mop p(x) $.
Since $ p(b) \notin T' $, it must be the case that $\inter\models a=b$.  But then,
$\combeh\combpar\anymodels_{T'} a=b $, and so $ p(b) \in T' $, a contradiction.
As the expansions of \combeb\combpar and \combhp\combpar contain neither $ q' $ nor
$ p(b) $, this disproves $\combhp\combpar\anyexpequiv_{\oga}\combeh\combpar$ and
$\combeb\combpar\anyexpequiv_{\oga}\combeh\combpar$.

The combinations \combhpdis\combpar and \combebdis\combpar both have a stable expansion that contains $ q' $ and not $ q $ or $p(b)$: since $ q $ is not in the expansion,
$ \lnot q $, and also $ \lnot p(b)$, must be satisfied in every model, by
the $\PIA$ axioms.  Consequently, every model must satisfy $ a \neq b$ and thus also
$ q' $ and $ s $, which are hence included in the expansion. This
disproves $\combop\chi\combpar\anyexpequiv_{\oga}\combop\xi\combpar$  for
$ (\chi, \xi)$ = $(_{HP}, \mbox{}^\lor_{HP})$, $(_{HP}, \mbox{}^\lor_{EB})$,
$(_{EB}, \mbox{}^\lor_{EB})$, and
$(_{EB}, \mbox{}^\lor_{HP})$.

Consider the following modification of $ (\fot, \lp) $: in \fot replace $ q $ with $ q(b) $ and in \lp replace $ q $ and $ q' $ with $ q(x)$ and $ q'(x) $.
Observe that the embeddings $ \embebdis(\lp)$ and $ \embhpdis(\lp)$ do not contain $\PIA$ axioms involving $ q(b) $ or $ q'(b)$. By an argument analogous to the case of \combhp\ and \combeb, we obtain that $ s $ is included in the single stable expansion of \combhpdis\combpar, but not in the single expansion of \combebdis\combpar
and that neither expansion contains $q'(b)$ or $p(b)$.
This disproves $\combhpdis\combpar\anyexpequiv_{\oga}\combebdis\combpar$.

One can then argue, analogously to the case of the original combination above,
that there is no stable expansion $ T $ of \combeh\combpar that contains neither
$ q'(b) $ nor $ p(b) $. This disproves $\combeh\combpar\anyexpequiv_{\oga}\combebdis\combpar$
and $\combeh\combpar\anyexpequiv_{\oga}\combhpdis\combpar$.

\smallskip

\noindent(\ref{c8}):\quad Reconsider $\lp = \{p(a);\, p(x);\, q(x)
\leftarrow p(x)\}$ following Proposition \ref{prop:corr-safe-formulas}.
\end{proof}

\bibliography{references}

\begin{received}
Received Month Year; revised Month Year; accepted Month Year
\end{received}
\end{document}